%% file: paper.tex
\documentclass[sigconf,10pt]{acmart}

\usepackage{balance}

\def\BibTeX{{\rm B\kern-.05em{\sc i\kern-.025em b}\kern-.08emT\kern-.1667em\lower.7ex\hbox{E}\kern-.125emX}}
    
\usepackage{amsmath,amssymb,amsfonts,dsfont,multirow,multicol,epsfig,url,array,makecell,balance,color,epstopdf}
\usepackage[ruled, vlined, linesnumbered]{algorithm2e}
\usepackage{hhline}
\usepackage{array}
\usepackage{enumerate}
\usepackage{enumitem}
\usepackage{subfigure}
\usepackage{booktabs}

\input{symbol-def.tex}

\begin{document}
\copyrightyear{2019} 
\acmYear{2019} 
\setcopyright{acmcopyright}
\acmConference[SIGMOD '19]{2019 International Conference on Management of Data}{June 30-July 5, 2019}{Amsterdam, Netherlands}
\acmBooktitle{2019 International Conference on Management of Data (SIGMOD '19), June 30-July 5, 2019, Amsterdam, Netherlands}
\acmPrice{15.00}
\acmDOI{10.1145/3299869.3319873}
\acmISBN{978-1-4503-5643-5/19/06}

\ccsdesc[500]{Mathematics of computing~Graph algorithms}
\ccsdesc[500]{Information systems~Data mining}

\keywords{SimRank; Power-Law Graphs; Personalized PageRank}

\fancyhead{}
\title{\prsim: Sublinear Time SimRank Computation on Large Power-Law Graphs}
\subtitle{[Technical Report]}
 
\author{Zhewei Wei}
\email{zhewei@ruc.edu.cn}
\authornote{Work partly done at Beijing Key Laboratory of Big
    Data Management and Analysis Method, Renmin University of
    China.}
  \affiliation{%
  \institution{School of Information, DEKE MOE, Renmin University of China}
}

\author{Xiaodong He}
\email{hexiaodong\_1993@ruc.edu.cn}
\affiliation{%
  \institution{4Paradigm Inc.}
   \city{Beijing}
 \country{China}
}

\author{Xiaokui Xiao}
\email{xkxiao@nus.edu.sg}
\affiliation{%
  \institution{School of Computing,
National University of Singapore}
}

\author{Sibo Wang}
\email{swang@se.cuhk.edu.hk}
\affiliation{%
  \institution{The Chinese University of Hong Kong}
}

\author{Yu Liu}
\email{dokiliu@pku.edu.cn}
\affiliation{%
  \institution{Peking University}
}

\author{Xiaoyong Du}
\author{Ji-Rong Wen}
\authornote{Ji-Rong Wen is the corresponding author.}
\email{{duyong,jrwen}@ruc.edu.cn}
\affiliation{%
  \institution{Renmin University of China}
}

\input{abstract.tex}
\maketitle

\input{intro.tex}
\input{preliminaries.tex}

\input{single_source.tex}

\input{analysis.tex}

\input{related.tex}

\input{experiment.tex}
\input{conclusions.tex}
\input{acknowledgements.tex}
\allowdisplaybreaks

\begin{small}
\bibliographystyle{plain}
\bibliography{paper}
\end{small}

\appendix
\input{appendix.tex}

\end{document}

%% file: symbol-def.tex
\def\header{\vspace{2mm} \noindent}

\def\tblcapup{\vspace{0mm}}
\def\tblcapdown{\vspace{2mm}}

\newcommand{\pushright}[1]{\ifmeasuring@#1\else\omit\hfill$\displaystyle#1$\fi\ignorespaces}
\newcommand{\pushleft}[1]{\ifmeasuring@#1\else\omit$\displaystyle#1$\hfill\fi\ignorespaces}

\def\done{\hspace*{\fill} {$\square$}}

\def\e{\varepsilon}
\def\d{\bar{d}}
\def\s{\hat{s}}

\def\E{\mathrm{E}}
\def\Var{\mathrm{Var}}

\def\scw{\sqrt{c}}

\def\prsim{PRSim\xspace}

\def\inN{\mathcal{I}}
\def\outN{\mathcal{O}}
\def\din{d_{in}}

\def\W{\mathcal{W}}

\def\epi{\hat{\pi}}

\def\pib{\psi}

\def\rb{r}

\def\heta{\widehat{\eta\pi}}

\def\brmax{r_{max}}



%% file: abstract.tex
\begin{abstract}

{\it SimRank} is a classic measure of the similarities of nodes in a graph. Given a node $u$ in graph $G =(V, E)$, a {\em single-source SimRank query} returns the SimRank similarities $s(u, v)$ between node $u$ and each node $v \in V$. This type of queries has numerous applications in web search and social networks analysis, such as link prediction, web mining, and spam detection. Existing methods for single-source SimRank queries, however, incur query cost at least linear to the number of nodes $n$, which renders them inapplicable for real-time and interactive analysis.

{ This paper proposes \prsim, an algorithm that exploits the structure of graphs to efficiently answer single-source SimRank queries. \prsim uses an index of size $O(m)$, where $m$ is the number of edges in the graph,  and guarantees a query time that depends on the {\em reverse PageRank} distribution of the input graph. In particular, we prove that \prsim runs in sub-linear time  if the degree distribution of the input graph follows the power-law distribution, a property possessed by many real-world graphs. Based on the theoretical analysis, we show that the empirical query time of all existing SimRank algorithms also depends on the reverse PageRank distribution of the graph.} Finally, we present the first experimental study that evaluates the absolute errors of various SimRank algorithms on large graphs, and we show that \prsim outperforms the state of the art in terms of query time, accuracy, index size, and scalability.
\end{abstract}


%% file: intro.tex
\section{Introduction} \label{sec:intro}
Measuring similarities and proximities of nodes in the graph is a classic task in graph analytics.  Several link-based similarity measures have been proposed, including Personalized PageRank~ \cite{page1999pagerank}, Simfusion~\cite{xi2005simfusion},  P-rank~\cite{zhao2009p} and Panther~\cite{zhang2015panther}. Among them, {\em SimRank} \cite{JW02}, proposed by Jeh and Widom, is regarded as one of the most influential similarity measures, and  has been adopted in numerous applications such as web mining \cite{Jin11}, social network analysis \cite{NK07}, and spam detection \cite{SH11}.
Given a graph $G=(V,E)$, the SimRank similarity of nodes  $u$ and $v$, denoted as $s(u, v)$, is defined as
\begin{equation} \label{eqn:intro-simrank}
s(u, v) =
\begin{cases}
1, & \text{if $u = v$}\\
{\displaystyle \frac{c}{|\inN(u)| \cdot |\inN(v)|} \hspace{-1mm}\sum_{u' \in \inN(u)}{\sum_{v' \in \inN(v)}{\hspace{-2mm} s(u', v')}} }, & \text{otherwise}
\end{cases}
\end{equation}
where $\inN(u)$ denotes the set of in-neighbors of $u$, and $c \in (0,1)$ is a decay factor typically set to 0.6 or 0.8 \cite{JW02,LVGT10}. This formulation is based on two intuitive statements: (1) two objects are similar if they are referenced by similar objects, and (2) an object is most similar to itself.
Due to its recursive nature, SimRank computation is a non-trivial problem and has been extensively studied for more than a decade. Existing work mostly considers three types of SimRank queries: (1) {\em Single-pair} queries, which ask for the SimRank similarity between two given nodes  $u $ and  $v$; (2) {\em All-pair} queries, which ask for the SimRank similarity between any pair of nodes  $u $ and  $v$; (3) {\em Single-source} queries, which ask for the SimRank similarity between every node and $u$. All-pair queries require storing $O(n^2)$ node pairs, and thus is infeasible for large graphs. Meanwhile, single-source queries has become the focus of recent research~\cite{KMK14,MKK14,TX16,FRCS05,LeeLY12,LiFL15,SLX15,YuM15b,LiFL15,jiang2017reads,liu2017probesim}, due to its connections to recommendation applications. In this paper, we aim to answer {\em approximate} single-source SimRank queries, defined as follows:
\begin{definition}[Approximate Single-Source Queries] \label{def:prelim-single-source}
Given a node $u$ in a directed graph $G$ and an absolute error threshold $\e$, an approximate single-source SimRank query returns an estimated value $\s(u, v)$ for each node $v$ in $G$, such that
$$\left\lvert \s(u, v) - s(u, v) \right\rvert \le \e$$
holds for any $v$ with at least { $1 - \delta$} probability. \done
\end{definition}


\begin{table*} [!t]
\centering
\renewcommand{\arraystretch}{1.5}
\begin{small}
  \vspace{-1mm}
\caption{Comparison of single-source SimRank algorithms with $\boldsymbol{\e}$ additive error and $\boldsymbol{1 -  \delta}$ success probability.}\label{tbl:intro-compare}
\vspace{-3mm}
 \begin{tabular} {|c|c|cc|c|c|} \hline
   Algorithm&  Query Time &\multicolumn{2}{c|}{Query Time (Power-Law Graphs)} & Space Overhead& Preprocessing Time \\ \hline
          \multirow{3}{*}{\prsim}  &  \multirow{3}{*}{ $O\left({n\log {n\over \delta}  \over \e^2} \cdot
\sum_{w \in V} \pi(w)^2 \right)$} & \multicolumn{1}{l}{$O\left({ \log {n\over \delta} / \e^2}
             \right)$} & \multicolumn{1}{l|}{ for $ \gamma > 2$}  &  \multirow{3}{*}{$O\left(\min\{n/\e, m\}\right)$}
                                         & \multirow{3}{*}{$O\left({m/  \e}
                                           \right)$}  \\ 
      & & \multicolumn{1}{l}{$O\left({ \log {n \over \delta} \cdot \log n / \e^2}
             \right)$} & \multicolumn{1}{l|}{ for $\gamma  = 2$}&
                                         &  \\ 
     &  &\multicolumn{1}{l}{$O\left( \min \left\{ {n^{1\over \gamma}/ \e^{2-{1\over \gamma}}}
          ,  {n^{{2\over \gamma} -1}/ \e^2}\right\}
             \right)$} & \multicolumn{1}{l|}{ for $1 < \gamma < 2$}&
                                         &  \\ \hline
   TSF~\cite{SLX15} & \multicolumn{3}{c|}{$O\left({n\log {n\over \delta} / \e^2} \right)$} &  $O\left({n\log {n\over \delta} /
                                              \e^2} \right)$
                                         &$O\left({n\log {n\over \delta} / \e^2}
                                           \right)$ \\ \hline
 READS~\cite{jiang2017reads} & \multicolumn{3}{c|}{$O\left({n\log {n\over \delta} / \e^2} \right)$} &  $O\left({n\log {n\over \delta} /
                                              \e^2} \right)$
                                         &$O\left({n\log {n\over \delta} / \e^2}
                                           \right)$ \\ \hline
   ProbeSim~\cite{liu2017probesim} & \multicolumn{3}{c|}{$O\left({n\log {n\over \delta} / \e^2} \right)$} &  0 & 0 \\ \hline
   SLING~\cite{TX16} & \multicolumn{3}{c|}{$O\left({n/ \e} \right)$} &  $O\left({n/ \e} \right)$
                                         & $O\left({m/  \e}
                                           +{n\log {n\over \delta} / \e^2}
                                           \right)$  \\ \hline
 \end{tabular}
\end{small}
\vspace{-3mm}
\end{table*}

{ \vspace{-2mm}
\header{\bf Power-law graphs.}  It was experimentally observed that most real-world networks are scale-free and follow
power-law degree distribution. In particular, let $P_o(k)$ and $P_i(k)$
denote the fraction of nodes in the graph having out-degree and in-degree at least $k$, respectively. Then,  on a {\em power-law graph}, $P_o(k)$ and $P_i(k)$ satisfy that
$P_o(k) \sim k^{-\gamma}$ and $P_i(k) \sim k^{-\gamma'}$~\cite{BollobasBCR03},
where $\gamma$ and $\gamma'$ are the (cumulative) power-law exponents that
usually take values from $1$ to $2$. Recent work has demonstrated that by exploiting this fact, 
we can improve the asymptotic bounds for various graph algorithms such as triangle counting~\cite{brach2016algorithmic}, transitive closure~\cite{brach2016algorithmic}, perfect matching~\cite{brach2016algorithmic}, PageRank computation~\cite{lofgren2015personalized,wei2018topppr} and maximum independent set~\cite{liu2015towards}.
}

\vspace{-1mm}
\header{\bf Motivations.}
{ Since many graph algorithms can benefit from the structure of real-world graphs, a natural question is: Can we do the same for SimRank algorithms? On one hand, we are interested in designing a more efficient SimRank algorithm by exploiting the structure of the graphs, since existing work for SimRank computation \cite{KMK14,MKK14,TX16,FRCS05,LeeLY12,LiFL15,SLX15,YuM15b,LiFL15,jiang2017reads,liu2017probesim} has  missed this opportunity for optimization.}
{
  On the other hand,  we are also interested in analyzing how the graph structure affects the performance of existing SimRank algorithms. More precisely,
  it has been observed in previous work~\cite{zhangexperimental} that the performance of existing SimRank algorithms
  may vary dramatically on graphs with similar numbers of nodes and edges. A typical example is the {\em Twitter (TW)} and {\em IT-2004 (IT)} data sets, both of which have around 40 million nodes and 1 billion edges. However, as shown in~\cite{zhangexperimental}  and in our experiments, the query times of most SimRank algorithms are significantly smaller on {\em IT-2004} than on {\em Twitter}.  Based on this phenomenon, ~\cite{zhangexperimental} suggests that {\em Twitter (TW)}  is ``locally dense''  and {\em IT-2004 (IT)} is ``locally sparse''. However, it is still desirable to obtain a quantifiable measure that describes the hardness of each graph in terms of SimRank computation.} Finally, since obtaining ground truth for single-source SimRank queries requires $n^2$ space, which is infeasible for large graphs, most existing work only evaluate the accuracy of the algorithms on small graphs. The only exception is recent work~\cite{liu2017probesim}, which evaluates precision for approximate top-$k$ queries on graphs with billion edges using the idea of {\em pooling}. However, there is no prior experimental study that evaluates absolute error for single-source queries on large graphs.

\vspace{-1mm}\header{\bf Our contributions.}
This paper studies the approximate single-source SimRank queries, and makes the following contributions.

{
  (1) We propose \prsim, an algorithm that leverages the graph structure to efficiently answer approximate single-source SimRank queries.  The query time complexity of \prsim is related to the {\em reverse PageRank} of the input graph $G$, which is defined as the PageRank of the graph $G'$ constructed by reversing the direction of each edge in $G$. Let $\pi(w)$ denote reverse PageRank of node $w$, and $\sum_{w\in V} \pi(w)^2 $ denote the second moment of the reverse PageRanks.
The average expected query cost for \prsim on worst-case graphs is bounded by
$O\left({n\log {n\over \delta}  \over \e^2}\cdot \sum_{w\in V} \pi(w)^2 \right)$. By the fact that $\sum_{w\in V} \pi(w)^2  \le \left(\sum_{w\in V} \pi(w)\right)^2 = 1$,
  \prsim provides at least the same complexity as the random walk based algorithms (ProbeSim, TSF, and READS) do on worst-case graphs.  Furthermore, \prsim uses an index of size $O(m)$, which significantly improves the scalability of the algorithm.  See Table~\ref{tbl:intro-compare}  for the theoretical comparison between our algorithm and the state of the art.

  On the other hand, we show that on power-law graphs, the second moment $\sum_{w\in V} \pi(w)^2 $ is an asymptotic variable that is close to $0$, which means \prsim actually  achieves sub-linear query cost on real-world graphs. 
More precisely, 
  Let $\gamma$ denote the cumulative power-law exponent of the out-degree distribution.  We show that
  the average  expected query cost for \prsim on power-law graphs is bounded by:
  \begin{equation}
\label{eqn:query}
\hspace{-2mm}\E[Cost]= \left\{ \hspace{-2mm}
\begin{array}{ll}
O({1\over \e^2} \log {n\over \delta}), &\textrm{for } \gamma > 2; \\
  O({1\over \e^2} \log {n\over \delta} \cdot \log n) , & \textrm{for }\gamma = 2; \\
O\left(\min\left\{{n^{1\over \gamma}\over \e^{2-{1\over \gamma}}}, {n^{{2\over \gamma}-1}\over \e^2}
  \right\} \right) , & \textrm{for } 1 < \gamma < 2,
\end{array}\right.
\end{equation}
for  $ {1\over n^{\Omega(1)}}< \e < 1 $ and $\delta > {1\over n^{\Omega(1)}}$. 
   To understand this complexity, we first note that when $\gamma  \ge 2$, our bounds depend only on $\log n$, which is significantly better than the corresponding bound of any previous SimRank algorithms.  For $1< \gamma < 2$, since $\e > {1\over n}$, we have ${n^{1\over \gamma} \over \e^{2-{1\over \gamma}}} \le {n \over \e}$. This implies that \prsim also outperforms SLING on power-law graphs. To the best of our knowledge, this is the first sublinear algorithm for single-source SimRank queries on power-law graphs.

   (2)  To achieve the desired query cost in Table~\ref{tbl:intro-compare}, we design several novel techniques for computing SimRank and  {\em Personalized PageRank (PPR)} . First,  we propose an algorithm that estimates the {\em last meeting probabilities}~\cite{TX16} (see Section for definition) for ALL nodes in $O(\log {n\over \delta} / \e^2)$ time. This improves the $O(n\log {n\over \delta} / \e^2)$ bounds in~\cite{TX16} by an order of $O(n)$ and is the key to achieve sub-linearity. Second, we propose an index scheme which performs the {\em  backward search}~\cite{lofgren2015personalized} algorithm only on a number $j_0$ of {\em hub nodes}. The parameter $j_0$ enables us to manipulate the tradeoffs between index size and query time, which improves the scalability of our algorithm. 
Finally, we design {\em Variance Bounded Backward Walk},  an algorithm that estimates the  Personalized PageRank values to a given target node $w$  with additive error $\e$ in $O(n\pi(w) \log {n\over \delta}/\e^2)$ time, where $\pi(w)$ is the reverse PageRank of node $w$. Since the average value of $\pi(w)$ is $1/n$, this significantly improves the $O(n\log {n\over \delta}/\e^2)$ time complexity of the {\em Randomized Probe} algorithm~\cite{liu2017probesim}, and is the key to the relation between the time complexity and the reverse PageRank distribution. We also note that the Variance Bounded Backward Walk  algorithm actually improves the time complexity of state-of-the-art PPR algorithms to target nodes for dense graphs~\cite{wang2018efficient}, and may be of independent interest.

(3) Based on the time complexity of \prsim, we conduct experiments to confirm that the hardness of SimRank queries is indeed reversely related to the out-degree power-law exponent $\gamma$ of the graph. This observation provides a quantifiable measure for the concept of locally dense and locally sparse networks introduced in~\cite{zhangexperimental}. In particular, the out-degree distribution of {\em IT-2004}  is significantly more skewed than that of  {\em Twitter} (see Figure~\ref{fig:degree}), 
which explains the performance discrepancy of existing SimRank algorithms on these two datasets. }
 We also conduct a large set of experiments that evaluate \prsim  against the state of the art on benchmark data sets. In particular, our experiments include the first empirical study on the tradeoffs between absolute error and query cost for single-source SimRank algorithms on graphs with billions of edges. Our empirical study shows that \prsim outperforms the state of the art in terms of query time, accuracy, index size, and scalability.

\begin{figure}[!t]
\begin{small}
 \centering
 \vspace{0mm}
  \includegraphics[height=34mm]{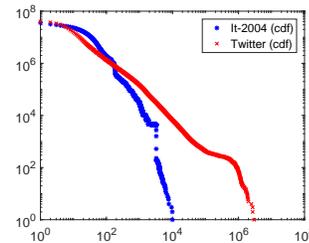}
\vspace{-3mm}
 \caption{Out-degree distributions of IT and TW.} \label{fig:degree}
\vspace{-4mm}
\end{small}
\end{figure}





%% file: preliminaries.tex
\begin{table} [t]
\centering
\renewcommand{\arraystretch}{1.3}
\begin{small}
  \tblcapup
  \vspace{0mm}
\caption{Table of notations.}\label{tbl:def-notation}
\vspace{-2mm}
 \begin{tabular} {|l|p{2.2in}|} \hline
   {\bf Notation}       &   {\bf Description}                                       \\ \hline
   $n, m$      &   the numbers of nodes and edges in $G$                            \\ \hline
   $\inN(v), \outN(v)$       &   the set of in-neighbors and
                               out-neighbors of a node $v$
   \\ \hline
   $d_{out}(v)$, $d_{in}(v)$  & the out-degree and in-degree of node $v$ \\
     \hline
   $s(u, v)$    &   the SimRank similarity of nodes $u$ and $v$    \\ \hline
   $\s(u, v)$    & an estimation of $s(u, v)$                                    \\ \hline
 $c$    &   the decay factor of SimRank          \\ \hline
   $\e$         &   the maximum absolute error allowed in SimRank computation          \\ \hline
    $\pi(w)$   & the reverse PageRank of node $w$\\
    \hline
    $\pi(u, w), \pi_\ell(u, w)$   & the RPPR and $\ell$-hop RPPR values of $w$ with respect to $u$\\
   \hline
      $\epi(u, w), \epi_\ell(u, w)$   & estimators of $\pi(u, w)$ and $\pi_\ell(u, w)$\\
    \hline
   $\rb_\ell(v,w)$, $\pib_\ell(v,w)$ & the {residue and reserve} of $v$
                                        at level $\ell$ from $w$ in the backward search \\
    \hline
 \end{tabular}
 \vspace{-4mm}
\end{small}
\end{table}

\vspace{0mm}
\section{Preliminaries} \label{sec:prelim}
Table~\ref{tbl:def-notation} shows the notations that are frequently
used in the remainder of the paper.

\vspace{-1mm}\header{\bf $\boldsymbol{\scw}$-walk and Reverse PageRank.}
We
unify the definition of SimRank and reverse PageRank under the
notation of $\scw$-walk. Let $G= (V, E)$ be a directed graph with $n$
nodes and $m$ edges. 
{  Given a source node $u \in V$ and a decay factor $c$, a {\em reverse
$\scw$-discounted random walk (or $\scw$-walk in short)} 
from $u$ is a traversal of $G$ that starts from $u$ and}, at each step,
either (i) terminates at the current node with $1-\scw$ probability,
or (ii) proceeds to a randomly selected in-neighbor of the current
node with $\scw$ probability. { We define
the {\em reverse PageRank} $\pi(w)$ of a node $w$ to be the
probability that an
$\scw$-walk from a uniformly chosen source node
terminates at $w$. It is easy to see that the reverse PageRank of a
node $w$ in the original graph $G$ equals to the PageRank of $w$ in
the reverse graph $G'$ constructed by reversing the direction of each edge in $G$.}

Given a source node $u$ and a target node $w$,  we further
define the
{\em reverse Personalized PageRank (RPPR)} $\pi(u,w)$ of $w$ with respect to $u$
to be the probability that an $\scw$-walk from $u$
terminates at $w$.  Again, the reverse Personalized PageRank on the
original graph $G$ equals to
the Personalized PageRank on the reverse graph $G'$. Since the RPPR values from a given source node $u$ form a probability
distribution,  we have $ \sum_{w \in V} \pi(u, w) = 1.$
Meanwhile, since the 
reverse PageRank $\pi(w)$ is equal to the probability that an
$\scw$-walk from a random source node
terminates at $w$, we have $\sum_{u \in V} \pi(u, w) = n\pi(w).$





\vspace{-1mm}\header{ \bf $\boldsymbol{\ell}$-Hop RPPR.} In this paper, we will
mainly use a variant of Personalized PageRank called {\em $\ell$-hop
  Reverse Personalized PageRank ($\ell$-hop RPPR)}. Given a source
node $u$, the $\ell$-hop RPPR $\pi_\ell(u, w)$ of node $w$ respected
to $u$ is the probability that a reverse $\scw$-walk from $u$ terminates at node $w$ with exactly $\ell$ steps. By the definition of $\ell$-hop RPPR, we have
\begin{equation}
  \label{eqn:l+1ppr}
\pi_{\ell+1}(y, w) =\sum_{x \in \inN(y)}{\sqrt{c}
  \over \din(y)} \pi_\ell(x, w).
\end{equation}
On the other hand, it is easy to see that RPPR $\pi(u, w)$ can be expressed as the
sum of  $\ell$-hop RPPR, that is, $\sum_{\ell=0}^\infty\pi_\ell(u, w) = \pi(u, w).$
Thus, we have $  \sum_{\ell = 0}^\infty\sum_{w \in V} \pi_\ell (u, w) = 1,$
and
\begin{equation}
  \label{eqn:l-PageRank}
  \sum_{\ell = 0}^\infty \sum_{u \in V} \pi_\ell (u, w) = n\pi(w).
\end{equation}


\vspace{-1mm}
\vspace{-1mm}\header{\bf SimRank, $\boldsymbol{\scw}$-walk, and
  hitting probability.}
It is
shown in~\cite{TX16} that the SimRank similarity $s(u,v)$ between two
different nodes $u$ and $v$
can also be formulated using $\scw$-walks.
Given two distinct nodes $u$ and $v$, we start a $\scw$-walk
from each node. If the two $\scw$-walks visit the same node after
exactly $i$ steps, we say the two $\scw$-walks meet at step $i$.
 \cite{TX16}  shows that $s(u,v)$
is equal to the probability that the two $\scw$-walks meet.  

Moreover, \cite{TX16} proposes SLING, an algorithm that uses the
following formula to estimate SimRank values:
\begin{equation}
  \label{eqn:sling}
s(u,v)=\sum_{\ell=0}^{\infty}\sum_{w\in V}
h_\ell(u,w)h_\ell(v, w) \eta(w).
\end{equation}
Here $h_\ell(u, w)$ denote the {\em hitting probability} that an
$\scw$-walk from node $u$ {\em visits} $w$ in its $\ell$-step, and
$\eta(w)$ is a parameter that characterizes the last-meeting probability:

\begin{definition}[Last-meeting probability]
The last-meeting probability $\eta(w)$ for node $w$ is the probability
that two $\scw$-walk from $w$ do not meet at $i$ step for any $i\ge
1$. 
\end{definition}
SLING precomputes $h_\ell(u,w) $ and $ \eta(w)$ with an
additive error up to $\e$, and stores them in the index. Given a query
node $u$, it retrieves all levels $\ell$ and nodes $w$ such that
$h_\ell(u,w) > \e$. For each $(\ell, w)$ pair, SLING retrieves all
nodes $v$ with $h_\ell(v,w) > \e$ and $\eta(w)$, and estimates $s(u,v)$
with Equation~\eqref{eqn:sling}.

There are two major issues with SLING. First, storing all
$h_\ell(u,w) $ with additive error up to $\e$ takes $O(n/\e)$ space,
which can be significantly larger than the graph size for reasonable
choices of $\e$. Second, approximating $\eta(w)$ for each $w \in V$
requires sampling a large number of random walks from each node in the
graph, which makes the preprocessing time infeasible on very large graphs. Our algorithm
overcomes these two drawbacks by (1) providing an index size that is at
most the size of the graph, and (2) designing an algorithm that
estimates $\eta(w)$ on-the-fly, using only $O(\log n / \e^2)$ time.




%% file: single_source.tex
\vspace{-2mm}
\section{\prsim algorithm}
In this section, we present \prsim, { an index-based
  algorithm that exploits the graph structure to efficiently
answer approximate single-source SimRank queries.}  We first provide the estimating formula that relates
SimRank and $\ell$-hop RPPR.

\vspace{-2mm}\subsection{SimRank and $\ell$-hop RPPR} \label{sec:simrank_ppr}
The relation between SimRank and reverse Personalized PageRank can be
directly derived from equation~\eqref{eqn:sling}. Observe the
fact that $\ell$-hop RPPR $\pi_\ell(u, w)$ equals to  the hitting
probability $h_\ell(u,w)$ multiplied the the termination probability
$\alpha = 1-\sqrt{c}$, and we have
\begin{equation}
    \label{eqn:formula}
s(u,v)={1\over (1-\sqrt{c})^2} \sum_{\ell=0}^{\infty}\sum_{w\in V}
\pi_\ell(u,w)\pi_\ell(v, w) \eta(w).
\end{equation}
There are two reasons for using $\ell$-hop RPPR over hitting
probability. Firstly, we have $\sum_{\ell=0}^\infty \sum_{w\in V} \pi_\ell
(u,w) =1$. As we will show later, this is critical for
estimating $\eta(w)$ in { $O(\log {n\over \delta} /\e^2)$
  time. } Secondly, we have $\sum_{\ell=0}^\infty \sum_{u\in V} \pi_\ell
(u, w) = n \pi(w)$. This property relates SimRank with the reverse
PageRank, and thus is essential for achieving sublinear query time.


{
Recall that given a source node $u$, our goal is to estimate SimRank values
$s(u,v)$ with additive error $\e$ for any node $v\in V$. By
Equation~\eqref{eqn:formula}, 
we can decompose the query process into three
subroutines: 1) Given a source node $u$,
compute the $\ell$-hop RPPR values $\pi_\ell(u,w)$ for any nodes $w\in
V$; 2) Compute last meeting probabilities $\eta(w)$ for each
$w \in V$; 3) For any node $v\in V$, compute  $\ell$-hop RPPR values $\pi_\ell(v,w)$ to any target
node $w$.  For the first task, we can employ a simple Monte Carlo
algorithm which generates  a number $n_r = O(\log {n \over
\delta} /\e^2)$ of  $\scw$-walks from $u$ and uses the proportion of
$\scw$-walks that terminate at $w$ with exact $\ell $ steps to
approximate $\pi_\ell(u, w)$. This algorithm runs
in $O(\log {n \over
\delta} /\e^2)$  time, so we will focus on the remaining two tasks.  }


\vspace{-2mm}\subsection{ Computing Last Meeting Probability}
The first challenge is how to estimate $\eta(w)$ for each $w \in V$
efficiently. SLING~\cite{TX16} generates $n_r = \Theta\left({\log
    {n\over \delta} \over
    \e^2}\right)$ pair of $\scw$-walks for each $w\in V$, and
obtains an approximation to $\eta(w)$ with error $\e$ for each $w\in
V$. However, this solution leads to a preprocessing time of $O\left(
  {n\log {n\over \delta} \over \e^2}\right)$, and thus, { is not feasible if we need
small error $\e$ on large
graphs. }

{ Our first key insight is that, instead of estimating the
$\ell$-hop PPR
$\pi_\ell(u,w) $ and last meeting probability $\eta(w)$ separately, we can
estimate their product $\eta(w) \pi_\ell(u,w) $ in
the query phase, using only  $n_r = \Theta\left({\log {n \over \delta} \over
    \e^2}\right)$ samples.} More precisely, we observe that $\eta(w)
\pi_\ell(u,w) $ 
is the probability that an $\scw$-walk from  $u$ terminates at $w$
with $\ell$ steps, and then, two independent $\scw$-walks from $w$
do not meet.  Therefore, we can generate an $\scw$-walk $\W(u)$ from
$u$, and then two $\scw$-walks $\W_1(w)$ and $\W_2(w)$ from the node
$w$ where $\W(u)$ terminates. If $\W_1(w)$ and $\W_2(w)$ do not meet,
we set the estimator $\heta_\ell(u, w) =1$. This way we obtain an
unbiased estimator for each $\eta(w) \pi_\ell(u,w) $, $w\in V$ and $\ell =
0, \ldots, \infty$. {  We also note that the summation $\sum_{w\in V}
\sum_{\ell=0}^\infty \eta(w) \pi_\ell(u,w)  \le \sum_{w\in V}
\sum_{\ell=0}^\infty \pi_\ell(u,w) =1$, which means  
we can use Chernoff bound~\ref{lmm:chernoff} to estimates
$\eta(w) \pi_\ell(u,w)$  with additive error $\e$ for any $w \in V,
\ell \ge 0$ with only $n_r = \Theta\left({\log {n \over \delta}
    \over\e^2}\right)$ samples. }

\vspace{-2mm}
\subsection{ Precomputing RPPR to Hub Nodes}
{ Given a target node $w$, computing $\ell$-hop RPPR
$\pi_\ell(v, w)$ for any node $v \in V$ is time-consuming, especially
when $w$ is a hub node with many
out-neighbors. Therefore, we will use index to help
reduce the cost. SLING~\cite{TX16}  proposes the following approach: for each (source) node
$v$, we precompute
$\pi_\ell(v, w)$ for any $w \in V$ and put $\pi_\ell(v, w)$ into
an inverted list, so we can efficiently track $\pi_\ell(v, w), v\in V$
for a given target node $w$. This approach, however, essentially
builds an index for every target node $w \in V$ and results in an
index of size $O\left({n \over\e}\right)$, which is usually significantly
larger than the graph size $m$ for reasonably small $\e$. }

\vspace{-2mm}
\begin{algorithm}[h]
\begin{small}
\caption{Preprocessing Algorithm \label{alg:indexing}}
\BlankLine
\KwIn{Graph $G$,  decay factor $c$, error parameter $\e$ }
\KwOut{Lists $L_\ell(w)$ consisting of tuples $(v, \pib_\ell(v, w))$ for each $w$ with top-$j_0$ reverse PageRank values and
  $\ell = 0, \ldots, \infty$}
  { Construct a tuple $(x,y,\din(y))$ for each edge
    $(x,y)\in E$\;}
   { Use counting sort to sort the $(x,y, \din(y))$ tuples according the ascending order of
     $\din(y)$.\;}
 \For{  {   each $(x,y, \din(y))$} }
   {
       {  Append $y$ to the end of $x$'s out-adjacency list\;}
   }

  Calculate reverse PageRank $\pi(w)$ for $w \in V$\;
\For{each node $w$ with top-$j_0$ reverse PageRank values}
{

  $\rb_\ell(v, w), \pib_\ell(v,w) \gets 0$ for $\ell =0,
\ldots, \infty, v \in V$\;
  $\rb_0(w, w) \gets 1$,   $c_1 \gets {12 \over (1-\scw)^2}$, $\brmax \gets {\e \over c_1}$\;
  \For{$\ell$ from $0$ to $\infty$}
  {
    \For{each $v \in V$ with $\rb_\ell(v,w)> \brmax$}
    {
      \For {each $z \in \outN(v)$}
      {
        $\rb_{\ell+1}(z, w) \gets \rb_{\ell+1}(z, w) + \scw\cdot \frac{ \rb_\ell(v, w)}{\din(z)}$
      }
      $\pib_\ell(v,w) \gets \pib_\ell(v,w) + (1-\sqrt{c})  \cdot
      \rb_\ell(v, w)$\;
      $\rb_\ell(v, w) \gets 0$\;
    }
    \For { each $v$ with { reserve $\pib_\ell(v, w)  > \brmax$}}{
      Append tuple $(v, \pib_\ell(v, w) )$  to $L_\ell(w)$\;
     }
}
}
\end{small}
\end{algorithm}

\vspace{-2mm}
{ To reduce the index size, we propose to build index
only for {\em hub nodes}. In particular, we identify  $j_0$ nodes
with the largest reverse PageRanks as hub nodes, where $j_0$  is a
user-specified parameter. We then perform the {\em backward
  search}~\cite{lofgren2015personalized} algorithm on each hub
node $w$ to precompute $\pi_\ell(v, w)$ for any $v \in V$ and any $\ell
>0$. The definition of hub nodes is based on two
intuitions. }First, recall that the reverse PageRank of node $w$ is the
probability that an $\scw$-walk from a random node $u$ terminates at
$w$. Therefore, a hub node $w$ is more likely
to be visited in a single-source SimRank query on $u$. Second, since
$\sum_{\ell=0}^\infty \sum_{v \in V} \pi_\ell(v, w) = n\pi(w)$, a hub  node
will also have more $(\ell,
w)$-tuples with $\pi_\ell(v, w) > \e$, which makes it more difficult
to compute $\pi_\ell(v, w)$ on the fly. Therefore, pre-computing
$\pi_\ell(v,w)$ for nodes $w$ with highest reverse PageRank reduces
the query cost most efficiently. We also note that we can choose the
value of $j_0$ to balance the query time, index size and preprocessing
time. For ease of presentation, we select $j_0$ such
that the index size is bounded by $O(m)$ in this section.


Algorithm~\ref{alg:indexing} illustrates the pseudocode for the
preprocessing algorithm. {  For reasons we shall see later, for each node $u$ with
out-neighbor set
$\outN(x) =\{y_1, \ldots, y_{d}\}$, we store the adjacency list of $x$
in a way such that $\din(y_1) \le \ldots \le \din(y_d)$. To sort the
adjacency list of each node in total $O(m)$ time,  we first construct a
tuple $(x,y, \din(y))$ for each edge $(x,y) \in E$. Then we employ
the {\em counting sort} algorithm to sort the
$m$ tuples $(x,y, \din(y))$ according to the ascending order of
$\din(y)$. Since $\din(y)$ is an integer in range $[0,n]$,
the counting sort algorithm runs in time $O(m+n)$. We then scan the $m$ sorted
tuples and, for each tuple $(x,y, \din(y))$, we append $y$ to the end
of $x$'s out-adjacency list. This algorithm sorts the out-adjacency
list of each node in $O(m+n)$ time. (Lines 1-4). 
We then
calculate the reverse PageRanks for each node $w\in V$, and retrieve
the $j_0$ nodes with the largest reverse PageRank as the hub nodes (line 5). For each hub
node $w$,
we use backward search~\cite{lofgren2015personalized} to compute an estimator $\pib_\ell(v, w)$ for
the $l$-hop RPPR $\pi_\ell(v, w)$, for each $\ell = 0, \ldots,
\infty$ and $v\in V$.} More precisely, we first set {\em residue}
$\rb_\ell(v, w)$ and a {\em reserve} $\pib_\ell(v, w) = 0$ to each
node $v$ and $\ell=0,\ldots, \infty$. Then, we set $\rb_0(w, w) = 1$
and the residue threshold $\brmax =  {(1-\scw)^2\e \over 12}$ (Lines
6-8). {Note that we choose
the constant 
$(1-\scw)^2$  to compensate the denominator $(1-\scw)^2$ in
equation~\eqref{eqn:formula}, and the constant $12$ so that we can sum
various errors up to at most  $\e$.} Starting from level $0$, we traverse from $w$, following the
out-going edges of each node (Line 9). On visiting a node $v$ at level $\ell$,
we check if $v$'s residue $\rb_\ell(v, t)$ is larger than the
threshold $\brmax$. If so, for each out-neighbor $z$ of $v$,  we
increase the residue $\rb_{\ell+ 1} (z, w)$ of $z$ at level
$\ell + 1$ by $\scw \cdot \frac{\rb_\ell(v,   w)}{d_{in}(z)}$ (Lines
10-12). Next,  we increase $\pib_\ell(v, w)$, $v$'s backward reserve at
level $\ell$ by $\scw \rb_\ell(v, w)$ (line 13). After that, we reset $v$'s
backward residue $\rb_\ell(v, w)$ to $0$ (line 14). After all nodes $v$ with
residue
$\rb_\ell(v, w)> \brmax$ are processed,  we append tuples $(v, \pib_\ell(v, w))$ to
a list $L_\ell(w)$ for each $v$ with reserve
$\pib_\ell(v, w)> \brmax$ (line 15-17). Note that for each a node $w$ and a level $\ell$
with at least one $\pib_\ell(v, w) > \brmax$, we store all tuples $(v,
\pib_\ell(v, w))$ with $\pib_\ell(v, w) > \e$ in a list $L_\ell(w)$,
so we can quickly retrieve them given $w$ and $\ell$ in the
query phase.  { The following lemma can be directly derived from~\cite{lofgren2015personalized}

\begin{lemma}[\cite{lofgren2015personalized}]
  \label{lem:backward_search}
For any hub node $w$, any $v \in V$ and $\ell \ge 0$, Algorithm\ref{alg:indexing} ensures $|\pib_\ell(v, w) - \pi_\ell(v, w)| <
\brmax = {(1-\scw)^2\e \over 12}$.
\end{lemma}
}
We have the following lemma that bounds the
space usage and running time of Algorithm~\ref{alg:indexing}
{ on
worst-case graphs}.

\vspace{-1mm}\begin{lemma}
  \label{lem:index_size}
The size of the index generated by Algorithm~\ref{alg:indexing}  is bounded by $O\left( {n \over \e}\sum_{j=1}^{j_0} \pi(w_j)
  \right)$. The preprocessing time is bounded by $O\left(
  m \over \e\right)$.
\end{lemma}
{  We set $j_0$ so that $O\left( {n \over \e}\sum_{j=1}^{j_0} \pi(w_j)
  \right) =O(m)$ in the theoretical analysis of \prsim, for ease of presentation. Note that if the largest reverse PageRank $\pi(w_1)$
  satisfies $\pi(w_1) > \e m/n$, we need to set $j_0 = 0$,  in which
  case \prsim becomes an index-free
  algorithm.  However, in practice, we can manipulate $j_0$ to get a
  tradeoff between the index size and query cost.}

\vspace{-2mm}\subsection{ Sampling RPPR to  Non-Hub Nodes}


{ The  third key  component of our method is a sampling-based algorithm that efficiently
computes $\ell$-Hop PPR values to  non-hub target nodes (i.e., nodes
with small reverse PPR values and thus are not in the index). 
Given a node $w$, the goal is to provide an unbiased estimator $\epi_\ell(v, w)$ for $\pi_\ell(v,w)$ for each $v
\in V$ and any $\ell \ge 0$. Once we obtain such a sampler, we can
estimate  each $\pi_\ell(v, w)$ with additive error $\e$ using $\log
{n \over \delta} /\e^2$ samples. 
~\cite{liu2017probesim} provides such a sampler by employing a   {\em
  Randomized Probe} algorithm, which runs in $O(n)$
time for a single sample. This time complexity, however, is unacceptable if we want
sub-linear query time.}

In this section, we propose an algorithm that
achieves the following goals: 1) Given a node $w$, the algorithm provides an unbiased estimator $\epi_\ell(v, w)$ for $\pi_\ell(v,w)$, for each $v
\in V$ and any $\ell \ge 0$; 2)  the algorithm runs in $O(n\pi(w))$ expected time. Note that
$n\pi(w) = \sum_{i=0}^\infty \sum_{v\in   V}\pi_i(v,w)$ is the
expected output size and consequently the minimum cost for generating
unbiased estimators $\epi_i(v, w)$ for $i=0, \ldots, \infty$, $v\in
V$. (3) The variance of $\epi_i(v, w)$ is bounded, so we can use Chebyshev's inequality to bound the error, and the Median
Trick to boost the success probability. 

\vspace{-1mm}
\begin{algorithm}[h]
\caption{Backward Walk\label{alg:bw}}
\KwIn{Directed graph $G=(V,E)$; node $w \in V$; level $\ell$\\}
\KwOut{$\epi_\ell(v, w)$ for each $v\in V$\\}
  $\epi_\ell(v, w) \gets 0$ for $\ell=0,\ldots, \infty$, $x\in V$\;
  $\epi_0(w,w) \gets 1- \scw$\;
  \For{$i=0$ to $\ell-1$}
  {
    \For{each $x \in V$ with non-zero $\epi_i(x,w)$ }
    {

        $r \gets rand(0,1)$\;
        \For{each $y \in \outN(x)$ and $\din(y)\le {\scw \over r}$}
        {
          $\epi_{i+1}(y, w) \gets \epi_{i+1}(y,w) + {\epi_i(x, w)}$;
        }

    }
  }
  \Return all non-zero $\epi_\ell(v, w)$\;
\end{algorithm}

\vspace{-2mm}
\header{\bf  Simple Backward Walk with Unbounded Variance.}
For ease of exposition, we first present a simple {\em Backward Walk} that achieves the first two goals. The pseudocode is
illustrated by Algorithm~\ref{alg:bw}. Given a node $w$ and a level
$\ell$, this algorithm also gives an unbiased estimator $\epi_\ell(v,
w)$ for each $v\in V$. We first initialize $\epi_0(w,w) =1 - \sqrt{c}$
and  $\epi_\ell(x, w) =0$ for other $\ell$ or $x\in V$ (Lines
1-2). Then, we iterate $i$ from $0$ to $\ell -1$ (Line 3). At level
$i$, for each $x \in V$ with non-zero $\epi_i(x,w)$, we generate a
random number $r$ from $(0, 1)$ (Line 4-5), and scan the out-neighbors
of $x$ until we encounter the first node $y$ with $\din(y) > {\scw
  \over r}$. Recall that in the preprocessing phase, we sort the out
adjacency list of $x$ so that nodes in $\outN(x)$ are ordered according
to their in-degrees (see Algorithm~\ref{alg:indexing}). Therefore, we
only have to visit the nodes with  $\din(y) \le {\scw \over r}$, which
is a subset of $\outN(x)$. For each out-neighbor $y$ of $x$ with $\din(y) \le {\scw
  \over r}$, we add $\epi_i(x,w)$ to $\epi_{i+1}(y,w)$ (Lines
6-7). Finally, after level $\ell-1$ is processed, we return each
non-zero $\epi_\ell(v, w)$ as the estimator for $\pi_\ell(v, w)$ (Line
8).

We can use a simple induction to prove the unbiasedness of
Algorithm~\ref{alg:bw}. For the base case, we have  $\E[\epi_0(w,w)]
=1-\scw= \pi_0(w, w)$. Assume that $\E[\epi_{i}(x,w)]  = \pi_i(x, w)$
for any $x\in V$. For a node $y$ at level $i+1$, each $\epi_{i}(x,w),
x\in \inN(y)$ is added to $\epi_{i+1}(y,w)$ with probability ${\scw
  \over  \din(y)}$, and thus $\E[\epi_{i+1}(y,w)] = \sum_{x \in \inN(y)}{\scw \over
                      \din(y)}\E[\epi_{i}(x,w)] $. Therefore, we have
 $\E[\epi_{i+1}(y,w)] = \sum_{x \in \inN(y)}{\scw \over
  \din(y)}\pi_{i}(x,w) = \pi_{i+1}(y, w).$
To analyze the running time, note that the cost for computing
$\epi_{i}(x,w) $ is bounded by the number of times that $\epi_{i}(x,w)
$ is incremented. Since each increment adds at least $(1-\scw)$ to
$\epi_{i}(x,w) $, this cost is bounded  by ${\epi_{i}(x,w) \over
  1-\scw }$. Summing over $i=0, \ldots, \infty$ and $x\in V$, and
using equation~\eqref{eqn:l-PageRank}, the total cost is at most $O(n\pi(w))$.

Unfortunately, the estimator $\epi_\ell(v, w)$ returned by
Algorithm~\ref{alg:bw} can be unbounded, since we may sum up all
estimators from level $i$ to form an estimator of level $i+1$. To make
thing worse, it is even unclear if $\epi_\ell(v, w)$ has bounded variance. This
means that $\epi_\ell(v, w)$ may not be sub-gaussian or
sub-exponential, and thus we are unable to apply concentration
inequality to bound the error.

\vspace{-2mm}
\begin{algorithm}[h]
\caption{Variance Bounded Backward Walk\label{alg:vbbw}}
\KwIn{Directed graph $G=(V,E)$; node $w \in V$; target level $\ell$\\}
\KwOut{$\epi_\ell(v, w)$ for each $v\in V$\\}
$\epi_\ell(v, w) \gets 0$ for $\ell=0,\ldots, \infty$, $x\in V$\;
$\epi_0(w,w) \gets 1-\scw$\;
\For{$i=0$ to $\ell-1$}
{
\For{each $x \in V$ with non-zero $\epi_i(x,w)$ }
{
\If{$r_0 \gets rand() < \scw$}{
\For{each $y \in \outN(x)$ and $\din(y) \le {\epi_i(x, w) \over 1 - \scw}$}
{
$\epi_{i+1}(y, w) \gets \epi_{i+1}(y,w) + {\epi_i(x, w) \over \din(y)}$;
}
$r \gets rand(0,1)$\;
\For{each $y \in \outN(x)$ and $ {\epi_i(x, w) \over 1 - \scw} < \din(y)\le  {\epi_i(x, w) \over r(1 - \scw)}$}
{
$\epi_{i+1}(y, w) \gets \epi_{i+1}(y,w) + 1-\scw$;
}
}
}
}

\Return all non-zero $\epi_\ell(v, w)$\;
\end{algorithm}

\vspace{0mm}
\header{\bf   Variance Bounded Backward Walk.}
To overcome the drawback of simple Backward Walk, we propose the {\em Variance Bounded
  Backward Walk} algorithm, which achieves bounded variance without
sacrificing the $O(n \pi(w))$ query bound or the unbiasedness
guarantee. Algorithm~\ref{alg:vbbw} illustrates the pseudocode of the Variance Bounded
Backward Walk algorithm. We set $\epi_0(w,w) =1 -
\sqrt{c}$ and  $\epi_\ell(x, w) =0$ for other $\ell$ or $x\in V$
(Lines 1-2). Then
we iterate $i$ from $0$ to $\ell -1$ (Line 3). At level $i$, for  each $x \in
V$ with non-zero $\epi_i(x,w)$, we first generate a random number
$r_0$ so that we can stop the process at $x$ with probability
$1-\scw$ (Lines 4-5). With probability $\scw$,  we first scan through the
out-neighbors of $x$ until we encounter the first node $y$ with
$\din(y) > {\epi_i(x, w) \over 1-\scw}$. For each out-neighbor $y$
with $\din(y) \le  {\epi_i(x, w) \over 1-\scw}$ we increase $\epi_i(y,
w)$ by ${\epi_i(x, w) \over \din(y)}$ (Lines 6-7). Then, we choose a random number
$r$ from $(0, 1)$ (Line 8), and continue to scan the out-neighbors of $x$ until
we encounter the first node $y$ with $\din(y) > {\epi_i(x, w) \over
  r(1-\scw)}$. Again, we only visit a subset of $\outN(x)$, as the
nodes in $\outN(x)$ are ordered according to their in-degrees.
For each out-neighbor $y$ of $x$ with $\din(y) \le
{\epi_i(x, w) \over   r(1-\scw)}$, we increment $\epi_{i+1}(y,w)$ by
$1-\scw$ (Lines 9-10). After $\ell$ levels are processed, we return all non-zero
$\epi_\ell(v, w)$ as estimators for $\pi_\ell(v, w)$ (Line 11).


\vspace{-1mm}
\header{\bf { Analysis}. }
We prove three properties of the Variance Bounded Backward Walk
algorithm. First,  the algorithm gives an unbiased estimator
$\epi_\ell(v, w)$ for $\pi_i(v,w)$ for each $v \in V$ and $i \le
\ell$. In particular, we have the following lemma.
\vspace{-1mm}\begin{lemma}
  \label{lem:vbbw_ubias}
 Consider a node $v$ on a target level $\ell$, and let $\epi_\ell(v,
 w)$ be an estimator provided by Algorithm~\ref{alg:vbbw}. We have $\E[\epi_\ell(v, w)]=\pi_\ell(v,w)$.
 \end{lemma}

Next, we show  that the running time of Algorithm~\ref{alg:vbbw} on
node $w$ is proportional to its reverse PageRank $\pi(w)$. In particular,
we have the following lemma.
\vspace{-1mm}\begin{lemma}
  \label{lem:vbbw_query} The complexity of Algorithm~\ref{alg:vbbw}
  on node $w$, regardless of the target level $\ell$,  is bounded by $O(n\pi(w))$.
\end{lemma}
Note that $n\pi(w) = \sum_{i=0}^\infty \sum_{v\in V}\pi_i(v,w)$,
which implies that the minimum number of operations to return a
unbiased estimator  $\epi_i(v, w)$ for each $\pi_i(v,w)$ is
$\Omega(n\pi(w))$. This essentially means that
Algorithm~\ref{alg:vbbw} achieves optimal sampling complexity for this
task.

Finally,
we note that although the estimator $\epi_\ell(v, w)$  is unbiased, it
may be unbounded on certain graphs. To see this, consider a graph that
has $n+2$ nodes $w, v, x_1, \ldots, x_n$. For each $i=1,\ldots, n$,
there is an edge from  $w$ to $x_i$ and an edge from $x_i$ to
$v$. Suppose we run  Algorithm~\ref{alg:bw} on node $w$ with target
level $\ell=2$. The algorithm first sets $\epi_0(w, w) = 1-\scw$. For
each $i = 1, \ldots, n$, the algorithm sets $\epi_1(x_i, w) = 1-\scw$
with probability $\scw$. This means there are approximately $\scw$
fraction of $x_i$'s with $\epi_1(x_i, w) = 1-\scw$. Finally, for each
$i=1,\ldots, n$ and $\epi_1(x_i, w) = 1-\scw$, the algorithm
increments $\epi_2(v, w)$ by $1-\scw$ with probability ${1\over
  n}$. This implies that in the worst-case, all $\epi_1(x_i, w) =
1-\scw$ for $i=1,\ldots, n$, and $\epi_2(v, w)$ can be as large as
$(1-\scw)n$. 

Fortunately, we can bound the variance of
Algorithm~\ref{alg:vbbw}, which enables us to use the Median 
Trick to boost accuracy. The following lemma states that the variance
of $\epi_\ell(v, w)$ is bounded by $\pi_\ell(v, w)$, the
actual value of the $\ell$-hop RPPR.

\vspace{-1mm}\begin{lemma}
\label{lem:variance}
For any  level $\ell \ge 0$ and node $v \in V$, we have  $\Var\left[\epi_\ell (v, w)\right] \le \E\left[ \epi_\ell (v, w)^2 \right] \le \pi_\ell(v, w).$
\end{lemma}

\vspace{-2mm}\subsection{ Putting Things Together}
{ Based on the definition of hub nodes,} we divide the SimRank
value $s(u,v)$  of nodes $u$ and $v$ into two terms  $s(u,v)=s_I(u, v) + s_B(u, v),$ where
\vspace{-1mm}
\begin{equation}
  \label{eqn:sI}
s_I(u, v) = {1\over (1-\sqrt{c})^2} \sum_{\ell=0}^{\infty}\sum_{j
  =1}^{j_0}  \pi_\ell(u,w_j)\pi_\ell(v, w_j) \eta(w_j),
\end{equation}
\vspace{-1mm}
and
\vspace{-1mm}
\begin{equation}
  \label{eqn:sB}
  s_B(u, v)=  {1\over (1-\sqrt{c})^2} \sum_{\ell=0}^{\infty}\sum_{j
    =j_0+1}^{n} \pi_\ell(u,w_j)\pi_\ell(v, w) \eta(w_j).
\end{equation}
\vspace{-1mm}
\prsim algorithm uses pre-computed index to generate an  estimator
$\s_I(u, v)$ for  $s_I(u, v)$, and uses backward walks to generate an estimator $\s_B(u, v)$ for $s_B(u, v)$.

Algorithm~\ref{alg:main} shows the pseudo-code of the query
algorithm for \prsim. Given a source node $u$ on a directed graph $G=(V,E)$, a
decay factor $c$ and an error parameter $\e$, the algorithm returns an
estimator $\s(u, v)$ for each $v \in V$. We set the constant $c_1 = {12\over
  (1-\scw)^2 }$, the number of samples in a round to $d_r = {c_1\over \e^2}$, the
number of rounds to $f_r = 3\log {n\over \delta} $, and the total
sample number to $n_r = d_rf_r =
\Theta\left({\log {n\over \delta} \over \e^2}\right)$ (Line 1). {Note
  that for the constant $c_1$, we choose
$(1-\scw)^2$  to compensate the denominator $(1-\scw)^2$ in
equation~\eqref{eqn:formula}, and $12$ so that we can sum
various errors up to at most  $\e$.} 
{ We choose the value of $d_r$ according to Chernoff
bound~\ref{lmm:chernoff}, and the value of $f_r$ according to the
Median Trick~\ref{lmm:median}}.
Then we initialize
estimators $\s(u,v)$ $\s_I(u, v)$, $\s_B(u, v)$ and $s_B^i(u, v)$ to be $0$ for $v\in V$ and
$i=1,\ldots, f_r$  (Line 2). We also set $\heta_\ell(u, w)$, the
estimator for $\eta(w)\cdot \pi_\ell(u, w)$, to be $0$ for $w\in V$ and
$\ell=0,\ldots, \infty$ (Line 3). Note that in order to achieve sublinear query time,
we can use  hash maps to store only the non-zero entries in $\s$, $\s_B$
$\s_I$, $\s_B^i$ and
$\heta$.

For each $i$ from $1$ to $f_r$ and $j$ from $1$ to $d_r$, we sample an
$\scw$-walk $\W(u)$ from $u$ (Lines 4-6). If $\W(u)$
terminates at node $w$ in $\ell$ steps, we further sample a pair of
$\sqrt{c}$-walks $\W_1(w)$ and $\W_2(w)$ from $w$ (Line 8). Recall
that the probability that the two $\scw$-walks do not meet is exactly
$\eta(w)$. If this event happens, we increase the estimator
$\heta_\ell(u, w)$ by ${1\over n_r}$ (Lines 9-10).  If $w$ is not
stored in the index, we estimate $\pi_\ell(v,  w)$ for each $v \in V$
with Algorithm~\ref{alg:vbbw}, and update the $i$-th estimator
$\s_B^i(u, v)$ by ${\epi_\ell(v, w) \over(1-\sqrt{c})^2d_r}$ for each
$v \in V$ (Lines 11-13). After $n_r = d_r\cdot f_r$ samples are
processed, we return $\s_B(u,v) = \textrm{Median}_{ 1\le i\le
  f_r}\s_B^i(u, v)$ as an estimator for $s_B(u, v)$ (Lines
14-15). Again, to ensure sublinear query time, we only compute median
for a node $v$ if there is at least one non-zero $\s_B^i(u,v)$ for
some $1\le i \le f_r$. Finally, for each $(w, \ell)$-tuple with
{ $\heta_\ell(u, w) > {\e \over c_1}$}  and $w$
in the index, we retrieve $\epi_\ell(v, w) $ for each $v \in V$ from
the index, and update $\s_I(u, v)$ by ${\heta_\ell(v, w)\over
  (1-\sqrt{c})^2}$ (Lines 16-18). We return all non-zero $\s(u, v) =
\s_I(u, v) + \s_B(u, v)$ as the estimator for $s(u, v)$, for $v\in V$
(Line 19).

\vspace{-1mm}
\begin{algorithm}[h]
\caption{Query Algorithm\label{alg:main}}
\KwIn{Directed graph $G=(V,E)$; node $u$; decay factor $c$; error
  parameter $\e$; {   Failure probability $\delta$}\\}
\KwOut{$\s(u, v)$ for each $v\in V$\\}
{  $c_1 \gets {12 \over (1-\scw)^2}$, $d_r \gets {c_1 \over \e^2}$, $f_r
\gets 3 \log {n\over \delta}$, $n_r
\gets d_r\cdot f_r$}\;
$\s(u, v), \s_I(u, v), \s_B(u, v), \s_B^i(u, v) \gets 0$ for each $v\in V$, $i = 1, \ldots, f_r$\;
$\heta_\ell(u, w) \gets 0$ for $w\in V$, $\ell = 0, \ldots,
\infty$\;
\For{$i=1$ to $f_r$}{
  \For{$j=1$ to $d_r$}{
  Sample an $\sqrt{c}$-walk $\W(u)$ from $u$ \;
  \If{$\W(u)$  terminates at node $w$ with $\ell$ steps}{
    Sample two independent $\scw$-walks $\W_1(w)$ and $\W_2(w)$ from $w$\;
    \If{$\W_1(w)$ and $\W_2(w)$ do not meet}{
      { $\heta_\ell(u, w)\gets \heta_\ell(u, w)+ {1\over n_r}$}\;
      \If{$w \notin Index$}{
        Estimate $\epi_\ell(v, w)$ for $v \in V$ with
        Algorithm~\ref{alg:vbbw}\;
                $\s_B^i(u,v) \gets  \s_B^i(u, v)+ {\epi_\ell(v, w) \over (1-\sqrt{c})^2d_r}$\;
        }
  }
}
}
}
  \For{each $v$ with nonzero $\s_B^i(u, v)$ for some $1 \le i \le f_r$}{
    $\s_B(u,v) \gets \textrm{Median}_{ 1\le i \le f_r}\s_B^i(u, v)$\;
  }
 \For{each $(w, \ell)$ with  $\heta_\ell(u, w)> {\e \over c_1} $ and $w \in Index$}{
    \For{each $(v, \pib_\ell(v, w)) $ tuple in $L_\ell (w)$ in Index} {
      $\s_I(u,v) \gets  \s_I(u, v)+ {\heta_\ell(u, w)\pib_\ell(v, w) \over (1-\sqrt{c})^2}$\;
    }
  }
\Return all non-zero $\s(u, v)\gets \s_B(u,v) + \s_I(u, v)$\;
\end{algorithm}

\vspace{-2mm}
\header{\bf   Error Analysis.} We now analyze the overall error bounds of the
\prsim algorithm. Recall that  given a source node $u$ and a target
node $v$, $s(u, v) = s_I(u, v) + s_B(u, v)$ where $s_I(u, v)$ and
$s_B(u,v)$ are defined by equations~\eqref{eqn:sI} and~\eqref{eqn:sB}, respectively.
Algorithm~\ref{alg:main} uses index to generate an
estimator $\s_I(u,v)$ for each  $s_I(u,
 v), v\in V$, and uses backward walks to generate an estimator
 $\s_B(u,v )$ for each  $s_B(u, v), v\in V$.
We have the following two lemmas that bound the errors of the two
approximations.

    \vspace{-1mm}\begin{lemma}
      \label{lem:error_I}
      Given a source node $u$, for any $v\in V$,
      Algorithm~\ref{alg:main} provides an estimator $\s_I(u, v)$ for
      $s_I(u, v)$ such that:
            \begin{equation}
        \label{eqn:I}
      \Pr\left[|\s_I(u, v) - s_I(u, v)| >{\e \over 2}  \right] \le {{ \delta}\over
        2n}.
      \end{equation}
    \end{lemma}

    \vspace{-1mm}\begin{lemma}
      \label{lem:error_B}
            Given a source node $u$, for any $v\in V$,
      Algorithm~\ref{alg:main} provides an estimator $\s_B(u, v)$ for
      $s_B(u, v)$ such that:
            \begin{equation}
        \label{eqn:B}
      \Pr\left[|\s_B(u, v) - s_B(u, v)| >{\e \over 2}  \right] \le
      {{  \delta}\over
        2n}.
      \end{equation}
    \end{lemma}

\noindent
Combining Lemmas~\ref{lem:error_I} and~\ref{lem:error_B} follows
      that
      \begin{align*}
       \Pr\left[|\s(u, v) - s(u, v)|  > \e  \right] 
       { \le {\delta\over 2n}+{\delta \over 2n}= {\delta
        \over n}}.
      \end{align*}
Applying union bound on $n$ nodes follows  Theorem~\ref{thm:error}.

\vspace{-1mm}\begin{theorem}
  \label{thm:error}  \prsim answers single-source SimRank
  queries with additive error $\e$ with probability at least
  $1-\delta$.
\end{theorem}

\header{\bf   Query Time Analysis for  Worst-Case Graphs.}  We first analyze the query time of the \prsim
algorithm on worst-case graphs.
Given a node $u\in V$, let $C(u)$ denote the query cost of \prsim on
$u$, and $C = {1\over n}\sum_{u\in V} C(u)$ denote the average query cost.  We divide
$C(u)$ into three terms: $C(u) = C_F(u)  + C_I(u) +
C_B(u),$
  where $C_F(u)$ denote the cost for computing  $\heta_\ell(u,
  w)$ from source node $u$, $C_I(u) $ denote the
  query cost for retrieving reserves   $\pib_\ell(v, w)$  from the
  index, and $C_B(u)$ denote the query cost for estimating
  $\epi_\ell(v, w)$ with backward walks.  Let $C_F=  {1\over n} \sum_{u\in
  V}C_F(u)$, $C_I = {1\over n} \sum_{u\in
  V}C_I(u)$ and  $C_B = {1\over n} \sum_{u\in
  V}C_B(u)$ denote the average query cost of $C_F(u)$, $C_I(u)$ and
$C_B(u)$, respectively. We can express the expected average query
cost of Algorithm~\ref{alg:main} as
$\E[C] = \E[C_F] + \E[C_I] +\E[C_B] .$

  For $\E[C_F]$, recall that we generate a number  $n_r =
  \Theta\left({\log {n\over \delta} \over
    \e^2}\right)$ of $\scw$-walks to estimate $\heta_\ell(u, w)$. Since each
  $\scw$-walk takes constant time, we have
$C_F(u) = O\left({\log {n\over \delta} \over \e^2}\right)$, and  $\E[C_F] = O\left({\log {n\over \delta} \over
  \e^2}\right)$. We have the following lemmas for $\E[C_I] $ and $\E[C_B] $.

\vspace{-1mm}\begin{lemma}
  \label{lem:query_I}
Let $c_1 = {12
  \over (1-\scw)^2}$ and $C_I$ denote the average
cost for querying the index. We have
 $$\E[C_I] = O\left(\min\left\{  {n \over \e}\sum_{j=1}^{c_1\over \e} \pi(w_j), {n  \over \e^2}
\sum_{j=1}^{j_0} \pi(w_j)^2\right\} \right).$$
\end{lemma}

    \vspace{-1mm}\begin{lemma}
        \label{lem:query_B}
Let $C_B$ denote the average
cost for performing Variance Bounded Backward Walks. We have
 $\E[C_B]= O\left({n\log {n\over \delta}  \over \e^2}
\sum_{j=j_0+1}^{n} \pi(w_j)^2 \right).$
\end{lemma}

{ 
By  Lemma~\ref{lem:query_I}, we have $\E[C_I]  \le O\left( {n\log {n\over
    \delta}\over \e^2}
\sum_{j=j_0+1}^{n} \pi(w_j)^2 \right)$. Combining with
Lemma~\ref{lem:query_B} follows Theorem~
\ref{thm:query}.

\vspace{-1mm}\begin{theorem}
  \label{thm:query}
Suppose the query node $u$ is uniformly chosen from $V$. The
 expected query cost of \prsim on worst-case graphs is  bounded by
 \begin{equation}
\label{eqn:query_worst_case}
\E[C]= O\left({n\log {n\over \delta} \over \e^2} \cdot \sum_{w\in V} \pi(w)^2\right).
\end{equation}
\end{theorem}
}

\header{\bf  Query Time Analysis for  Power-Law Graphs.}
Recall that on a
power-law graph, the fractions $P_o(k)$ and $P_i(k)$ of nodes with
out- and in-degree at least $k$ satisfy that
$P_o(k) \sim k^{-\gamma}$ and $P_i(k) \sim k^{-\gamma'}$~\cite{BollobasBCR03},
where $\gamma$ and $\gamma'$ are the cumulative power-law exponents that
usually take values from $1$ to $3$. It is shown in
\cite{BahmaniCG10,lofgren2015personalized,wei2018topppr} that the
PageRank of a power-law graph also follows power-law with same
exponent $\gamma'$ as the {\em in-degree distribution}. Thus,
the reverse PageRank follows the same power-law distribution as the
out-degree distribution. 
In particular, let
$P_\pi(x)$ denote the portion of nodes with reverse PageRank value at
least $x$,  then $P_\pi(x) \sim x^{-\gamma}.$ 

Now consider the following alternating statement of the above power-law
distribution: let $w_1, \ldots, w_n$ denote the nodes in the graph
sorted in descending order of their reverse PageRank values, that is,
$\pi(w_1) \ge \pi(w_2) \ge \ldots \ge \pi(w_n)$. We have that the
$j$-th largest reverse PageRank value $\pi(w_j)$ is proportional to $ j^{-\beta}$. Here $\beta$ is the power-law exponent that takes value from $(0, 1)$.  This assumption has
been widely adopted in the literature of PageRank
computations~\cite{BahmaniCG10,lofgren2015personalized,wei2018topppr}. To
understand the relation between two exponents $\gamma$ and $\beta$,
note that there are $j$ nodes with reverse PageRank value at least $x=
{\kappa j^{-\beta} \over n^{1-\beta}} $, and thus we have { 
$j \sim \left( { j^{-\beta}
  \over n^{1-\beta}} \right)^{-\gamma} \sim j^{\beta \cdot \gamma}.$
It follows that $\beta = {1 \over \gamma}$.
Therefore, for power-law graphs, we have
\begin{equation}
  \label{eqn:PageRank_power_law}
  \pi(w_j) =\kappa \cdot  {j^{-\beta}/ n^{1-\beta}} = \kappa
  \cdot  {j^{-{1\over \gamma}}/ n^{1-{1\over \gamma}}} ,
  \end{equation}
where $\kappa$ is a normalization constant such that $\kappa
\sum_{j=1}^n {j^{-{1\over \gamma}} \over n^{1-{1\over \gamma}}} =1$.}
{ 
Combing equation~\eqref{eqn:PageRank_power_law} and
Lemma~\ref{lem:index_size}, the index size is bounded by $O\left({n\over \e}\sum_{j=1}^{j_0} {j^{-{1\over \gamma}} \over
   n^{1-{1\over \gamma}}} \right)= O\left({n\over \e} \cdot
  {j^{1-{1\over \gamma}} \over n^{1-{1\over \gamma}}} \right)= O\left({n^{1\over \gamma} j_0^{1-{1\over \gamma}} \over
   \e }\right) .$
 Here we use the property of Riemann zeta function
 (see Lemma~\ref{lem:rz}). By setting $j_0 = n(\e \d)^{\gamma \over
   \gamma - 1}$,
 we have index  size is bounded by
 $O\left({n^{1\over \gamma}  n^{1-{1\over \gamma} }  \e \d \over
     \e }\right) = O(m).$ Plugging $  \pi(w_j) = \kappa
  \cdot  {j^{-{1\over \gamma}} \over n^{1-{1\over \gamma}}} $  and $j_0 = n(\e \d)^{\gamma \over
   \gamma - 1}$ into
 Lemma~\ref{lem:query_B} and Lemma~\ref{lem:query_I}, and we have the
 following theorem.}
 
 { 
\vspace{-1mm}\begin{theorem}
  \label{thm:query_power_law}
 Assume that the out-degree distribution of the graph follows power-law distribution
 with exponent $\gamma \ge 1$, and let $\e \ge \log^{\gamma - 1\over
     2 - \gamma} n / (n^{\gamma - 1 \over \gamma}
   \d^{ 2 - \gamma})$, $\delta > 1/n^{\Omega(1)}$.  Suppose the query node $u$ is uniformly chosen
   from $V$. By setting $j_0 = n(\e \d)^{\gamma \over
   \gamma - 1}$, the
 expected cost of Algorithm~\ref{alg:main} is  bounded by
 \begin{equation}
\label{eqn:query}
\E[C]= \left\{
\begin{array}{ll}
O({1\over \e^2} \log {n\over \delta}), &\textrm{for } \gamma > 2; \\
  O({1\over \e^2} \log {n\over \delta} \log n) , & \textrm{for }\gamma = 2; \\
O\left(\min\left\{{n^{1\over \gamma}\over \e^{2-{1\over \gamma}}},
  {n^{{2\over \gamma}-1}\over \e^2}
  \right\} \right) , & \textrm{for } 1 < \gamma  < 2.
\end{array}\right.
\end{equation}
The size of the index generated by Algorithm~\ref{alg:indexing}  is bounded by $O(m)$. The preprocessing time is bounded by $O\left(
  m \over \e\right)$.
\end{theorem}
}

{ 
\header{\bf Dynamic Graphs.} Our algorithm is able to support dynamic
graphs where edges may be inserted or deleted.  Recall that \prsim generates the
index by performing the backward search 
algorithm. It is shown in~\cite{ZhangLG16} that the results of
the backward search to a
randomly selected target node $w$ can be
maintained with cost $O(k+{\bar{d}\over \e})$, where $k$ is the total number
of insertions/deletions. Since our index stores the results of the backward search
for $j_0$ target nodes, it can process
$k$ insertions/deletions in  $O(kj_0+{m\over \e})$ time. Therefore,
the per-update-cost for processing $k$ updates is bounded by $O(j_0+{m\over \e k})$. However, a thorough investigation of this issue is
beyond the scope of our paper. 
}


%% file: analysis.tex
\vspace{-2mm}

%% file: related.tex
\section{Related Work} \label{sec:related}
In what follows, we  briefly review some of the state-of-the-art solutions for SimRank computation. We exclude SLING~\cite{TX16}, which we have discussed in Section~\ref{sec:prelim}.

\vspace{-1mm}\header {\bf Monte Carlo and READS.} Based on the $\scw$-walk interpretation, we can use the following Monte Carlo algorithm~\cite{FRCS05,TX16} to estimate the SimRank value $s(u, v)$: we generate $n_r$ pairs of $\scw$-walks from $u$ and $v$, and use the percentage of $\scw$-walks that meet as an estimation of $s(u,v)$. Using concentration inequality, one can show that by setting $n_r = \Theta\left({\log {n\over \delta} \over \e^2 }\right)$, the Monte Carlo algorithm estimates $s(u, v)$ with an additive error $\e$ with probability at least $1-\delta$. For a single-source query on node $u$, we can generate $n_r$ walks from each node $v\in V$ and estimate $s(u, v)$ with additive error $\e$. The query cost is  $O\left({n\log {n\over \delta} \over\e^2}\right)$, which is inefficient on large graphs.  

A recent work proposes the READS algorithm~\cite{jiang2017reads} based on the Monte Carlo approach. READS pre-computes the $\scw$-walks from each node, and compresses the $\scw$-walks by merging them into trees. Given a query node $u$, READS retrieves the $\scw$-walks starting from $u$, finds all $\scw$-walks that meet with $u$'s  $\scw$-walks, and then updates the SimRank estimator for each $v$ related to these $\scw$-walks. Several
optimization techniques were adopted to improve the query efficiency of READS. The major issue of READS is that it requires generating and storing a large number of $\scw$-walks from each node in the preprocessing phase. The query cost also remains $O(n\log {n\over \delta} /\e^2)$, which is the same as that of the classic Monte Carlo algorithm.

\vspace{-1mm}\header{\bf ProbeSim.} ProbeSim~\cite{liu2017probesim} is an index-free algorithm that computes single-source and top-k SimRank queries on large graphs. Given a query node $u$, the ProbeSim algorithm samples a $\scw$-walk  $\W(u)$ from $u$. For a node $w$ visited by $\W(u)$ at the $\ell$-th step, the algorithm performs a {\em Probe} procedure that computes the probability of an $\scw$-walk from each node $v$
visiting $w$ at the $\ell$-th step. To rule out the probability that a pair of $\ell$-walks may meet multiple times,  the Probe algorithm avoids the nodes previously visited by $\W(u)$. It is shown in \cite{liu2017probesim} that the ProbeSim algorithm gives an unbiased estimator for the SimRank values $s(u, v), v\in V$. Therefore, by repeating the sampling procedure $O(\log {n\over \delta} / \e^2)$ times, ProbeSim answers single-source SimRank queries with probability at least $1-\delta$.

There are two subtle problems with ProbeSim. First, to avoid
multiple meeting nodes, the Probe from node $w$ has to avoid the nodes
on $\W(u)$, which means it is impossible to pre-compute the Probe
results to speed up the query time. Second, as we will show later,
the probability that a node $w$ in the graph is visited by the
$\scw$-walk from $u$ is proportional to $\pi(w)$, the reverse PageRank of $w$.  On the other hand,  the complexity of the Probe algorithm on $w$ is also proportional to $\pi(w)$. This essentially means it is likely that a {\em hub node} with high reverse PageRank value is visited by the $\scw$-walk from $u$, and it will incur significant cost in the Probe phase. Finally, the algorithm also requires $O(n\log {n\over \delta} /\e^2)$ query
cost to answer a single-source query.

\vspace{-1mm}\header{\bf TSF.} TSF~\cite{SLX15} is a two-stage random-walk sampling
algorithm for single-source and top-$k$ SimRank queries on dynamic graphs. Given a
parameter $R_g$,  {TSF} starts by building $R_g$ {\em one-way
  graphs} as an index structure.  Each one-way graph is constructed by
uniformly sampling \emph{one} in-neighbor from each vertex's in-coming
edges. The one-way graphs are then used to simulate random walks
during query processing. To achieve high efficiency, { TSF} allows
two $\scw$-walks to meet multiple times, and thus overestimate the
actual SimRank values. Furthermore, { TSF} assumes that every
random walk  would not contain any cycle, which does
not hold in practice.

\header{\bf Other Related Work.}
{\em Power method} \cite{JW02} is the classic algorithm that computes all-pair SimRank similarities for a given graph. Let $S$ be the SimRank matrix such that $S_{ij} = s(i, j)$, and $A$ be the transition matrix of $G$. Power method recursively computes the SimRank Matrix $S$ using the following formula \cite{KMK14}
\vspace{-1mm}\begin{equation} \label{eqn:related-simrank}
S = (c A^\top S A) \vee I,
\end{equation}
\vspace{-1mm}
where $\vee$ is the element-wise maximum operator. Several follow-up works \cite{LVGT10,YZL12,YuJulie15gauging} improve the efficiency or effectiveness of the power method in terms of either efficiency or accuracy. However, these methods still incur $O(n^2)$ space overheads, as there are $O(n^2)$ pairs of nodes in the graph. { A recent work~\cite{wang2018efficientsimrank} reduces the cost to $O(NNZ)$, where $NNZ$ is the number of node pairs with large SimRank similarities. However, as shown in~\cite{wang2018efficientsimrank}, there are still a constant fraction of $O(n^2)$  node pairs with large SimRank similarities, so the worst case complexity remains $O(n^2)$}. 

Motivated by difficulty in dealing with the element-wise maximum operator $\vee$ in Equation~\ref{eqn:related-simrank}, some existing work \cite{FNSO13,He10,Yu13,Li10,Yu14,YuM15b,KMK14} consider the following alternative formula for SimRank:
\vspace{-1mm}\begin{equation} \label{eqn:related-simrank-wrong}
S = c A^\top S A +  (1-c)\cdot I.
\end{equation}
\vspace{-1mm}
However, it is shown that the similarities calculated by this formula are different from SimRank \cite{KMK14}.  

For single-source queries, Fogaras and R{\'{a}}cz \cite{FRCS05} propose a Monte Carlo algorithm that uses random walks to approximate SimRank values.
Maehara et al.\ \cite{MKK14} propose an index structure for top-$k$ SimRank queries, but it relies on heuristic assumptions about $G$, and hence, does not provide any worst-case error guarantee.
Li et al.\ \cite{LiFL15} propose a distributed version of the Monte Carlo approach in \cite{FRCS05}, but it achieves scalability at the cost of significant computation resources. Finally, there is existing work on variants of SimRank \cite{AMC08,FR05,YuM15a,ZhaoHS09} and on various graph applications~\cite{bhuiyan2018representing,ye2018using,lee2018evaluations}, but the proposed solutions are inapplicable for top-$k$ and single-source SimRank queries.

%% file: experiment.tex
\vspace{-2mm}
\section{Experiments} \label{sec:exp}

This section experimentally evaluates the proposed solutions against
the state of the art. All experiments are conducted on a machine with
a Xeon(R) CPU E7-4809@2.10GHz CPU and { 196GB} memory.

\vspace{-2mm}\subsection{Experimental Settings} \label{sec:exp-setting}

\noindent
{\bf Methods.} We compare \prsim against five SimRank algorithms: {READS}~\cite{jiang2017reads}, {
SLING}~\cite{TX16}, { TSF}~\cite{SLX15},
{ProbeSim}~\cite{liu2017probesim}  and  {TopSim}~\cite{LeeLY12}. As
mentioned in Section~\ref{sec:related}, READS, SLING and TSF are the
state-of-the-art index-based methods, and {ProbeSim} and  {TopSim} are the
state-of-the-art index-free methods.

\vspace{-1mm}\header
\noindent{\bf Ground Truth for single-pair queries.} Given a pair of nodes $u$
and $v$, we use the Monte Carlo algorithm to estimate $s(u, v)$ with high
precisions, and then use the result as the ground truth for $s(u, v)$. In
particular,  we set the parameters of the Monte
Carlo algorithm such that it incurs an error less than $0.00001$ with
confidence over $99.999\%$.

\vspace{-1mm}\header{\bf Pooling.} We extend the {\em pooling}
idea~\cite{liu2017probesim} to evaluate the effectiveness of the single-source algorithms
on large graphs. Given a source node $u$, we run each single-source
algorithm, order the nodes according to their estimated SimRank
values, and retrieve the top-$k$
nodes. We merge the top-$k$ nodes returned by
each algorithm, remove
the duplicates, and put them into a {\em
  pool}. As such, if we were to evaluate $\ell$ algorithms, then the pool size is
between $k$ and $\ell k$. For each
node $v$ in the pool, we obtain the ground truth of $s(u, v)$ using the
Monte Carlo algorithm, and retrieve $V_k = \{v_1, \ldots, v_k\}$, namely, the
$k$ nodes with the highest SimRank values from the pool.

\vspace{-1mm}\header{\bf Metrics. } To evaluate the absolute error of
single-source SimRank
algorithms, we calculate the average absolute errors for
approximating $s(u, v_i)$ for each $v_i$ in the pool. More precisely, for each $v_i \in
V_k$ returned by the pool, let $\s(u, v_i)$ be
the estimator for $s(u,v_i)$ returned by the algorithm
to be evaluated. We set
\vspace{-2mm}
$$ AvgError@k = {1\over k}\sum_{1\le i \le k}|\s(u, v_i) - s(u,
v_i)|.$$

\vspace{-2mm}
To evaluate the algorithms' abilities to return the top-$k$ results,  we use $V_k = \{v_1, \ldots, v_k\} $ as the ground truth for
the top-$k$ nodes. Note that these nodes are the best possible results
that can be returned by any of the algorithms to be evaluated.
Let $V_k'=\{v_1', \ldots,
v_k'\}$  denote the top-$k$ node set returned by the algorithm to be
evaluated. 
Note that {\em Precision@k} evaluates how many correct (or best possible) nodes are
included in $V_k'$.

\begin{figure*}[!t]
\begin{small}
 \centering
   \vspace{0mm}
    \begin{tabular}{ccccc}
     \hspace{-2mm} \includegraphics[height=28mm]{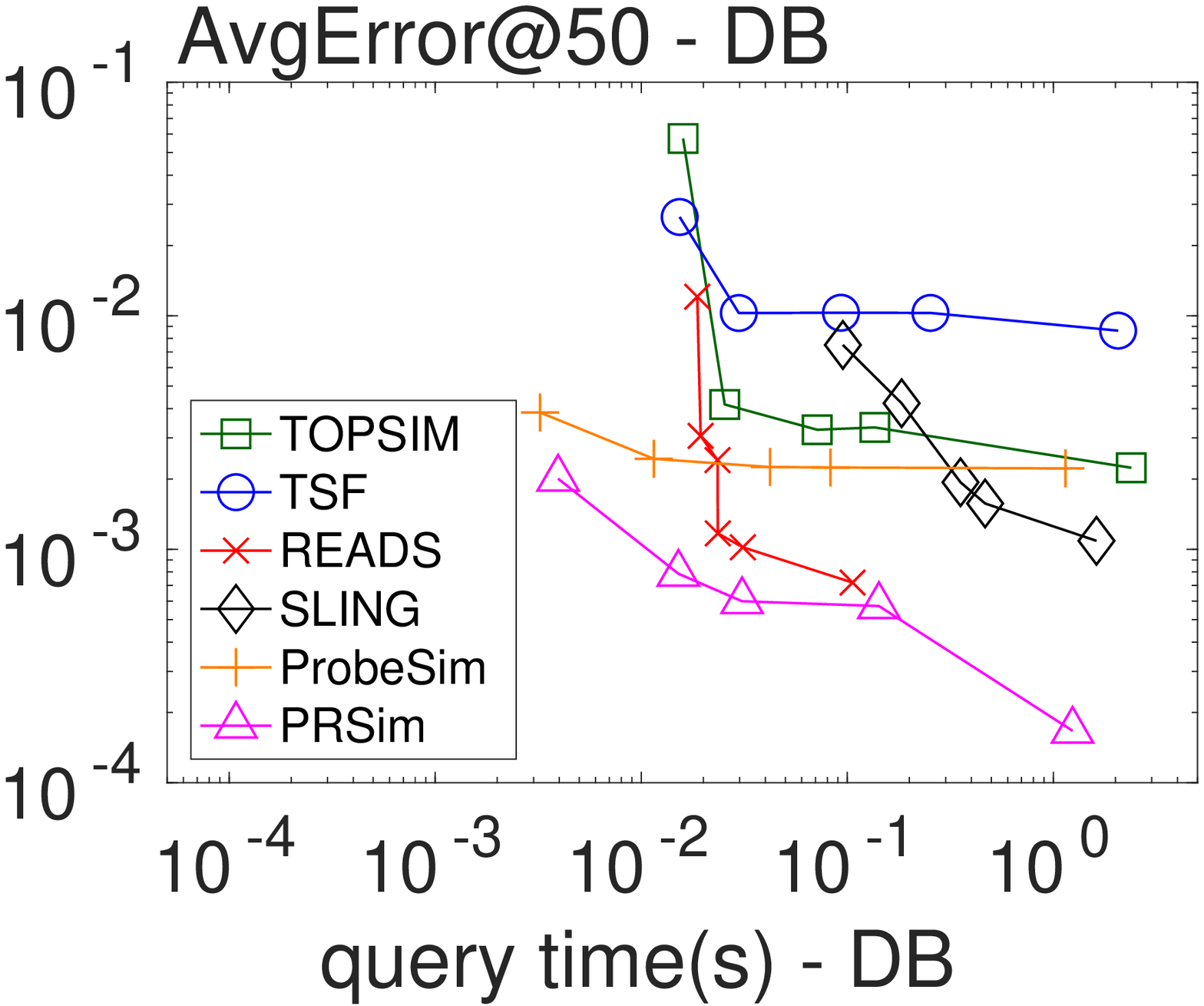}
   &
    \hspace{-2mm}
     \includegraphics[height=28mm]{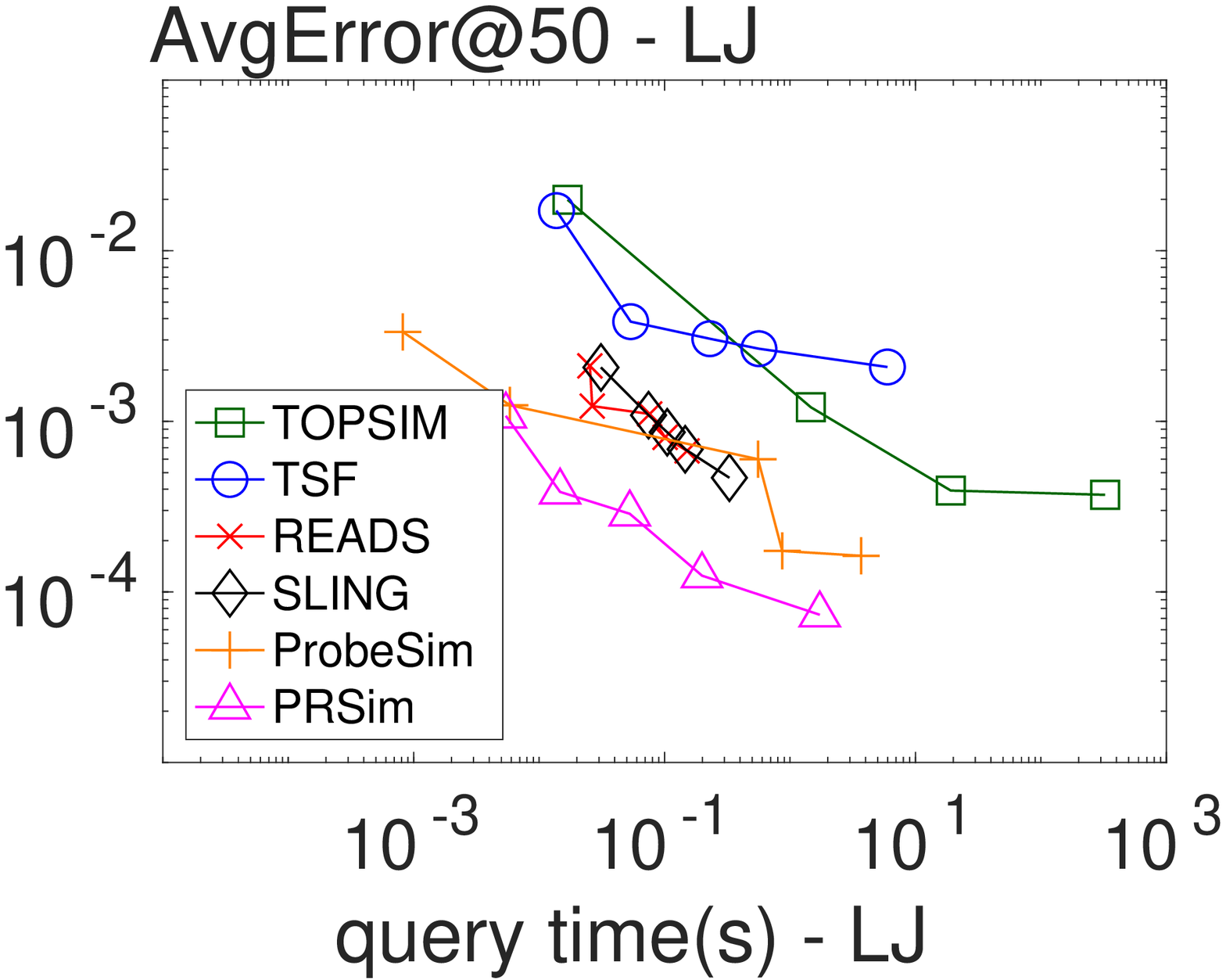}
      &
    \hspace{-2mm}
        \includegraphics[height=28mm]{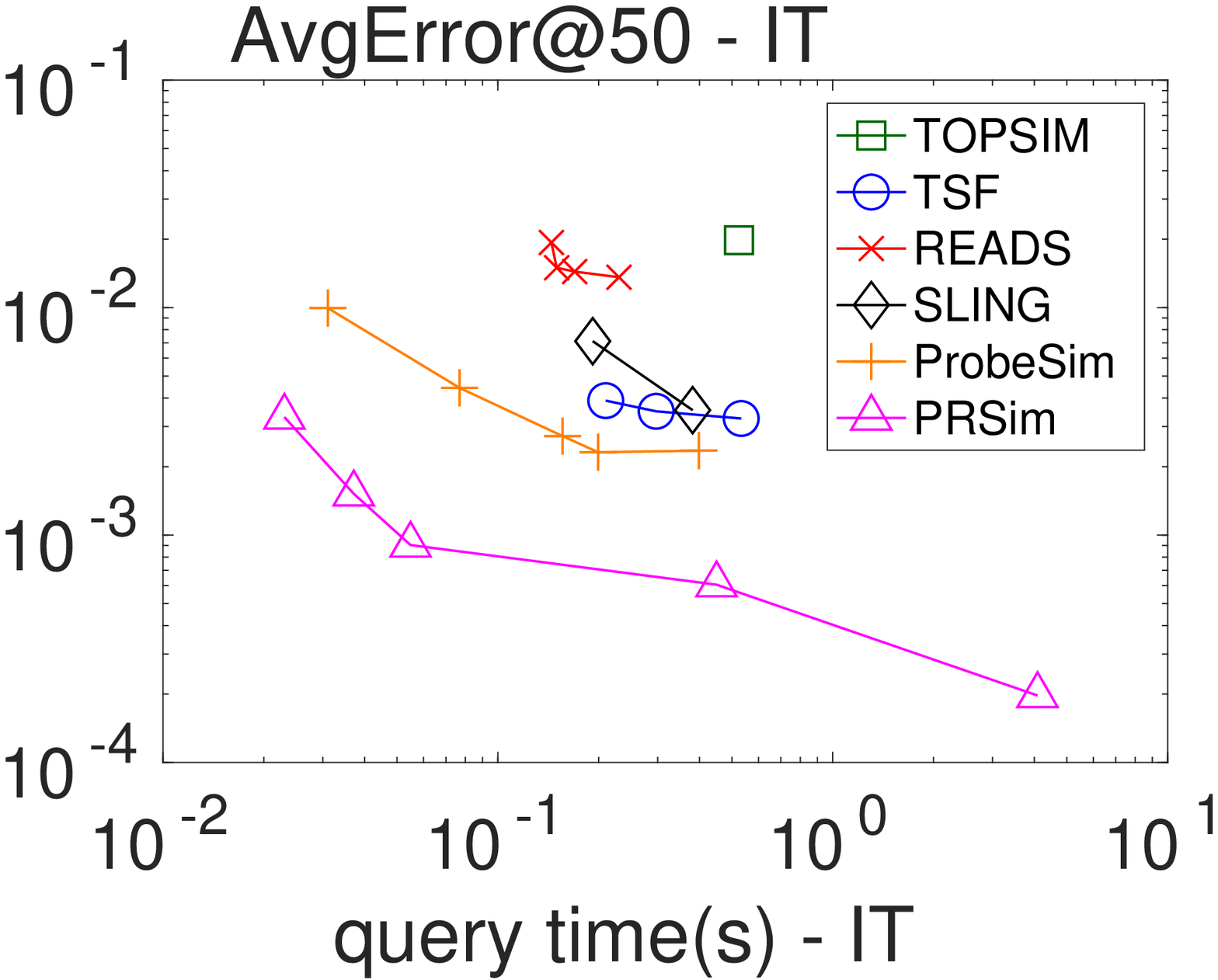}
         &
    \hspace{-2mm}
           \includegraphics[height=28mm]{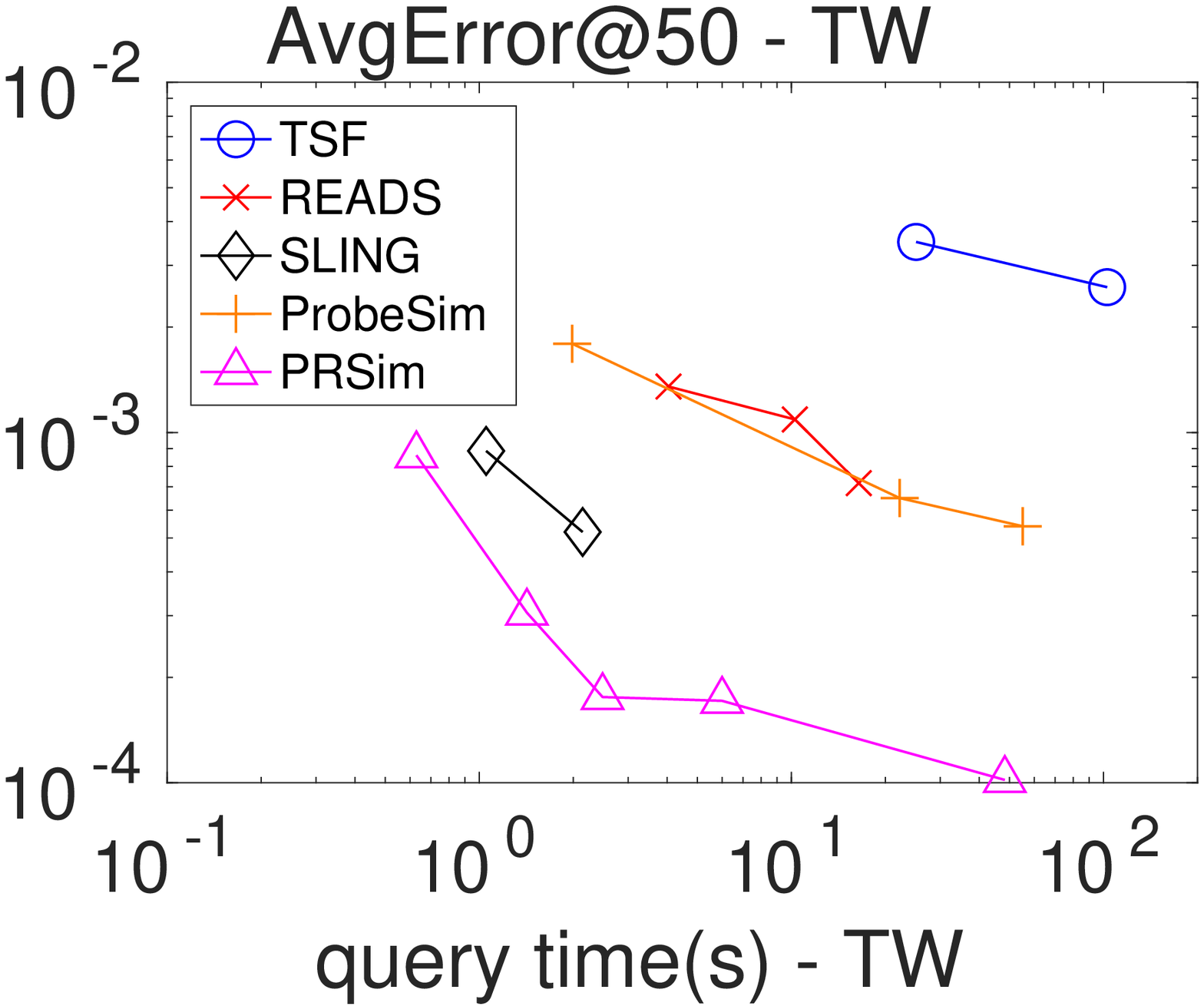}
           &
    \hspace{-2mm}
     \includegraphics[height=28mm]{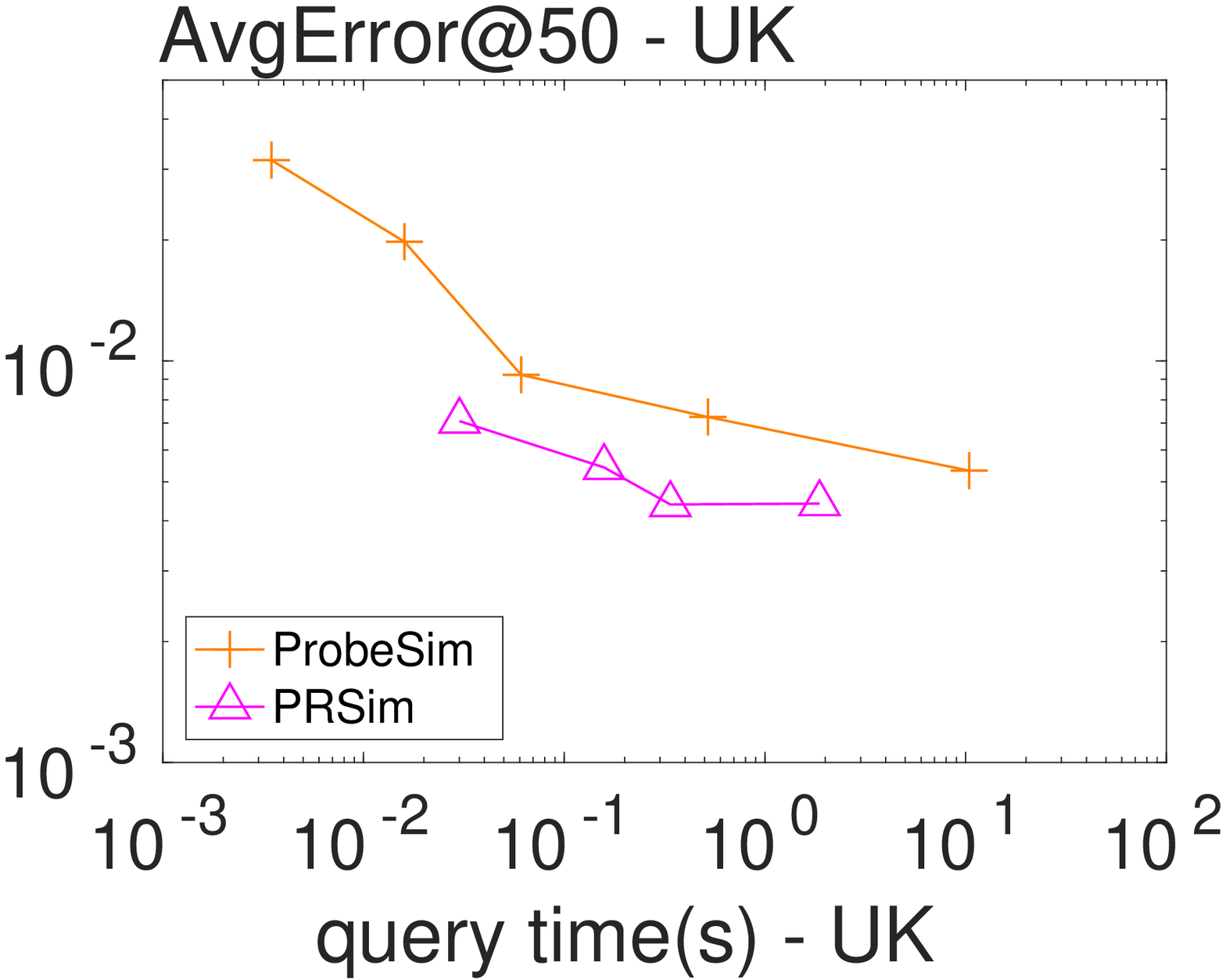}
   \vspace{0mm} \\
 \end{tabular}
\vspace{-3mm}
\caption{ $\mathbf{AvgError@50}$ v.s. {\em Query time } } \label{fig:error_query_time}
\vspace{0mm}
\end{small}
\end{figure*}

\begin{figure*}[!t]
\begin{small}
 \centering
   \vspace{0mm}
    \begin{tabular}{ccccc}
     \hspace{-0.8mm} \includegraphics[height=28mm]{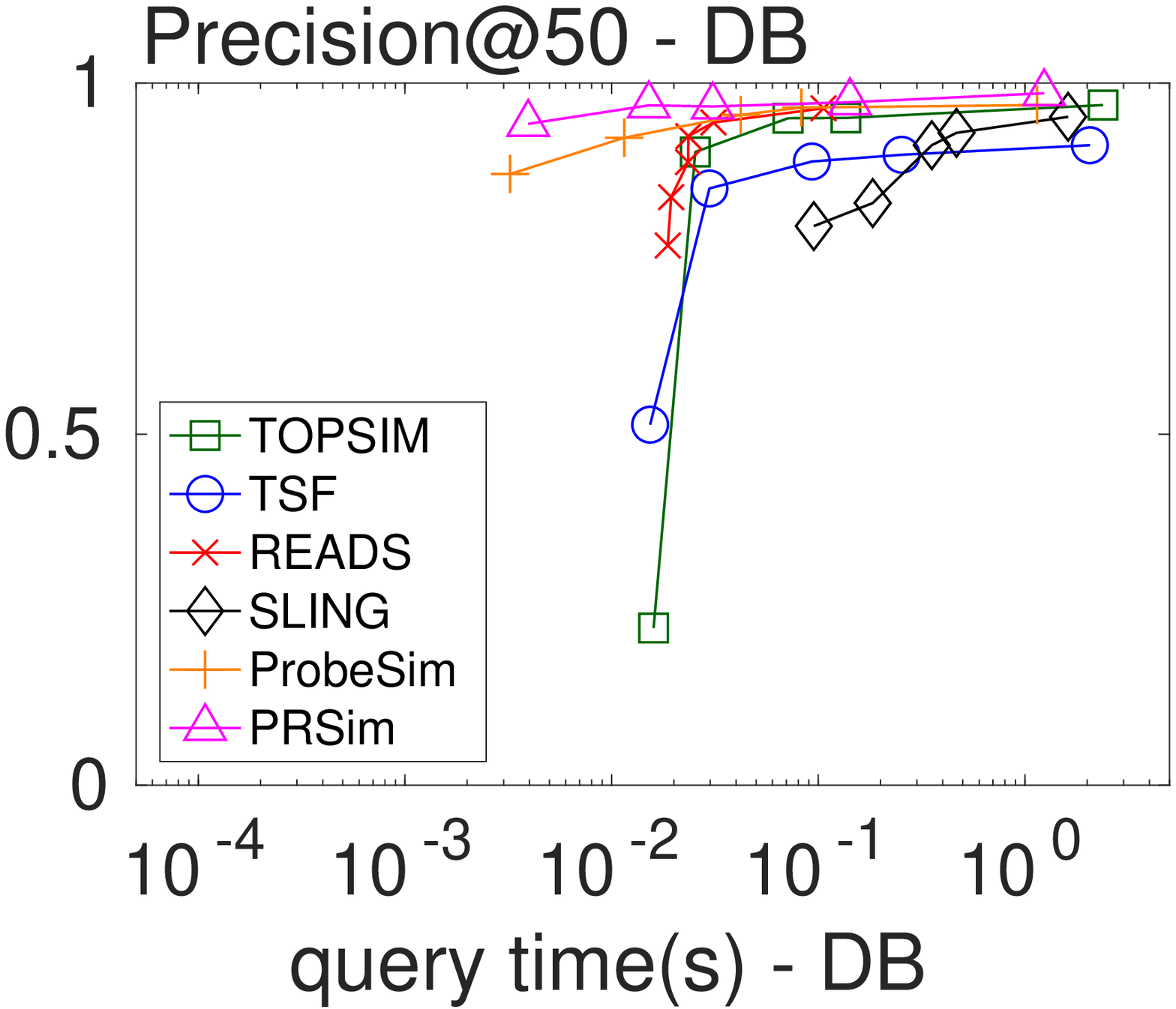}
   &
    \hspace{-0.8mm}
     \includegraphics[height=28mm]{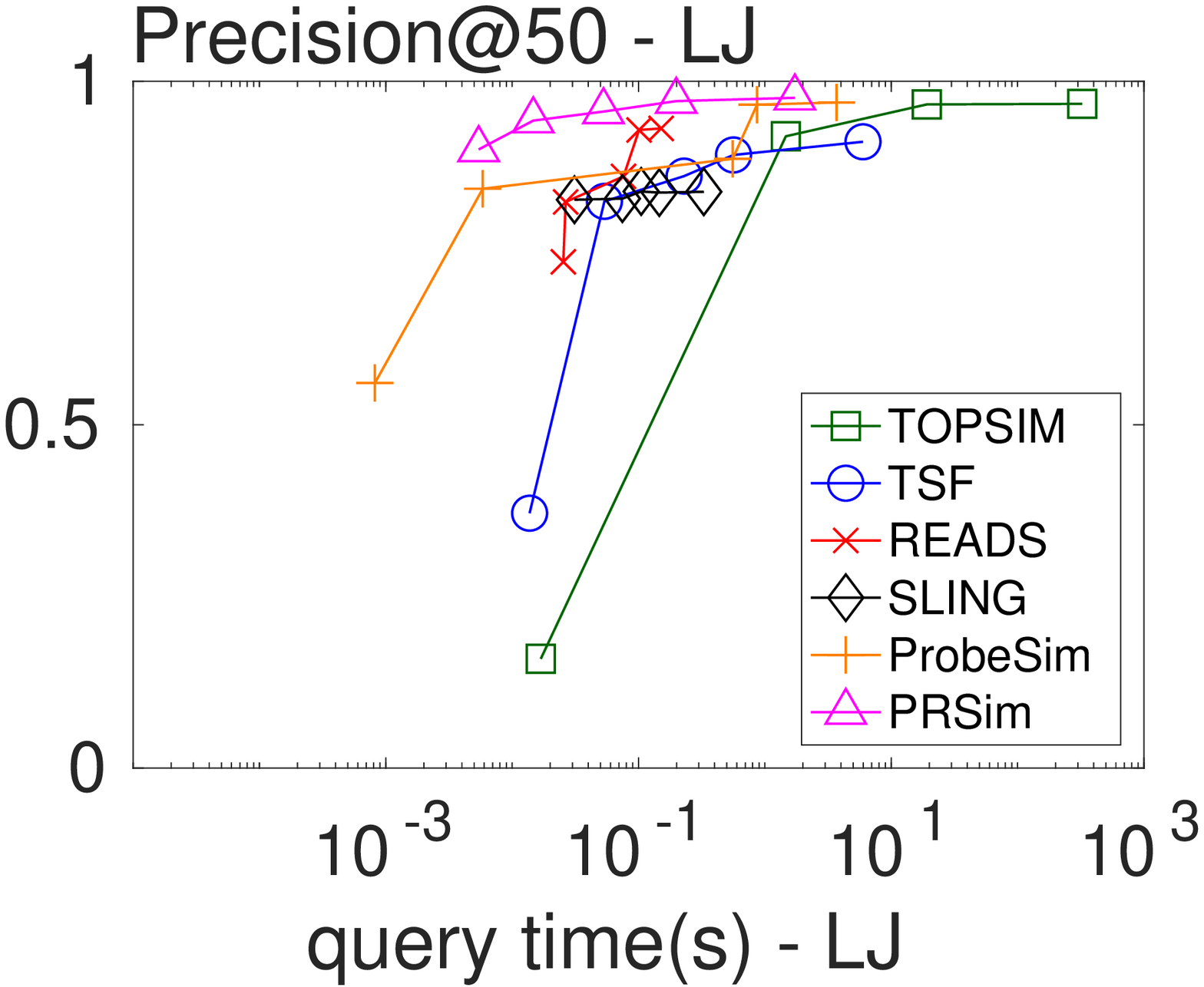}
      &
    \hspace{-0.8mm}
        \includegraphics[height=28mm]{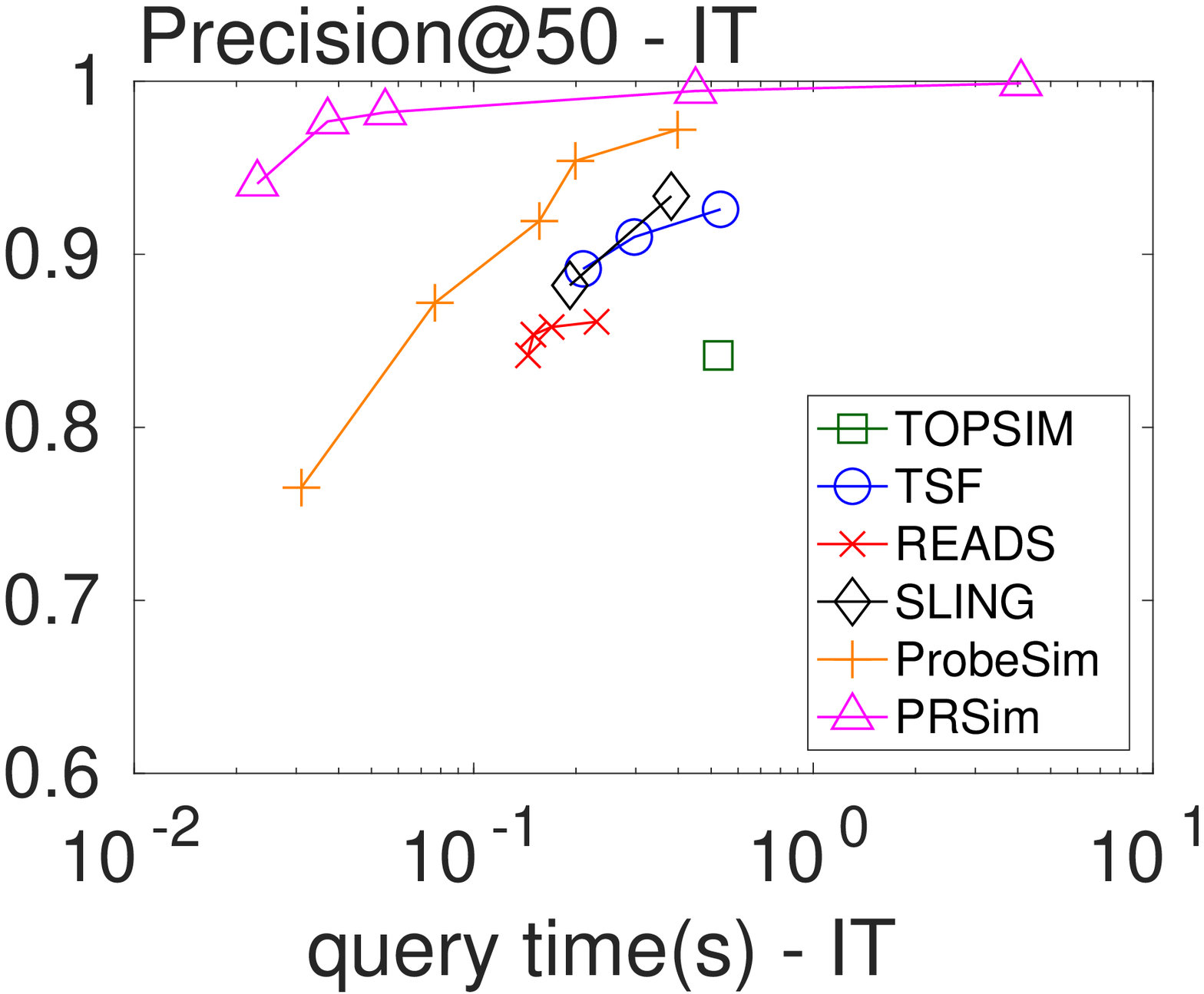}
         &
    \hspace{-0.8mm}
           \includegraphics[height=28mm]{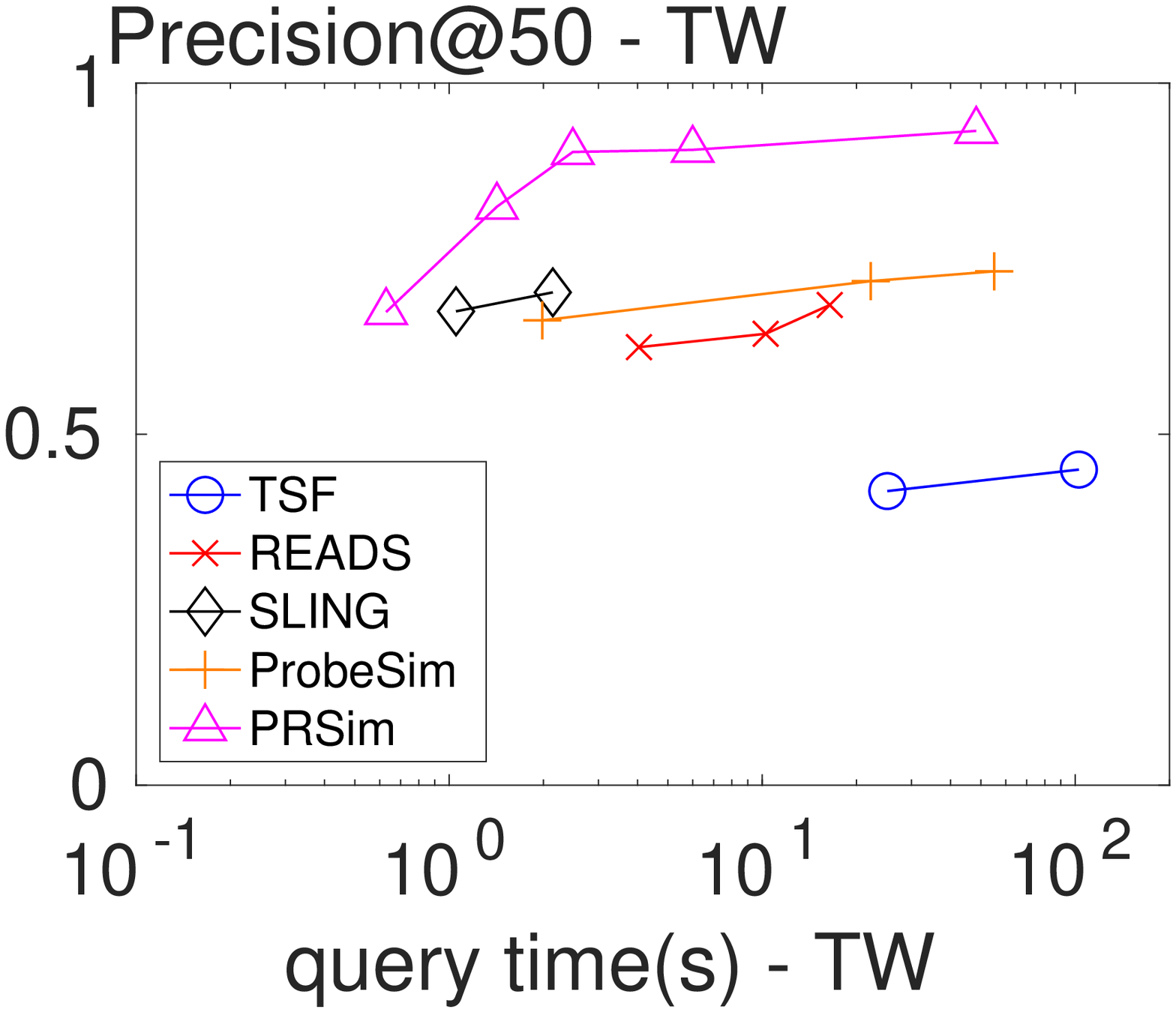}
           &
    \hspace{-0.8mm}
     \includegraphics[height=28mm]{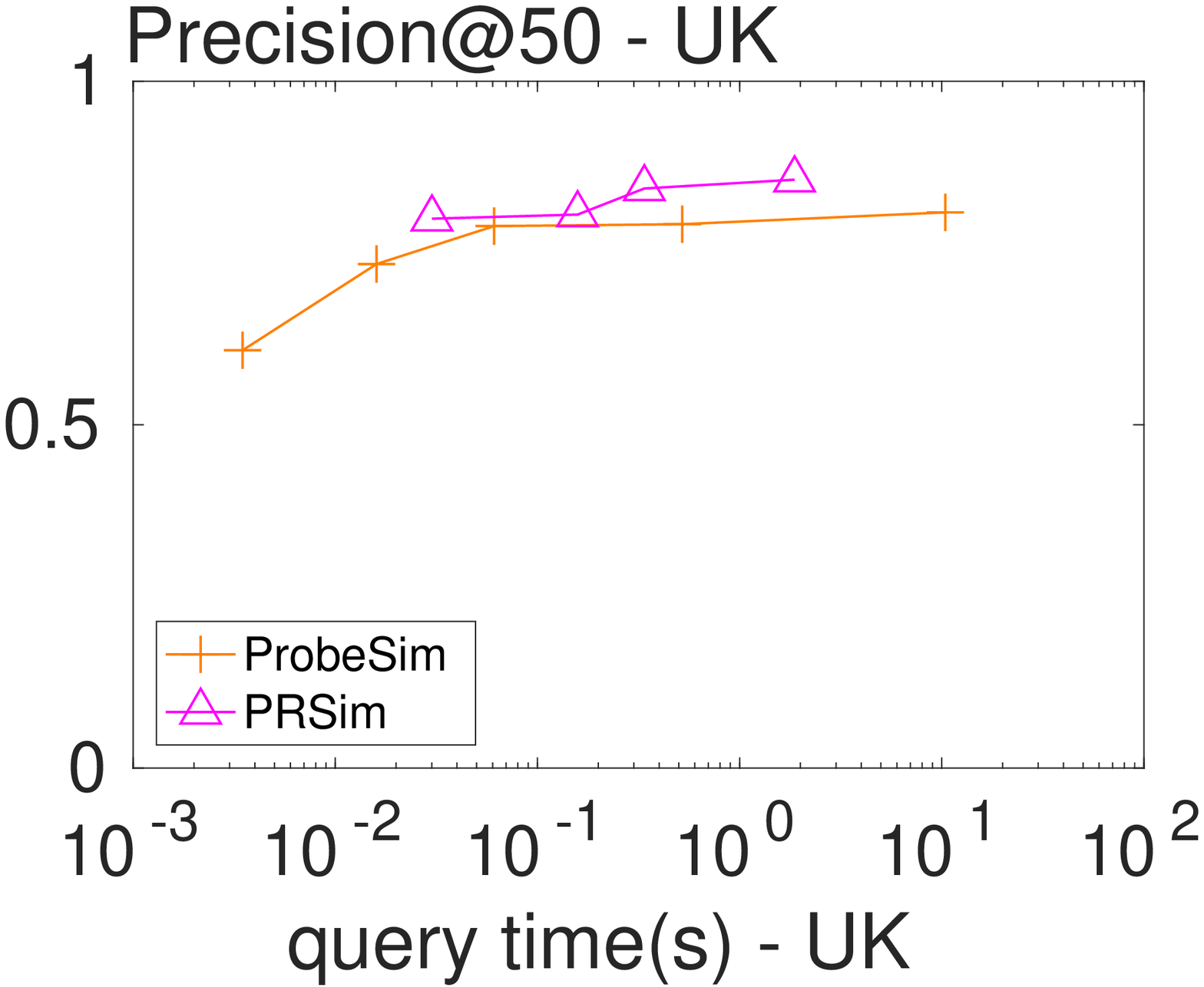}
   \vspace{-2mm} \\
 \end{tabular}
\vspace{-3mm}
\caption{ $\mathbf{Precision@50}$ v.s. {\em Query time } } \label{fig:precision_query_time}
\vspace{0mm}
\end{small}
\end{figure*}

\begin{figure*}[!t]
\begin{small}
 \centering
   \vspace{0mm}
    \begin{tabular}{ccccc}
     \hspace{-2mm} \includegraphics[height=28mm]{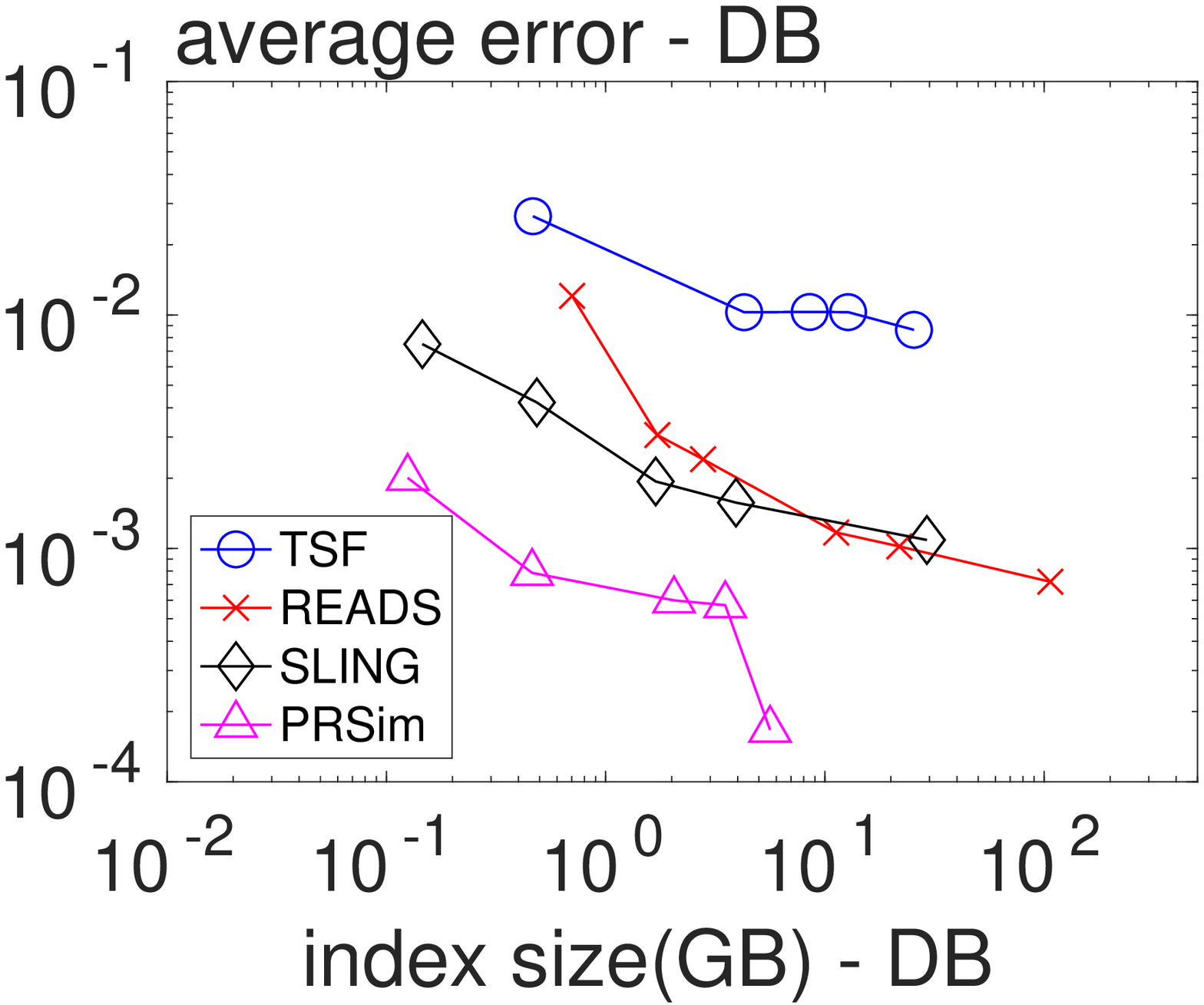}
   &
    \hspace{-2mm}
     \includegraphics[height=28mm]{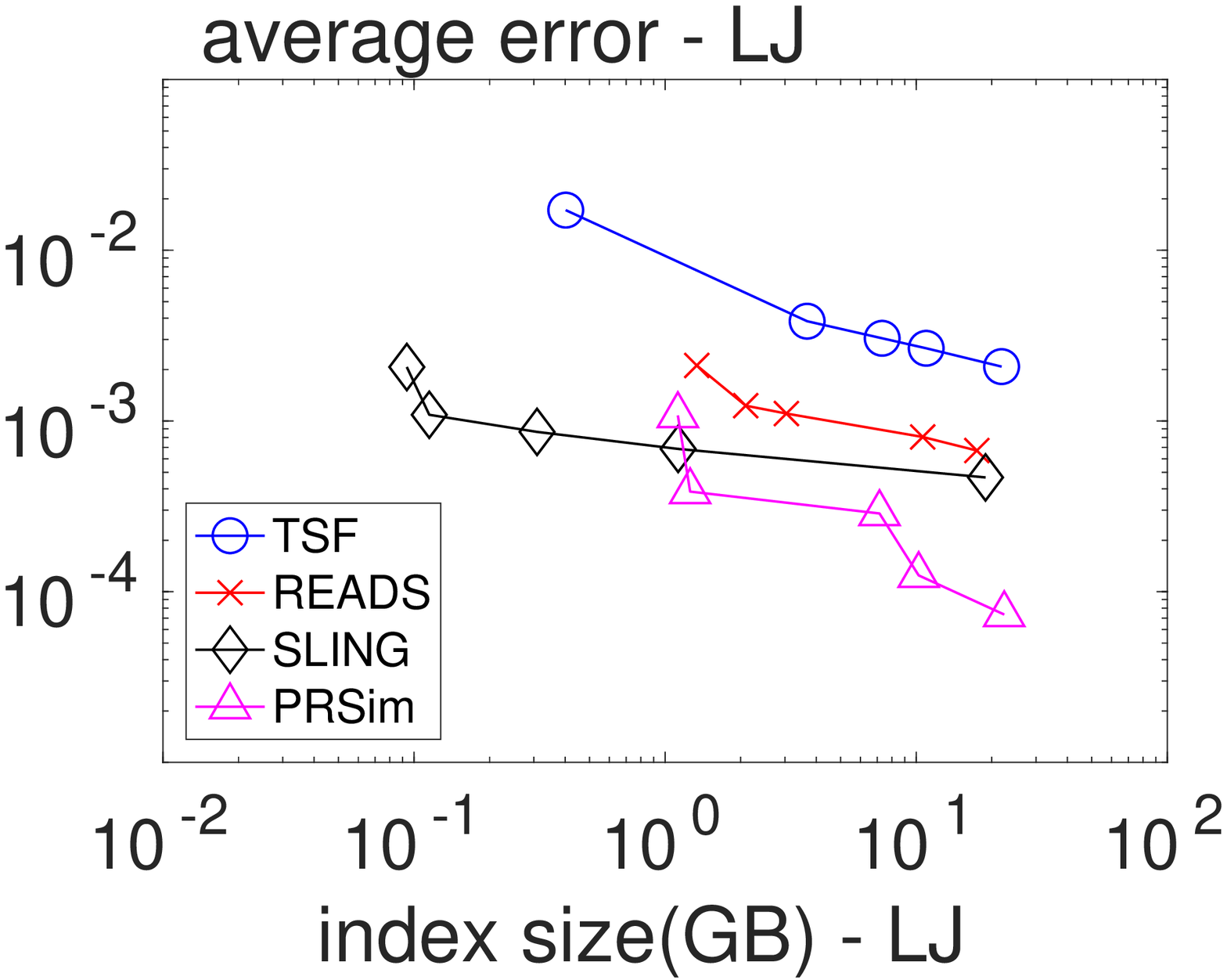}
      &
    \hspace{-2mm}
        \includegraphics[height=28mm]{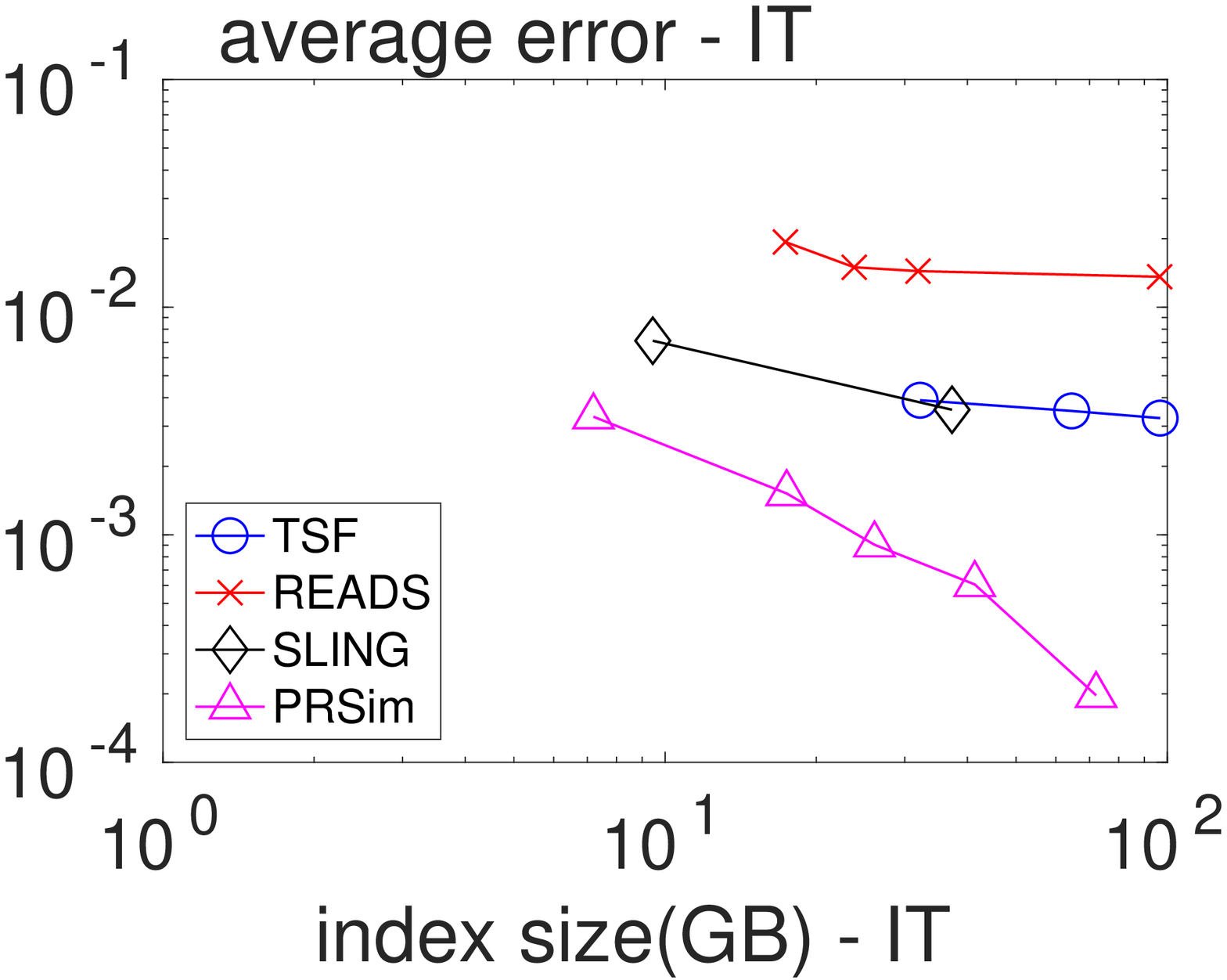}
         &
    \hspace{-2mm}
           \includegraphics[height=28mm]{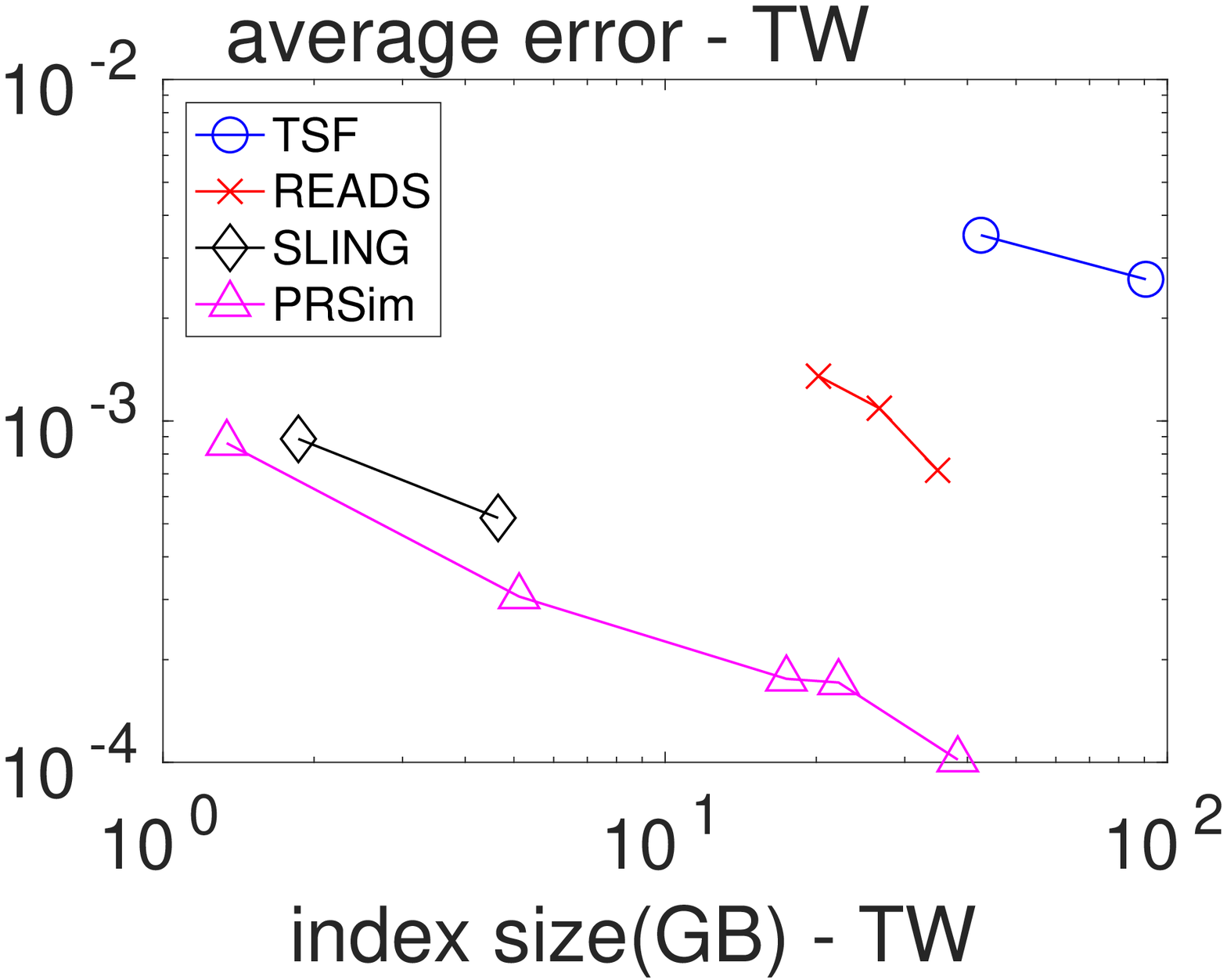}
           &
    \hspace{-2mm}
     \includegraphics[height=28mm]{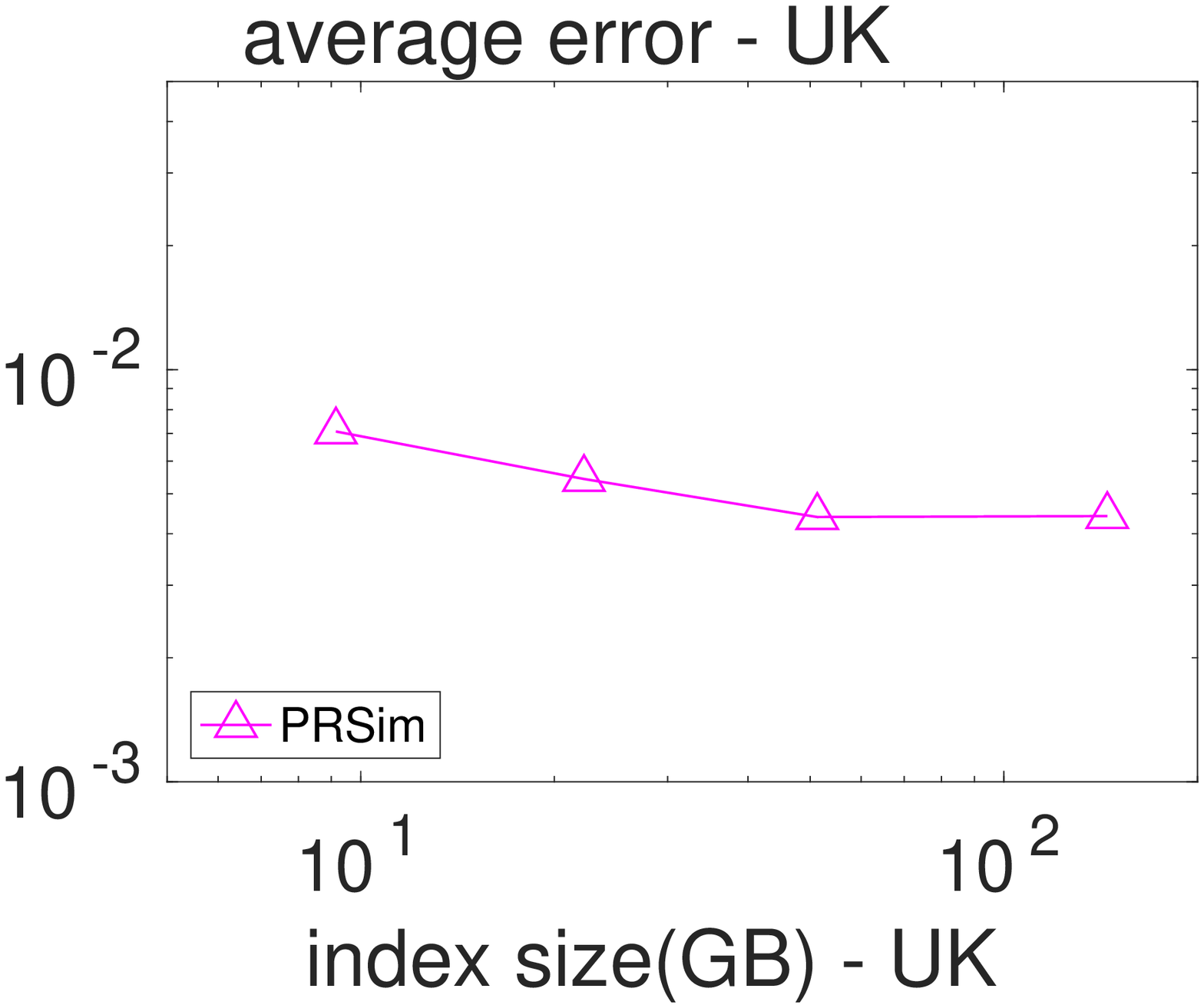}
   \vspace{0mm} \\
 \end{tabular}
\vspace{-3mm}
\caption{ $\mathbf{AvgError@50}$ v.s. {\em Index size} } \label{fig:error_index_size}
\vspace{0mm}
\end{small}
\end{figure*}

\begin{figure*}[!t]
\begin{small}
 \centering
   \vspace{0mm}
    \begin{tabular}{ccccc}
     \hspace{-2mm} \includegraphics[height=28mm]{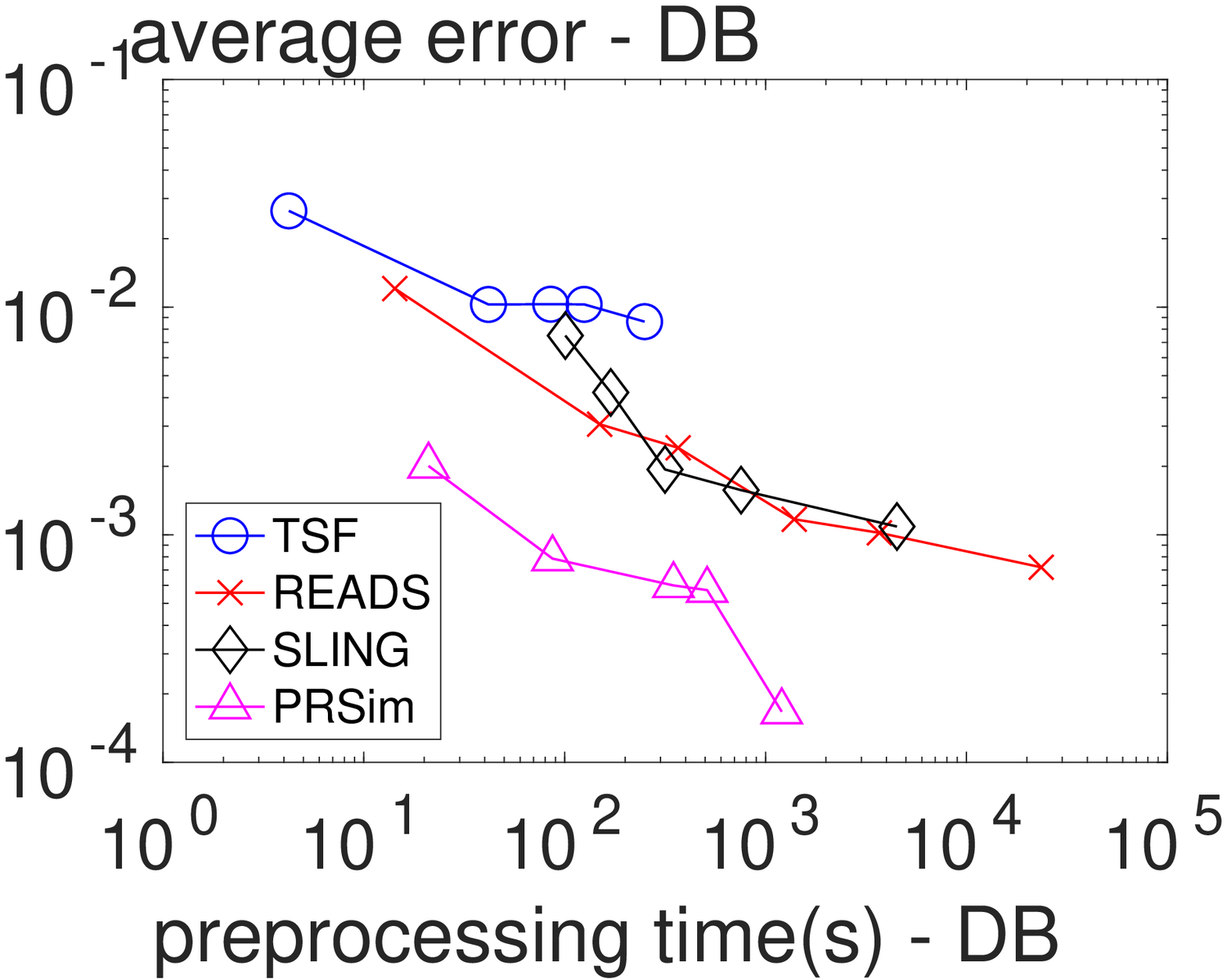}
   &
    \hspace{-2mm}
     \includegraphics[height=28mm]{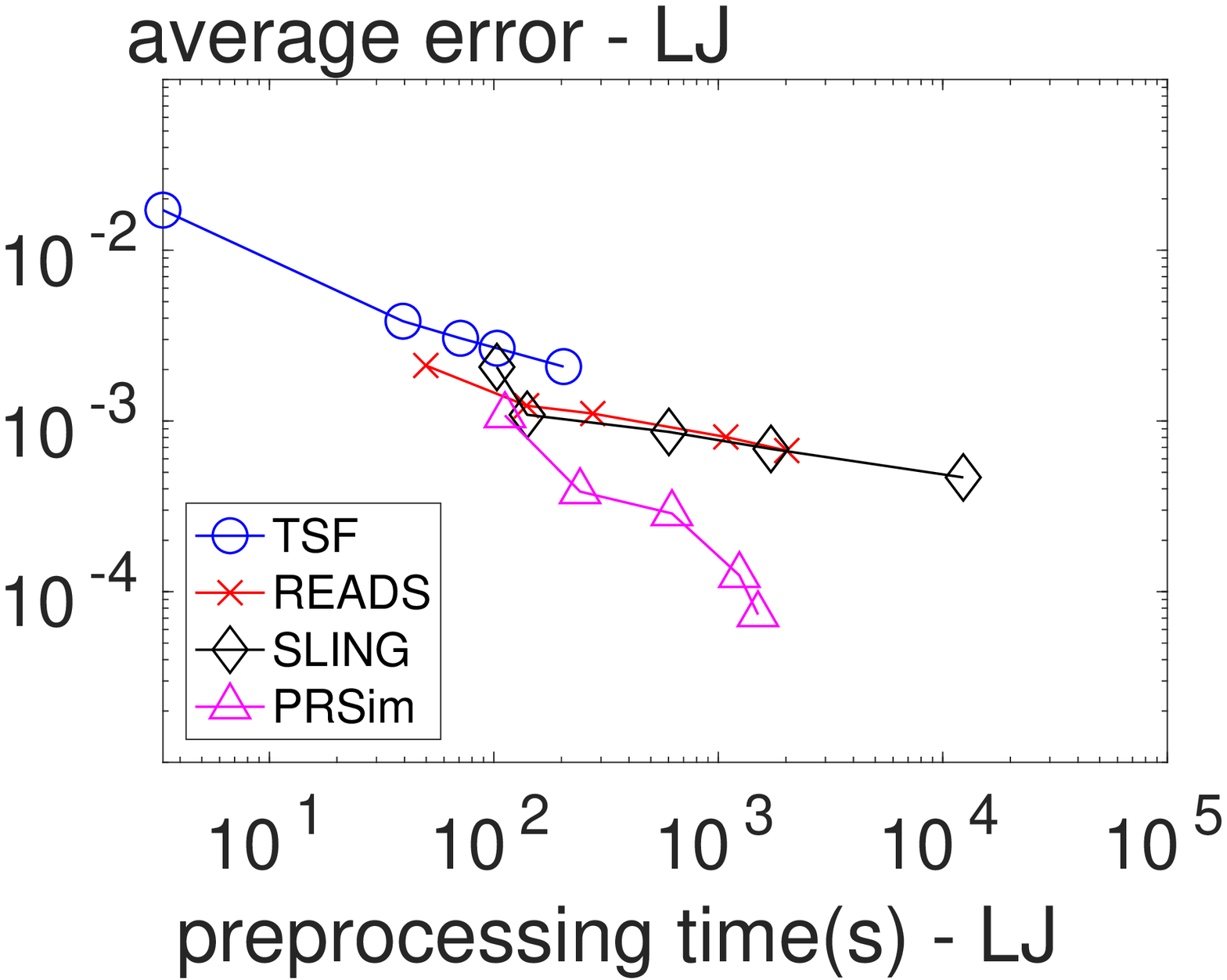}
      &
    \hspace{-2mm}
        \includegraphics[height=28mm]{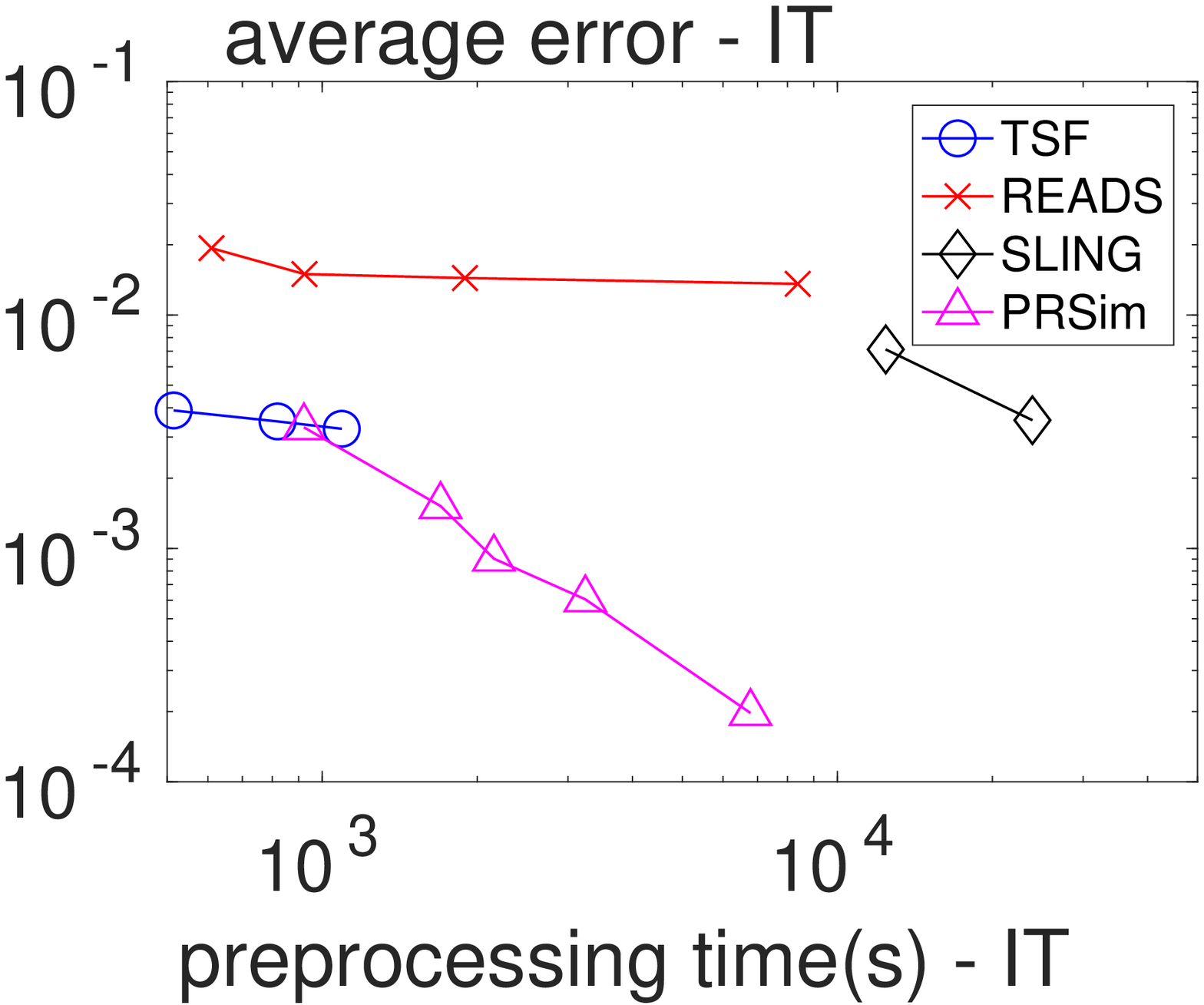}
         &
    \hspace{-2mm}
           \includegraphics[height=28mm]{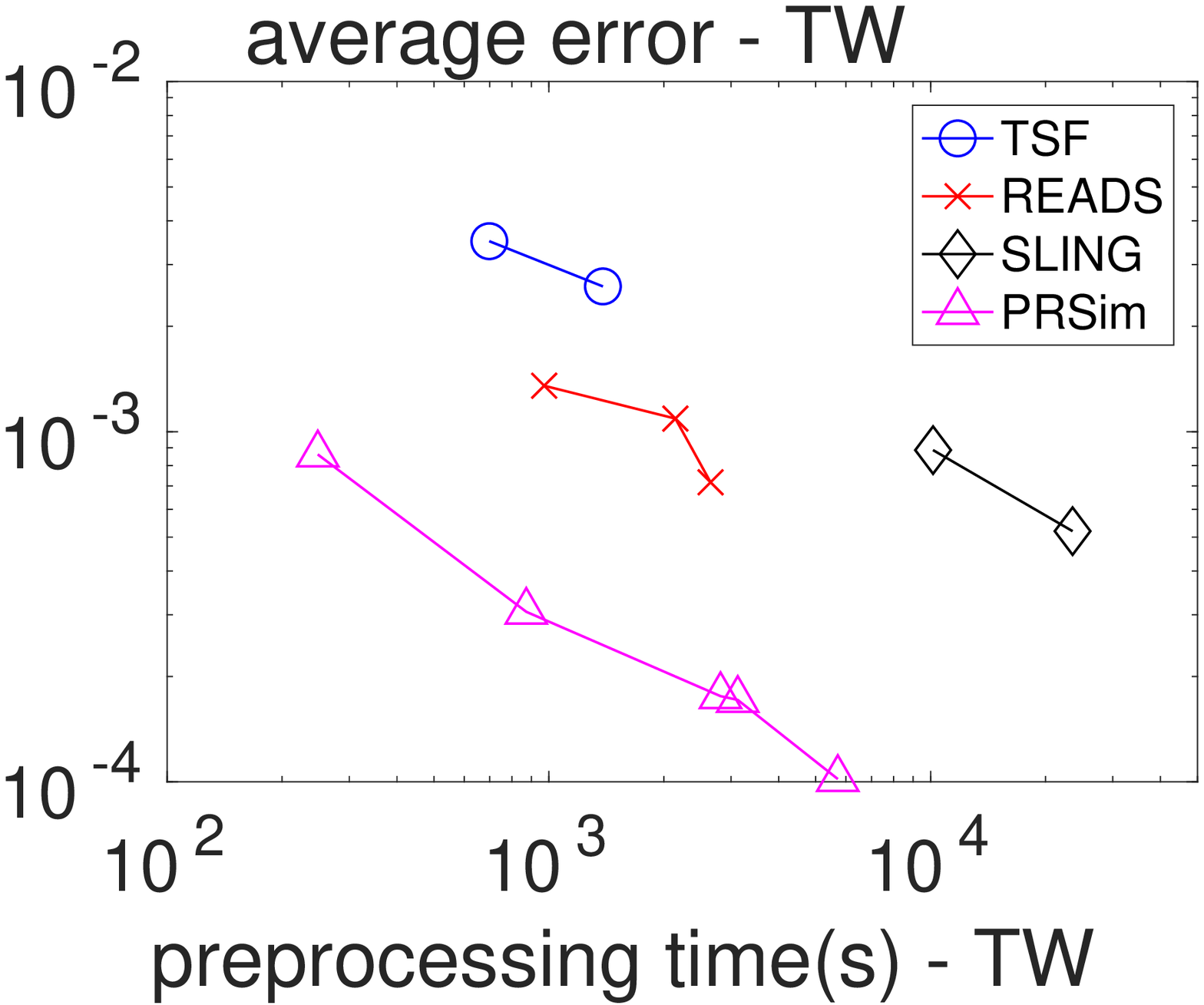}
           &
    \hspace{-2mm}
     \includegraphics[height=28mm]{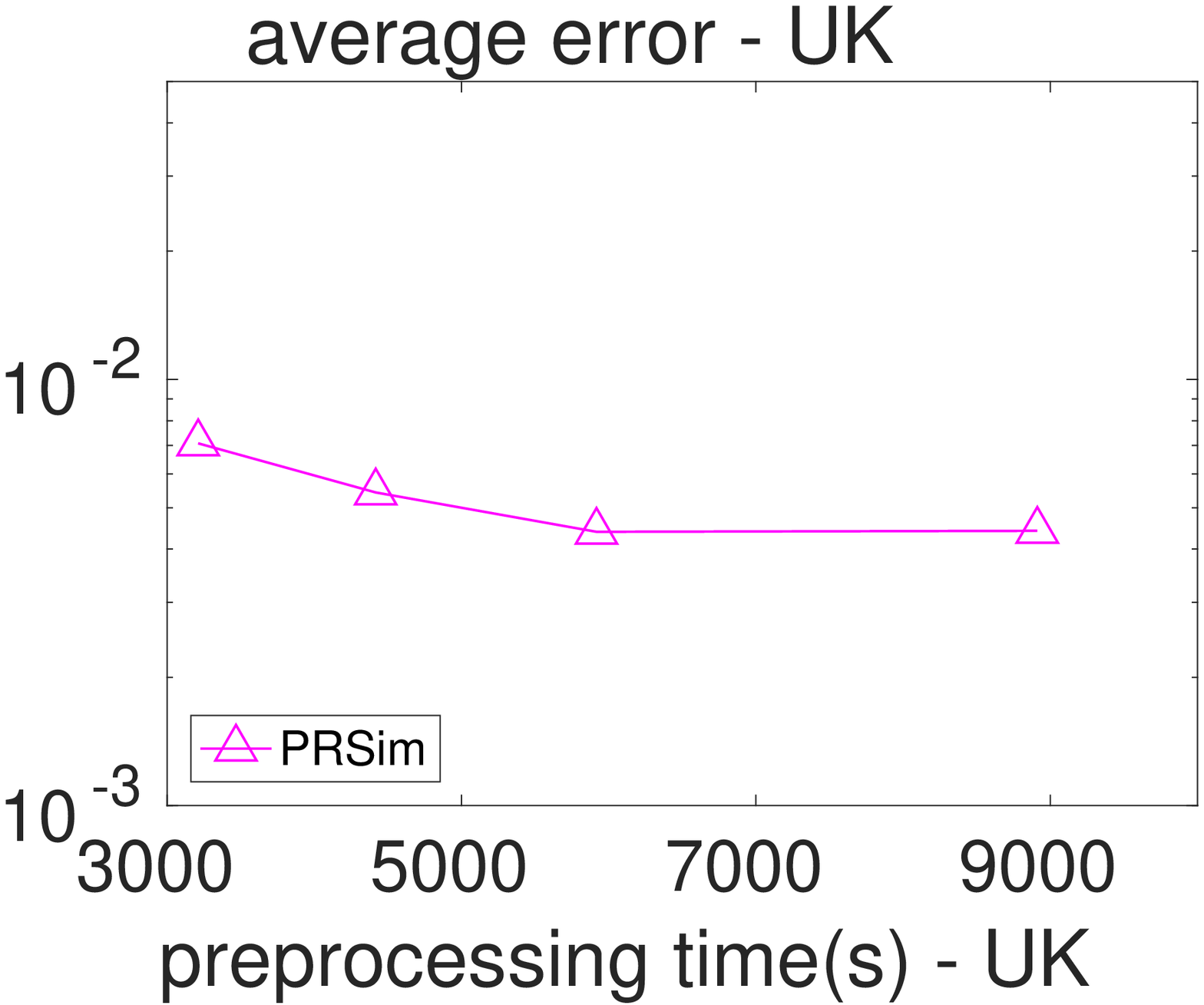}
   \vspace{0mm} \\
 \end{tabular}
\vspace{-3mm}
\caption{  $\mathbf{AvgError@50}$ v.s. {\em Preprocessing time} } \label{fig:error_index_time}
\vspace{0mm}
\end{small}
\end{figure*}

\subsection{Experiments on Real-World Graphs}
{ We evaluate the tradeoffs between accuracy and
complexity for each algorithm on real world graphs. 
We use $5$  data sets,} as shown in
Table~\ref{tbl:datasets}. All data sets are obtained from public
sources \cite{SNAP,LWA}.

\begin{table}[t]
\centering
\tblcapup
\caption{Data Sets.}
\vspace{-3mm}
\tblcapdown
\begin{small}
 \begin{tabular}{|l|l|r|r|} 
 \hline
 {\bf Data Set} & {\bf Type} & {\bf $\boldsymbol{n}$} & {\bf $\boldsymbol{m}$}	 \\ \hline
   DBLP-Author (DB) & undirected & 5,425,963 & 17,298,033 \\
 LiveJournal (LJ) &	directed	&	4,847,571	&
                                                            68,993,773
   \\

       It-2004 (IT)	&	directed & 41,291,594 & 1,150,725,436
   \\

Twitter (TW) & directed & 41,652,230& 1,468,365,182 \\
UK-Union (UK) & directed & 133,633,040 & 5,507,679,822 \\
 \hline
\end{tabular}
\end{small}
\label{tbl:datasets}
\vspace{-0mm}
\end{table}

{
\vspace{-1mm}\header
{\bf Parameters.} 
SLING~\cite{TX16} has a
parameter $\e_a$, the upper bound on the absolute error.  We vary
$\e_a$ in $\{0.5, 0.1, 0.05, 0.01, 0.005\}$, where $\e_a =0.05$ is the
default value in~\cite{TX16}.  TSF has two parameters $R_g$ and $R_q$,
where $R_g$ is the number of one-way graphs stored in 
the index,  and $R_q $ is the number of times each one-way
graph is reused in the query stage. We vary $(R_g, R_q)$ in $\{(10, 2),
(100, 20), (200, 30), (300, 40), (600, 80)\}$, where ~\cite{SLX15}
sets $(R_g, R_q)  =
(300, 40)$ by default.
TopSim has four internal parameters $T$, $h$, $\eta$
and $H$, where $T$ is  the depth of the random walks, $1/h$
is the minimal degree threshold used to identify a high degree node,
$\eta$ is the similarity threshold for trimming a random walk, and $H$ is the number of random walks to be expanded at
each level. We fix $H$ and $\eta$ to their default values $100$ and
$0.001$, and vary  $(T, 1/h)$ in $\{(1, 10), (3, 100), (3, 1000),
(3, 10000), (4, 10000)\}$. Note that \cite{LeeLY12} sets $(T,
1/h)  = (3, 100)$ by default.  The READS paper~\cite{jiang2017reads}  proposed  three algorithms: READS,
READS-D, and READS-Rq. We only include the static version of READS
in our experiments, as it is the fastest among the three~\cite{jiang2017reads}.
READS has two parameters $r$
and $t$, where $r$ is the number of $\scw$-walks generated for
each node in the preprocessing stage and $t$ is the maximum depth of the
$\scw$-walks. We vary $(r, t)$ in $\{(10, 2), (50, 5), (100, 10), (500,
10), (1000, 20)\}$, where $(r, t) \\= (100, 10)$ is the default setting
in~\cite{jiang2017reads}. For ProbeSim~\cite{liu2017probesim}, we vary
the error parameter $\e_a$ in $\{0.5, 0.1, 0.05, 0.01, 0.005\}$, where $\e_a = 0.1$ is the default
setting in~\cite{liu2017probesim}. For \prsim, we vary $\e$ in  $\{0.5,
 0.1, 0.05, 0.01, 0.005\}$. We also set
 $j_0$ to $\sqrt{n}$ so that the index size of \prsim increases with $1/\e$. We
 fix the failure  probability $\delta = 0.0001$ unless otherwise specified. We set the decay factor $c$ of
SimRank to 0.6, following previous work 
\cite{MKK14,Yu13,LVGT10,YuM15a,YuM15b}.

\header{\bf Experimental results.}
 On each data set, we
issue 100 single-source queries and 100 top-$50$ queries for each
algorithm and each parameter set,
and record the averages of the query time, index sizes, preprocessing
time,  {\em AvgError@50} and  {\em
  Precision@50}. 
For each algorithm and each dataset, we omit a parameter set if it
runs out of 196GB memory or 
takes over 10 hours to finish queries or preprocessing on that data
set.

Figures~\ref{fig:error_query_time},~\ref{fig:precision_query_time}
show the tradeoffs between {\em AvgError@50}  and the query
time and the tradeoffs between {\em Precision@50} and the query
time. The overall observation is that \prsim outperforms all
competitors by achieving lower errors and higher precisions with less query time on all datasets. Most notably, on the TW dataset,  \prsim achieves a {\em Precision@50} of $92\%$ using a query time
 of $5$ seconds, while the closest competitor, ProbeSim, achieves a precision
 around $75\%$ using over 50 seconds.  Furthermore, on the
 5-billion-edge  UK data set,  \prsim is  the only two index-based algorithms that are able to
 finish preprocessing and queries, which demonstrates the scalability
 of our algorithms. We also note that the query time of SLING and
 READS are not sensitive to the choices of parameters. This is as
 expected, since the majority of their query cost is spent on  reading the index, which
 is a cache-friendly task. After observing the skewed trend of READS
 on DB in Figure~\ref{fig:error_query_time}, we decide to evaluate an extra parameter set $(r,
 t) = (5000, 20)$ to see if READS can outperform \prsim in
 terms of query-time-error tradeoff, given significantly more indexing
 space. The result shows that  \prsim still achieves better accuracy
 with less query time. 

Figure~\ref{fig:error_index_size},
and ~\ref{fig:error_index_time} show the tradeoffs between{\em AvgError@50}  and the
index size and the tradeoffs between  {\em AvgError@50}  and the
preprocessing time, respectively. Again, our algorithm manages to
outperform all index-based algorithms (SLING, TSF, READS) by achieving a
lower error with less index size and preprocessing time.  In
particular, on the DB dataset, our algorithm is able to achieve an
average error of
$10^{-3}$ using an index of size $200MB$, while the  closest
competitor READS needs $100GB$. 
}




\vspace{-2mm}\subsection{Experiments on Synthetic Data Sets}
\label{sec:exp_synthetic}

\begin{figure}[!t]
\begin{small}
 \centering
   \vspace{0mm}
    \begin{tabular}{cc}
     \hspace{-4mm} \includegraphics[height=30mm]{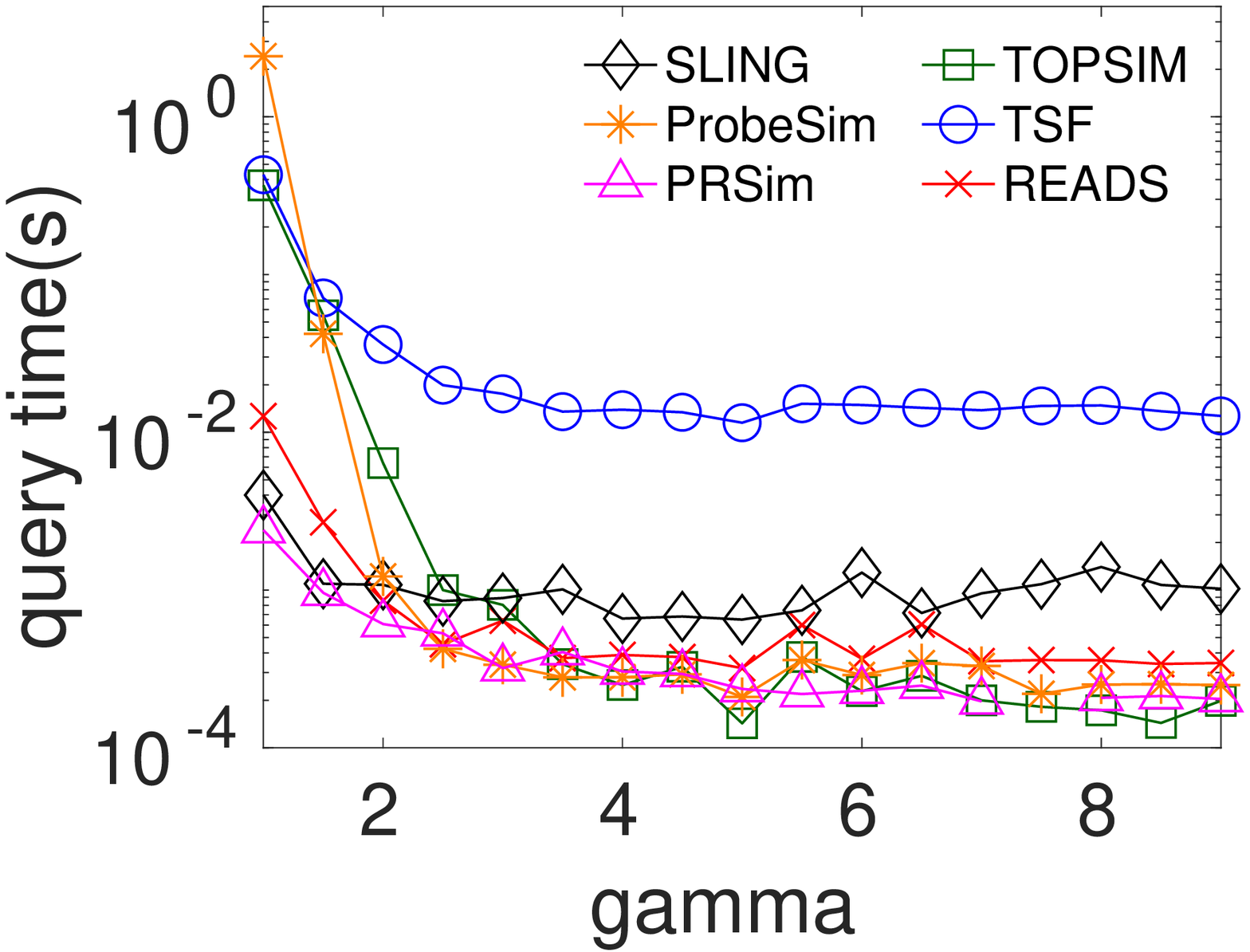}
   &
    \hspace{-4mm} \includegraphics[height=30mm]{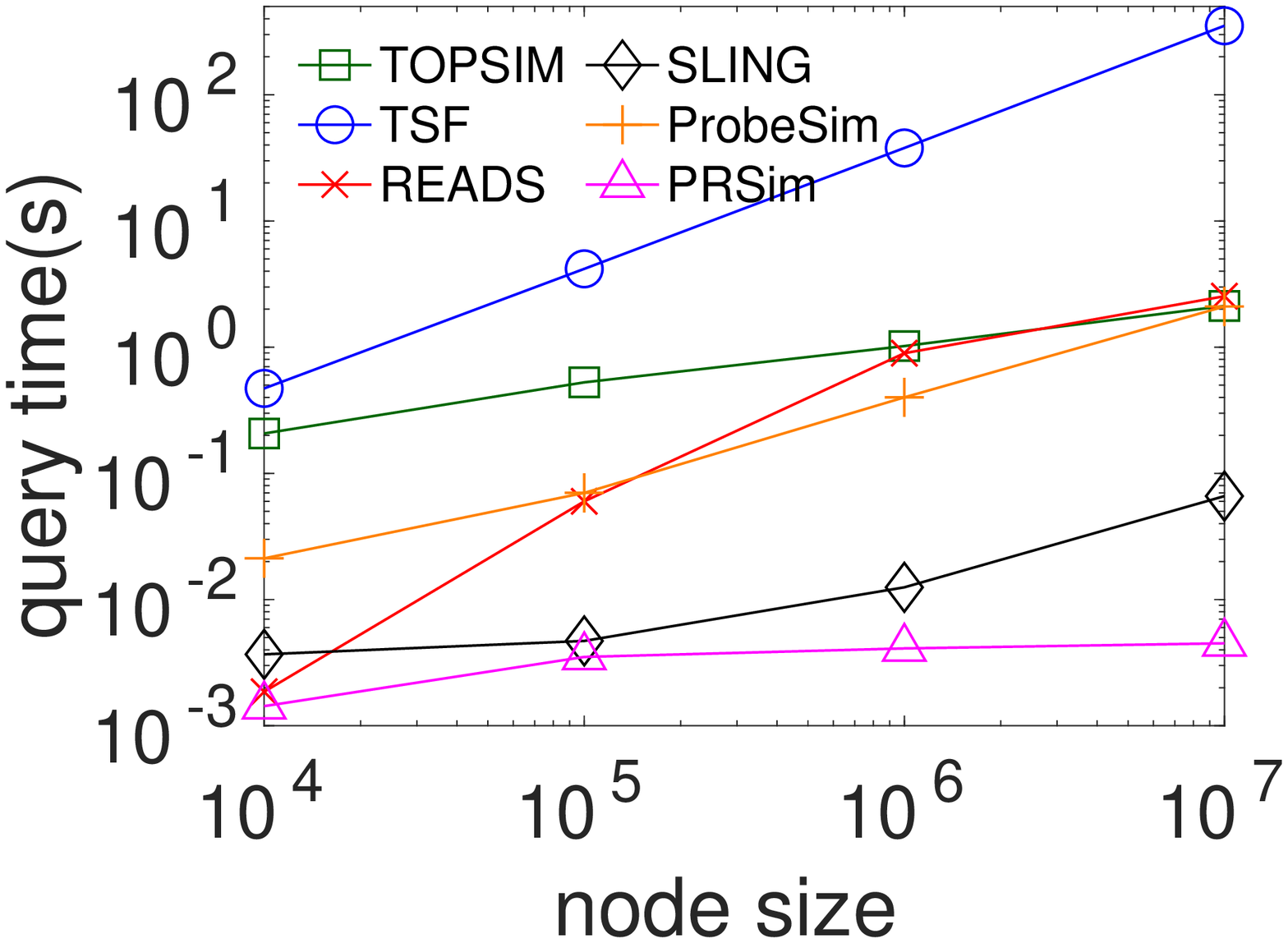}
   \vspace{0mm} \\
       \hspace{0mm} (a) { {\em Query time }  v.s. $\gamma$.}  &
       \hspace{0mm} (b) { {\em Query time } v.s. $n$.} \\
 \end{tabular}
\vspace{-3mm}
 \caption{ {Results on power-law graphs.}} \label{fig:scalability}
\vspace{-4mm}
\end{small}
\end{figure}

\begin{figure}[!t]
\begin{small}
 \centering
   \vspace{0mm}
    \begin{tabular}{cc}
     \hspace{-4mm} \includegraphics[height=30mm]{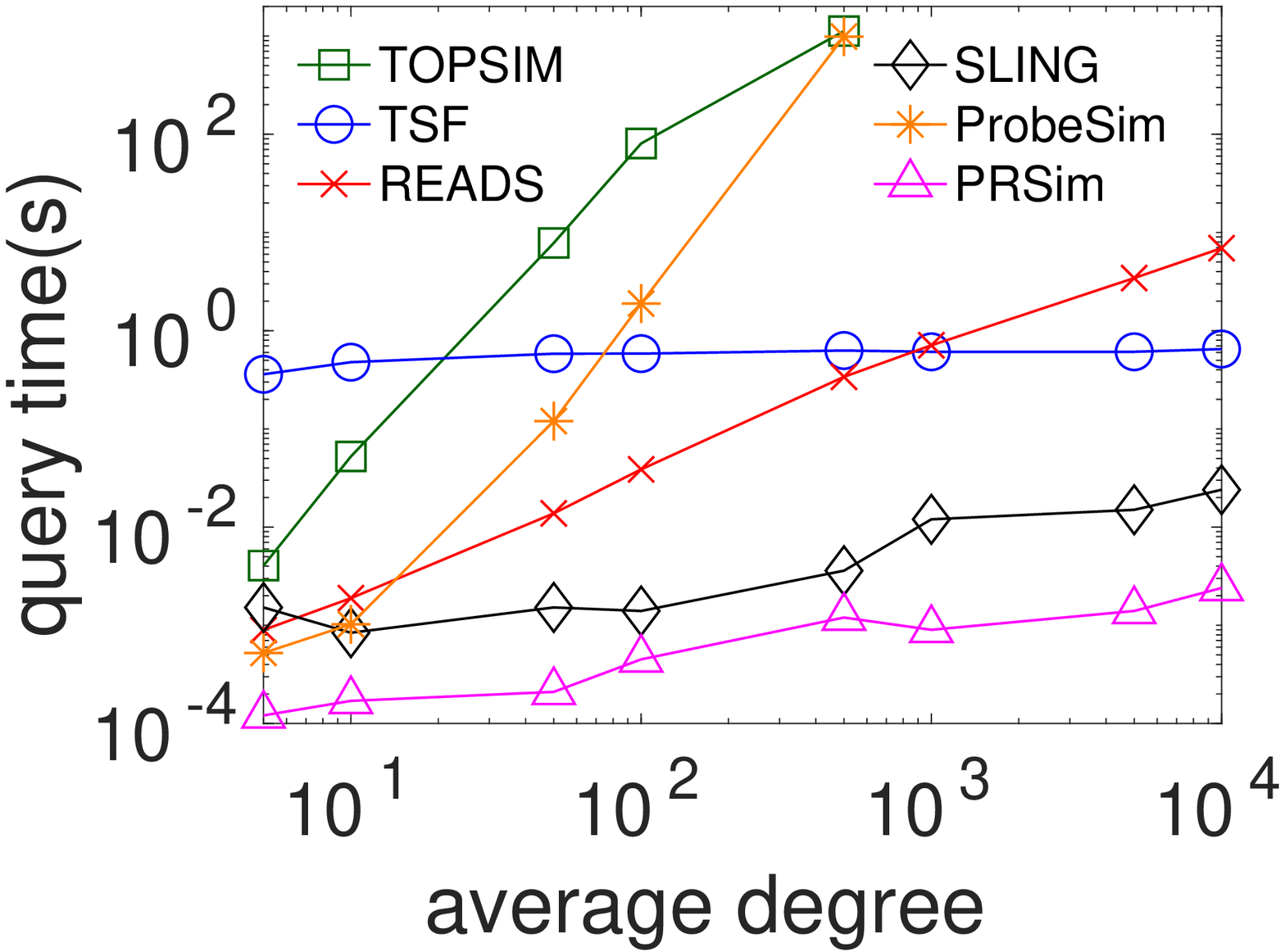}
   &
    \hspace{-4mm} \includegraphics[height=31mm]{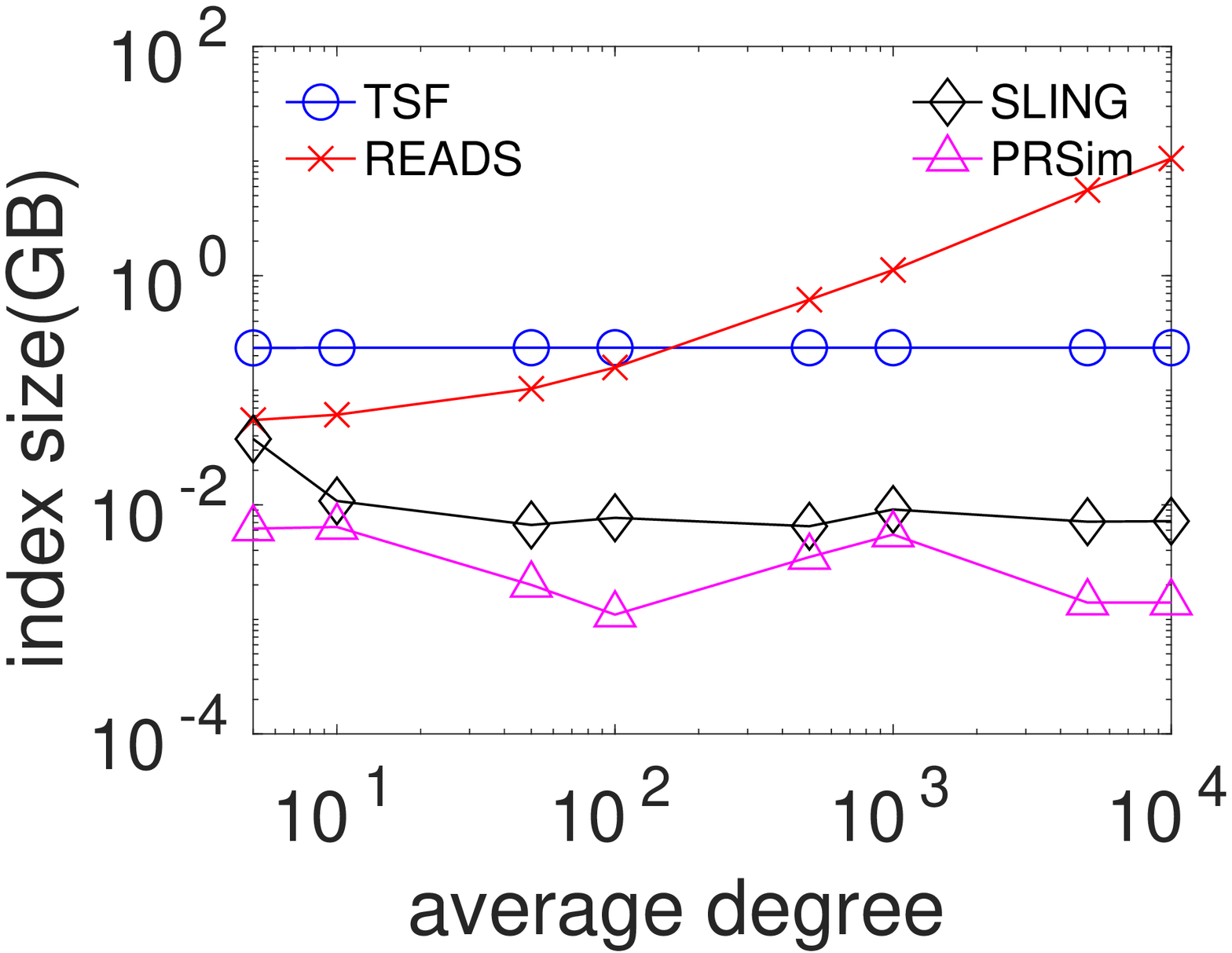}
   \vspace{-2mm} \\
      \hspace{0mm} (a) { {\em Query time }  v.s. $\bar{d}$.}  &
                                                                \hspace{0mm}
                                                                (b) {
                                                                {\em
                                                                Index Size} v.s. $\bar{d}$.} \\
 \end{tabular}
\vspace{-3mm}
 \caption{ {Results on non-power-law graphs.}} \label{fig:er}
\vspace{-4mm}
\end{small}
\end{figure}

{
We now evaluate \prsim and the competitors with fixed parameters on synthetic
datasets with varying network structure and sizes. 
We set
$\e_a = 0.25$  for SLING, $R_g = 300$ and $R_q=40$ for TSF, $T=3$,
$1/h = 100$, $\eta = 0.001$, and $H=100$ for TopSim, $\e_a =
0.25$ for ProbeSim, $r =100$ and $t=10$ for READS, and $\e = 0.25$ for \prsim. We
 fix the failure probability $\delta = 0.001$ unless otherwise specified. On each data set, we
issue 100 single-source queries with each algorithm to be evaluated,
and report the corresponding measures.
}

\header{\bf Hardness of SimRank computation and degree distributions.}
We first investigate the relation between the hardness of
SimRank computation and 
degree distributions. We generate a
set of undirected
power-law graphs with various power-law exponents using the hyperbolic
graph generator~\cite{aldecoa2015hyperbolic}. In particular,
we fix the number of nodes $n$ to be $100,000$ and the
average degree $\bar{d}$ to be
$10$, and vary the degree power-law exponent $\gamma$ from
{ $1$ to $9$}. Figure~\ref{fig:scalability}(a)
reports the average
query time of each algorithm. {
Recall that the theoretical analysis of
\prsim suggests that its query time
increases with $1/\gamma$. Figure~\ref{fig:scalability}(a) concurs with
this analysis. In fact, we observe that the query time of all
algorithms follows a similar distribution as the function $y =
1/\gamma$ on the log-log plot:
the query time decreases as 
we increase $\gamma$ from $1$ to $4$, and becomes stable after $\gamma
> 4$.} Based on this observation and on
the theoretical analysis for \prsim,  we make the following
conjuncture:
{
  \newtheorem{conjuncture}{Conjuncture}
  \vspace{-1mm}
\begin{conjuncture}
The hardness of SimRank computation is 
correlated to  the reciprocal of the power-law
exponent $\gamma$ of the out-degree distribution.
\end{conjuncture}
}

{  \vspace{-2mm}
    \header{\bf Scalability analysis.} To evaluate the scalability of our
  algorithm, we generate synthetic power-law
  graphs by fixing the exponent $\gamma = 3$ and average degree
  $\bar{d} = 10$, and vary the graph size
  $n$ from $10^4$ to $10^7$.  Figure~\ref{fig:scalability}(b) shows the
  running time of \prsim on these graphs. The results show that the
  running time of \prsim forms a concave curve in a log-log plot, which
  proves the sub-linearity of \prsim.
  
  \header{\bf Experiments on non-power-law Graphs.} We generate random
  graphs using  the Er\H{o}s and R\'{e}nyi (ER)
  model, where we assign an edge to each node pair with a user-specified probability $p$. We fix the number of nodes to $n=10,000$
  and set the value of $p$ so that the
  average degree $\bar{d}$ of each graph varies from $5$ to
  $10,000$. Figure~\ref{fig:er} shows the query time of each algorithm
  on these synthetic graphs. We observe that the query performance of
  ProbeSim degrades dramatically as we increase $\bar{d}$. On the
  other hand, \prsim is able to answer queries on very dense graphs
  efficiently. We attribute this quality to the fact that the
  Randomized Probe algorithm in ProbeSim
  always  goes through all out-neighbors of a target node, while our
  Variance Bounded Backward Walk algorithm only needs to visit a
  fraction of the out-neighbors.

}

{



}


%% file: conclusions.tex
\vspace{-1mm}
\section{Conclusions} \label{sec:conclusions}
This paper presents \prsim, an algorithm for single-source SimRank
queries. { \prsim connects the time complexity of SimRank  computation
with the distribution of the reverse PageRank,} and
achieves sublinear query time on power-law graphs with small index
size. Our experiments show that the algorithm significantly
outperforms the existing methods  in terms of query time, accuracy,
index size and scalability.



%% file: acknowledgements.tex
\vspace{-1mm}
\section{ACKNOWLEDGEMENTS}
This research was supported in part by National Natural Science Foundation of China
      (No. 61832017 and No. 61732014), by MOE, Singapore under grant
      MOE2015-T2-2-069, and by NUS, Singapore under an SUG. Sibo Wang was supported by CUHK Direct Grant No. 4055114. He was also supported by the  CUHK University Startup Grant No. 4930911 and No. 5501570.


%% file: appendix.tex
\vspace{-2mm}
\section{Inequalities}
\vspace{-2mm}
\subsection{Chernoff Bound} \label{sec:chernoff}

\vspace{-1mm}\begin{lemma}[Chernoff Bound \cite{ChungL06}] \label{lmm:chernoff}
For a set $\{x_i\}$ ($i \in [1, n_r]$) of i.i.d.\ random variables with mean $\mu$ and $x_i \in [0, 1]$,
$\Pr\left[\left|{1\over n_r}\sum_{i=1}^{n_x} x_i - \mu\right| \geq \e\right] \leq \exp\left(-\dfrac{n_r \cdot \e^2}{\frac{2}{3}\e + 2\mu}\right).\vspace{-1mm}$
\end{lemma}
\vspace{-2mm}
\subsection{Chebyshev's Inequality} \label{sec:chebyshev}

\vspace{-1mm}\begin{lemma}[Chebyshev's inequality] \label{lmm:chebysev}
  Let $X$ be a random variable, then $\Pr\left[\left| X -E[X]\right| \geq \e\right] \le {\Var[X] \over \e^2 }. $
\end{lemma}
\vspace{-2mm}
\subsection{Median Trick} \label{sec:median-of-mean}
\vspace{-1mm}\begin{lemma}[\cite{charikar2002finding}]\label{lmm:median}
  Let $X_1, \ldots, X_k$ be $k \ge 3\log {1\over \delta} $ i.i.d. random variables, such
  that $\Pr\left[\left| X_i -E[X_i]\right| \geq \e\right] \le {1\over 3}$.
  Let $X =\textrm{Median}_{1 \le i \le k}X_i$, then
  $\Pr\left[\left| X -E[X]\right| \geq \e\right] \le \delta$.
\end{lemma}
\vspace{-2mm}
\subsection{Partial sum of  Riemann zeta function }\label{sec:rz}
\vspace{-1mm}\begin{lemma}
  \label{lem:rz} The partial sum of  Riemann zeta function satisfies the following
  property:
   \vspace{-1mm}\vspace{-1mm}\begin{equation}
\label{eqn:query_index_free}
\sum_{k=i+1}^j k^{-\alpha}  = \left\{
\vspace{-1mm}\begin{array}{ll}
O(j^{1-\alpha}), &\textrm{for } \alpha < 1; \\
  O(\log j - \log i) , & \textrm{for }\alpha = 1; \\
O\left(i^{1-\alpha}\right) , & \textrm{for }\alpha > 1.
\end{array}\right.
\end{equation}
\end{lemma}


\section{Proofs} \label{sec:proofs}
\vspace{-2mm}
\subsection{Proof of Lemma~\ref{lem:index_size}}
\vspace{-1mm}\begin{proof}
 Let $w_1, \ldots, w_n$ be the nodes of the graph sorted in descending
 order of the reverse PageRank value $\pi(w_j)$. Let $size(w_j)$ denote index
 size for node $w_j$. Then, $size = \sum_{j=1}^{j_0} size(w_j)$ is the
 total size of the index. For each $w_j$, recall that
 Algorithm~\ref{alg:indexing} uses backward search to find node $x$
 and level $\ell$ with $\ell$-hop RPPR $\pi_\ell(x, w) \ge \e$, and
 record the tuple $(x, \ell, \pi_\ell(x, w) )$. Hence, the
 space usage $size(w_j)$ is bounded by the total number of
 pairs $(x, \ell)$ with $\ell$-hop RPPR $\pi_\ell(x, w) \ge \e$, i.e.,
 $size(w_j) \le \sum_{\ell=0}^{\infty}\sum_{x \in V}I(\pi_\ell(x, w)
 \ge \e),$
 where $I(\pi_\ell(x, w) \ge \e)$ is an indicating function such that
 $I(\pi_\ell(x, w) \ge \e) = 1 $ if $\pi_\ell(x, w) \ge \e$ and
 $I(\pi_\ell(x, w) \ge \e) = 0 $ otherwise. We observe that
 $I(\pi_\ell(x, w) \ge \e) \le {\pi_\ell(x, w) \over \e}$, and thus
$
   size(w_j) \le  \sum_{\ell=0}^{\infty}\sum_{x \in V}{\pi_\ell(x, w)
   \over \e} = {n \pi(w_j) \over \e}.$
\end{proof}

\subsection{Proof of Lemma~\ref{lem:vbbw_ubias}}
\header{\bf Notations.} 
We begin by defining two types of random variables.
Consider a node $y$ at level $i+1$ and a node $x \in \inN(y)$. For
ease of presentation, we let $A$  denote the set of $x \in \inN(y)$ such that $\epi(x,
w)  >
\din(y)(1-\scw)$ and $B$ denote the set of $x \in \inN(y)$  such
that $\epi(x,
w)  \le
\din(y)(1-\scw)$. We use $R_{i}(x)$ to denote the random variable indicating that the random
number $r_0 < \scw$. For each $x \in B$, we define random variable $Z_i(x,y) = 1$ if random
number $r \le {\epi_i(x, w) \over \din(y) (1-\scw)}$, and $Z_i(x,y) =
0$ otherwise.  Recall that for a node $x \in A$, we increment
$\epi_{i+1} (y, w)$ by ${\epi_i(x, w) \over \din(y)}$ if and only if
$R_i(x) =1$; for a node $x \in B$, we increment
$\epi_{i+1} (y, w)$ by $1-\scw$ if and only if
$R_i(x) =1$ and $Z_i(x, y) =1$. We can express $\epi_{i+1}(y, w)$ as
\vspace{-1mm}\vspace{-1mm}\begin{equation}
  \label{eqn:Z}
\epi_{i+1}(y, w) = \sum_{x \in A}R_{i}(x) {\epi_i(x, w) \over \din(y)} +
\sum_{x \in B}R_{i}(x) Z_i(x, y) (1-\scw).
\end{equation}

\vspace{-1mm}\begin{proof}[Proof of Lemma~\ref{lem:vbbw_ubias}]
We prove the lemma by induction. For the base case, we have
$\epi_0(w,w) =1-\sqrt{c} = \pi_0(w,w)$. Assume that $\E[\epi_i(x,
w)]=\pi_i(x,w)$ for any $x\in V$. For an node $y\in V$, we will
show that $\E[\epi_{i+1}(y, w)] = \pi_{i+1}(y, w)$. Conditioning on
$ \epi_{i}(x, w)$ in equation~\eqref{eqn:Z} follows
that
\vspace{-1mm}\vspace{-1mm}\begin{align*}
  &\quad E\left[\epi_{i+1}(y, w) \mid \epi_{i}(x, w), x\in V\right] \\
  &= \sum_{x \in A}E[R_{i}(x)] {\epi_i(x, w) \over \din(y)} +
\sum_{x \in B} E[R_{i}(x) Z_i(x, y)] (1-\scw).
\end{align*}
We have $E[R_i(x)] = \Pr[r_0 \le \scw] =\scw$ and \vspace{-1mm}$$E[Z_i(x, y)] =
\Pr[r < {\epi_i(x, w) \over \din(y)(1-\scw)}] = {\epi_i(x, w) \over \din(y)(1-\scw)}.\vspace{-1mm}$$ Since $R_i(x)$ and
$Z_i(x,y)$ are independent random variables, we have
$E[R_i(x)Z_i(x, y)] = {\scw \epi_i(x, w) \over \din(y)(1-\scw)}$. It
follows that
\vspace{-1mm}\begin{align*}
  &\quad E\left[\epi_{i+1}(y, w) \mid \epi_{i}(x, w), x\in V\right] \\
  &= \hspace{-1mm} \sum_{x \in A}{\scw \epi_i(x, w) \over \din(y)} +
  \hspace{-1mm}   \sum_{x \in B}{\scw \epi_i(x, w) (1-\scw) \over \din(y)(1-\scw)} =
    \hspace{-3mm} \sum_{x \in \inN(y)}\hspace{-3mm} {\scw \epi_i(x, w) \over \din(y)}.
\end{align*}
By the induction hypothesis,  we have  $E[\epi_i(x, w) ] = \pi_i(x,
w)$ for $x \in \inN(y)$, and thus
$E[\epi_{i+1}(y, w)]  = \sum_{x \in \inN(y)} {\scw \pi_i(x, w) \over \din(y)} = \pi_{i+1}(y, w)$,
which proves the lemma.
\end{proof}

\vspace{-2mm}
\subsection{Proof of Lemma~\ref{lem:vbbw_query}}
\vspace{-1mm}\begin{proof}
Let $cost_{i+1}(y)$ denote the number of times that $\epi_{i+1}(y, w)$ gets incremented
  at level $i + 1$. Note that the total cost is bounded by
  $\sum_{i=0}^\ell
  \sum_{x\in V}cost_i (x)$. A key observation is that each increment performed by
  Algorithm~\ref{alg:vbbw} adds at least $1-\scw$ to $\epi_{i+1}(y,
  w)$. To see this, note that   Algorithm~\ref{alg:vbbw}  increments $\epi_{i+1}(y,
  w)$ by ${\epi_i(x, w) \over \din(y)}$ only if $\din(y) < {\epi_i(x,
    w) \over 1-\scw}$, or equivalently ${\epi_i(x, w) \over \din(y)} >
  1-\scw$.  Therefore the number of times that $\epi_{i+1}(y,
  w)$ gets incremented is bounded by $\pi_{i+1}(y,
  w) \over (1-\scw)$, and thus the total cost is bounded by
\vspace{-1mm}$$E\left[\sum_{i=0}^\infty \sum_{x\in V}cost_i (x)\right]=   {1\over 1-\scw} \sum_{i=0}^\infty \sum_{x \in V}\pi_i (x, w)) =
O(n\pi(w)).\vspace{-1mm}$$
This proves the lemma.
   \end{proof}

\vspace{-2mm}
\subsection{Proof of Lemma~\ref{lem:variance}}
\vspace{-1mm}\begin{proof}
We will prove  $\E[\epi_\ell (x, w)^2]
\le \pi_\ell(x, w)$ by induction. For the base case, we have
$E[\epi_0(w,w)^2] =(1-\scw)^2 \le \pi_0(w,w).$
Assume that $\E[\epi_i(x,
w)^2] \le \pi_i(x, w)$ for any $x\in V$. For an node $y\in V$, we will
show that $\E[\epi_{i+1}(y, w)^2] \le \pi_{i+1}(y, w)$. Conditioning on
$ \epi_{i}(x, w)$ for all $x\in V$
\vspace{-1mm}\vspace{-1mm}\begin{align}
  & E\left[\epi_{i+1}(y, w)^2 \mid \epi_{i}(x, w), x\in V\right]=
    \nonumber \\
  & E\left[ \left(\sum_{x \in A}R_{i}(x) {\epi_i(x, w) \over \din(y)} +
\sum_{x \in B}R_{i}(x) Z_i(x, y) (1-\scw)\right)^2\right]. \label{eqn:ssq}
\end{align}
We expand equation~\eqref{eqn:ssq} into 5 terms:
\vspace{-1mm}\vspace{-1mm}\begin{align*}
  &\quad E\left[\epi_{i+1}(y, w)^2 \mid \epi_{i}(x, w), x\in V\right] = \hspace{-0.5mm}X_1\hspace{-0.5mm}+\hspace{-0.5mm}X_2\hspace{-0.5mm}+\hspace{-0.5mm}X_3\hspace{-0.5mm}+\hspace{-0.5mm}X_4\hspace{-0.5mm}+\hspace{-0.5mm}X_5
    \nonumber \\
  &= \sum_{x \in A}\hspace{-1mm}E\left[R_{i}(x)^2\right]\hspace{-1mm} {\epi_i(x, w)^2 \over
    \din(y)^2}  \hspace{-1mm}+\hspace{-1mm} \sum_{x \in B}\hspace{-1mm}E\hspace{-1mm}\left[R_{i}(x)^2 Z_i(x, y)^2\right]
    \hspace{-1mm}(1-\scw)^2 \\
  &\quad + \sum_{x_1\neq x_2 \in A}E\left[R_{i}(x_1)R_i(x_2)\right]
    {\epi_i(x_1, w) \epi_i(x_1, w)  \over
    \din(y)^2}   \\
    &\quad + \sum_{x_1\neq x_2 \in B}E\left[R_{i}(x_1) Z_i(x_1, y) R_{i}(x_2) Z_i(x_2, y)\right]
      \cdot (1-\scw)^2 \\
     &\quad + \sum_{x_1 \in A,  x_2 \in B}E\left[R_{i}(x_1) R_{i}(x_2)
       Z_i(x_2, y)\right]{\epi_i(x_1, w) \over \din(y)} \cdot  (1-\scw).
\end{align*}
We use $X_1, X_2, X_3, X_4$ and $X_5$ to denote these 5 terms, and
calculate them individually. Since $E\left[R_{i}(x)^2\right]  =E\left[R_{i}(x)\right] =\scw$,
we have $X_1 = \sum_{x \in A} {\scw \epi_i(x, w)^2 \over
  \din(y)^2}.$  Using the induction hypothesis, we have $\E[\epi_i(x,
w)^2] \le \pi(x, w)^2$, and thus
\vspace{-1mm}\vspace{-1mm}\begin{equation}
  \label{eqn:X1}
\E[X_1]\le \sum_{x \in A} {\scw \pi_i(x, w) \over
  \din(y)^2}= {1\over \din(y)} \sum_{x \in A} {\scw \pi_i(x, w)\over
  \din(y)} = {S_A \over \din(y)},
\end{equation}
where  $S_A = \sum_{x \in A} {\scw  \pi_i(x, w) \over
  \din(y)}$. Since $\E[\epi_i(x, w)] =
  \pi_i(x, w)$, and $E\left[R_{i}(x)^2 Z_i(x, y)^2\right] =
  {\scw \epi_i(x, w) \over
    \din(y)(1-\scw)}$, we have
\vspace{-1mm}\vspace{-1mm}\begin{equation}
  \label{eqn:X2}
\E[X_2] = (1-\scw) \sum_{x \in B} {\scw  \pi_i(x, w) \over
  \din(y)} = (1-\scw) S_B.
\end{equation}
Here we define $S_B = \sum_{x \in B} {\scw  \pi_i(x, w) \over
  \din(y)}$. Note that $S_A+S_B = \sum_{x \in \inN(y)} {\scw  \pi_i(x, w) \over
  \din(y)} = \pi_{i+1}(y,w).$

  By the independence of $R_{i}(x_1), Z_i(x_1, y),
  R_{i}(x_2), Z_i(x_2, y)$ for $x_1 \neq x_2$, we have $X_3 =
  \sum_{x_1\neq x_2 \in A}  {c \epi_i(x_1, w) \epi_i(x_2, w)  \over
    \din(y)^2} $, $X_4 =
  \sum_{x_1\neq x_2 \in B}  {c \epi_i(x_1, w) \epi_i(x_2, w)  \over
    \din(y)^2} $,  $X_5 =
  \sum_{x_1\in A, x_2 \in B}  {c \epi_i(x_1, w) \epi_i(x_2, w)  \over
    \din(y)^2} $. Therefore,  $X_3 + X_4 +X_5$ can be expressed as \vspace{-1mm}$$
  X_3 + X_4 +X_5 = \sum_{x_1 \neq x_2 \in \inN(y)}  {c \over
    \din(y)^2} \cdot \epi_i(x_1, w) \epi_i(x_2, w)  .\vspace{-1mm}$$ Using the inequality that $  \epi_i(x_1, w) \epi_i(x_2, w)
  \le {1\over 2}\epi_i(x_1, w)^2+  {1\over 2}\epi_i(x_1, w)^2, $ and we have
\vspace{-1mm}\vspace{-1mm}\begin{align*}
  X_3 + X_4 +X_5 &\le   \sum_{x_1 \neq x_2 \in \inN(y)} {c \over
                   2 \din(y)^2} \left(\epi_i(x_1, w)^2+  \epi_i(x_2, w)^2\right) \\
                 &= \sum_{x \in \inN(y)}{c \left(\din(y)-1\right)  \over
  \din(y)^2} \epi_i(x, w)^2.
  \end{align*}
The last equation is due to the fact that each $ \epi_i(x, w)^2$
appears exactly $\din(y)-1$ times in the summation. By the induction
hypothesis that $\E[\epi_i(x, w)^2] \le \pi_i(x, w)$, we have
\vspace{-1mm}\vspace{-1mm}\begin{align}
&\quad \E[X_3+X_4+X_5] \le    \scw\left(1-{1\over
  \din(y)}\right) \sum_{x \in \inN(y)}{\scw \pi_i(x, w)   \over
                  \din(y)} \nonumber \\
                &= \scw\left(1-{1\over
  \din(y)}\right) (S_A+  S_B). \label{eqn:X35}
\end{align}
Combining Equations~\eqref{eqn:X1}-\eqref{eqn:X35}, it
follows that
\vspace{-1mm}\vspace{-1mm}\begin{align*}
 & \quad  \E\left[\epi_{i+1}(y, w)^2\right] \le \left(\scw + {1-\scw \over \din(y)}\right) S_A + \left(1- {\scw \over
    \din(y)}\right) S_B \\
  &=\left(1-\left(1-\scw\right)\cdot  \left(1-{1\over \din(y)}\right)\right) S_A + \left(1- {\scw \over
    \din(y)}\right) S_B \\
  &\le S_A+ S_B = \pi_{i+1}(y, w).
  \end{align*}
And the lemma follows.
\end{proof}

\vspace{-5mm}\subsection{Proof of Lemma~\ref{lem:error_I}}
\vspace{-1mm}\begin{proof}
  Recall that for $s_I(u, v)$,we have the estimator
  \vspace{-1mm}$$ \s_I(u, v) = {1\over (1-\sqrt{c})^2} \sum_{\ell=0}^{\infty}\sum_{j =1}^{j_0}
    \heta'_\ell(u, w_j) \pib_\ell(v, w_j),\vspace{-1mm}$$
    where $\heta'_\ell(u, w_j) = \heta_\ell(u, w_j)$ if $\heta_\ell(u, w_j) >
    {(1-\scw)^2\e \over 12}$ and $\heta'_\ell(u, w_j) = 0$ if
    otherwise. $\heta_\ell(u, w_j)$ is an estimator for $\eta(w_j)\pi_\ell(u,
    w_j)$ computed by Monte Carlo approach, and $\pib_\ell(v, w_j)$ is the
    reserve computed by $\ell$-hop  backward search.
    To bound the error of $\s_I(u, v)$, we further define
          \vspace{-1mm}$$ \s^1_I(u, v) = {1\over (1-\sqrt{c})^2} \sum_{\ell=0}^{\infty}\sum_{j =1}^{j_0}
          \heta_\ell(u,w_j)\pib_\ell(v, w_j),$$\vspace{-1mm}
and
      \vspace{-1mm}$$ \s^2_I(u, v) = {1\over (1-\sqrt{c})^2} \sum_{\ell=0}^{\infty}\sum_{j =1}^{j_0}
      \heta_\ell(u,w_j)\pi_\ell(v, w_j).$$\vspace{-1mm}

First, we claim that $\s_I(u, v)$ and $\s^1_I(u,v)$ differ by at most
${\e \over 6}$. More precisely, observe that $\heta'_\ell(u, w)$ and
$\heta_\ell(u, w)$  differ by at most ${(1-\scw)^2 \e \over 6}$, and
thus
  \vspace{-1mm}\vspace{-1mm}\begin{align}
   & \left|\s_I(u, v) \hspace{-1mm}-\hspace{-1mm} \s^1_I(u, v)\right| = \sum_{\ell=0}^{\infty}\sum_{j =1}^{j_0}{
 \left|  \heta'_\ell(u,w_j) \hspace{-1mm}-\hspace{-1mm} \heta_\ell(u,w_j)  \right|\pib_\ell(v, w_j) \over
                             (1-\sqrt{c})^2} \nonumber\\
    &\le \sum_{\ell=0}^{\infty}\sum_{j =1}^{j_0} {   {(1-\scw)^2\e \over 6} \cdot 1
      \cdot  \pib_\ell(v, w_j)  \over (1-\sqrt{c})^2} = {\e\over
      6}
      \sum_{\ell=0}^{\infty}\sum_{j =1}^{n}  \pib_\ell (v, w_j) \le  {\e
      \over 6} . \label{eqn:s1}
     \end{align}
      For the last inequality, we use the fact that the reserve
      $\pib_\ell (v, w_j)$ is at most $\pi_\ell(v, w_j)$, and thus
$   \sum_{\ell=0}^{\infty}\sum_{j =1}^{n}  \pib_\ell (v, w_j)
      \\ \le    \sum_{\ell=0}^{\infty}\sum_{j =1}^{n}  \pi_\ell (v, w_j)
      =1.$

Next, we show that $\s^1_I(u, v)$ and $\s^2_I(u,v)$ differ by at most
 ${\e \over 6}$. To see this, note that by the property of backward search, we
  have $\left|\pi_\ell(v, w_j) - \pib_\ell(v, w_j) \right| \le 2 \brmax =
  {(1-\scw)^2\e \over 6}$ for a node $w_j$ in the index.  It follows that
  \vspace{-1mm}\vspace{-1mm}\begin{align}
   & \left|\s^1_I(u, v) \hspace{-1mm} -\hspace{-1mm}\s^2_I(u, v)\right| \hspace{-1mm}= \hspace{-1mm} \sum_{\ell=0}^{\infty}\sum_{j =1}^{j_0}{
  \heta_\ell(u,w_j) \left| \pib_\ell(v, w_j) \hspace{-1mm} - \hspace{-1mm}\pi_\ell(v, w_j) \right| \over
                             (1-\sqrt{c})^2} \nonumber\\
    &\hspace{-1mm} \le \hspace{-1mm}\sum_{\ell=0}^{\infty}\sum_{j =1}^{j_0} {  \heta_\ell(u,w_j) \cdot 1
      \cdot   {(1-\scw)^2\e \over 4}  \over (1-\sqrt{c})^2}\hspace{-1mm} =\hspace{-1mm} {\e\over
      6}
      \sum_{\ell=0}^{\infty}\sum_{j =1}^{n} \hspace{-1mm}  \heta_\ell(u,w_j) \le  {\e
      \over 6}.  \label{eqn:s2}
     \end{align}
      For the last inequality, recall that Algorithm~\ref{alg:main} increments $ \heta
      $  at most $n_r$ times, and each increment is
      ${1\over n_r}$.

Finally, we  show that
  $\s^2_I(u,v)$ approximates $s_I(u,v)$ with error ${\e \over 4}$ with
  target probability. Following the definition of $\heta_\ell(u, w)$, we use a
  slightly different approach to construct $\s^2(u, v)$. For the $i$-th
  iteration, we sample a node $w$ and a level $\ell$ with
  probability $\eta(w)\pi_\ell(u, w)$, and set $X_i$ to be
  ${\pi_\ell(v, w_j) \over (1-\scw)^2}$. It can be verify that
  $\s^2_I(u,v) = {1\over n_r} \sum_{i=1}^{n_r} X_i$. For each
  $X_i$,
  \vspace{-1mm}$$\E[X_i] =  \sum_{\ell=0}^{\infty}\sum_{j =1}^{j_0}\eta(w_j)
  \pi_\ell(u,w_j) {\pi_\ell(v, w_j) \over (1-\scw)^2}
   = s_I(u, v), \vspace{-1mm}$$
   and $X_i  \le \max_{\ell, v}\left\{ {\pi_\ell(v, w_j) \over
       (1-\scw)^2} \right\} \le {1\over (1-\scw)^2}$. Since $n_r =
   \Theta(\log {n\over \delta} /\e^2)$, by Chernoff
   bound,
\vspace{-1mm}\vspace{-1mm}\begin{equation}
      \label{eqn:chernoff_I}
      \Pr\left[|\s^2_I(u, v) - s_I(u, v)| > {\e\over 6}\right] \le {\delta\over 2n}.
  \end{equation}
Combining Equations~\eqref{eqn:s1}-\eqref{eqn:chernoff_I}, 
we prove the lemma.
  \end{proof}

\vspace{-3mm}\subsection{Proof of Lemma~\ref{lem:error_B}}
 \vspace{-1mm}\begin{proof}
    Consider a single
        $\scw$-walk from $u$.
        Recall that Algorithm~\ref{alg:main}
        first samples a node-level pair $(w_j, \ell)$ with probability
        $\pi_\ell(u, w_j) \eta(w_j)$. If $j > j_0$,  it performs backward walk to
        generate an unbiased estimator $\epi_\ell(v, w)$ for each $v
        \in V$, and set the estimator $\s_B(u, v)$ to be ${\epi_\ell(v, w_j)  \over
          (1 - \scw)^2 }$. It follows that
        \vspace{-1mm}$$E\left[\s_B(u, v)  \right]  \hspace{-1mm} = \hspace{-1mm} \sum_{\ell=0}^\infty
        \sum_{j=j_0+1}^n \pi_\ell(u, w_j) \eta(w_j) \cdot
        {\epi_\ell(v, w_j)  \over  (1 - \scw)^2 } \hspace{-1mm} =
        \hspace{-1mm}   s_B(u, v). \vspace{-1mm}$$
        We can bound the variance $\Var\left[\s_B(u, v)  \right]  \le
        E\left[\s_B(u, v)^2\right]$ by
        \vspace{-1mm}\vspace{-1mm}\begin{align*}
                                                       E\left[\s_B(u,
                                                       v)^2\right]  = \sum_{\ell=0}^\infty
        \sum_{j=j_0+1}^n \pi_\ell(u, w_j) \eta(w_j) \cdot
          {E\left[\epi_\ell(v, w_j)^2\right]  \over (1- \scw)^4}.
          \end{align*}
        Lemma~\ref{lem:variance} implies that $E\left[\epi_\ell(v,
          w_j)^2\right]  \le \pi_\ell(v, w_j)$, and
                \vspace{-1mm}\vspace{-1mm}\begin{align*}\Var\left[\s_B(u, v)  \right]
          \hspace{-1mm}\le \hspace{-1mm} \sum_{\ell=0}^\infty
        \sum_{j=j_0+1}^n \hspace{-2mm}\pi_\ell(u, w_j) \eta(w_j) \cdot
          {\pi_\ell(v, w_j) \over (1- \scw)^4} \hspace{-1mm}=\hspace{-1mm} {s_B(u, v)   \over (1- \scw)^2}.
                \end{align*}
                Recall that for  a fixed $i$ with $1 \le i \le f_r$, Algorithm~\ref{alg:main} repeats above
                sampling process $d_r$ time and  use
                the mean over $d_r = {12 \over (1-\scw)^2 \e^2}$ samples, denoted $\s_B^i(u, v)$, as an estimator for
                $s_B(u, v)$.  It follows that
                \vspace{-1mm}$$\Var\left[\s_B^i(u, v)  \right]  \le  {s_B(u, v)
                  \over d_r (1- \scw)^2} = {\e^2s_B(u, v) \over 12} \le {\e^2
                  \over 12}. \vspace{-1mm}$$
                By Chebyshev's inequality, we have
       \vspace{-1mm}$$\Pr\left[|\s_B^i(u, v) - s_B(u, v)| >{\e \over 2}  \right] \le
       {4 \Var\left[\s_B^i(u, v)  \right]  \over
          \e^2} \le {1 \over 3}. \vspace{-1mm}$$
        Finally, Algorithm~\ref{alg:main} use $ \s_B(u, v)= \textrm{Median}_{1\le i
        \le f_r} \s_B^i(u, v)$ as the estimator for $\s_B(u, v)$. By
      setting $f_r = 3\log {n \over \delta}$ and applying
      the Median Trick (see Lemma~\ref{lmm:median}), we have
      \vspace{-1mm}\vspace{-1mm}\begin{equation}
        \label{eqn:Chernoff_B}
      \Pr\left[|\s_B(u, v) - s_B(u, v)| >{\e \over 2}  \right] \le {\delta\over
        2n},
    \end{equation}
    and the lemma follows.
\end{proof}

\vspace{-2mm}\subsection{Proof of Lemma~\ref{lem:query_I}}
 \vspace{-1mm}\begin{proof}
Fix the source node $u$ and consider a node $w_j$ and a level
  $\ell$. Recall that we
  retrieve all nodes $v$ with $\pib_\ell (v, w_j)$ from the index
  if and only if 1) $w_j$ is in the index, that is, $j \le j_0$, and
  2) $\heta_\ell (u, w_j) \ge  {(1-\scw)^2\e \over 8} = {\e \over
    c_1}$
  Let $size_\ell(w_j) = \Theta\left({n\pi_\ell(w_j) \over \e}\right)$ denote the
upper bound for the index size of $w_j$ at level $\ell$, and
$size_\ell(w_j) = \sum_{\ell=0}^\infty size_\ell(w_j) =
\Theta\left({n\pi(w_j) \over \e}\right)$ denote the upper bound for
the index size of $w_j$. We further define $\heta(u, w_j) = \sum_{\ell =0}^\infty \heta_\ell(u, w_j)
         $. Note that $\heta(u, w_j)$ is an unbiased estimator for
         $\sum_{\ell =0}^\infty \eta(w_j) \pi_\ell(u, w_j) = \eta(w_j) \pi(u, w_j) $. We can bound the
$C_I(u)$ as
\vspace{-1mm}\vspace{-1mm}\begin{align*}
  C_I(u) &\le \sum_{\ell=0}^\infty \sum_{j=1}^{j_0}I\left (\heta_\ell(u, w_j) >
           {\e \over c_1}\right) size_\ell(w_j),
\end{align*}
where $I\left (\heta_\ell(u, w_j) >
           {\e \over c_1}\right) $ equals $1$ if  $\heta_\ell(u, w_j) >
         {\e \over c_1}$ and equals $0$ if otherwise.
         Since $\heta_\ell(u, w_j) \le \heta(u, w_j)$, we have
         $I\left (\heta_\ell(u, w_j) >
           {\e \over c_1}\right)  \le I\left (\heta(u, w_j) >
           {\e \over c_1}\right) $, and thus
\vspace{-1mm}\vspace{-1mm}\begin{align*}
  C_I(u) &\le \sum_{\ell=0}^\infty \sum_{j=1}^{j_0}I\left (\heta(u, w_j) >
           {\e \over c_1}\right) size_\ell(w_j) \\
  & = \sum_{j=1}^{j_0}I\left (\heta(u, w_j) >
           {\e \over c_1}\right) size(w_j).
\end{align*}

         We now use two different approaches to bound $C_I(u)$. First,  observe
that for a given $u$, we have $
\sum_{j=1}^{j_0}\heta(u, w_j) \le 1$, which implies that  there are
at most ${c_1 \over \e}$ node $w_j$ with $\heta(u,
w_j) \ge {\e \over c_1}$. Since $size(w_1) \ge \ldots \ge
size(w_{j_0})$, we can choose
$\pi(u, w_1) \ge \e, \ldots \pi(u, w_{c_1\over \e}) \ge \e$ to maximize the
query cost $C_I(u)$. It follows
that $C_I(u)  \le \sum_{j=1}^{c_1 \over \e} size(w_j)  \le
                        O\left(\sum_{j=1}^{c_1 \over \e} {n\pi(w_j) \over \e} \right)$
hence proves the first part of the lemma.

For the second part, note that  $I\left (\heta(u, w_j) >
           {\e \over c_1}\right)$ is bounded by $ {\heta_\ell(u, w_j)  \over
           {\e / c_1 } } $.
         It follows that

\vspace{-1mm}\vspace{-1mm}\begin{align*}
   &\quad  \E[C_I(u)]
\le c_1 \sum_{j=1}^{j_0} { \E[\heta(u, w_j)]  \over
   \e } size(w_j) \\
   &=c_1 \sum_{j=1}^{j_0} {\eta(w_j) \pi(u, w_j) \over
     \e} size(w_j) \le c_1 \sum_{j=1}^{j_0} { \pi(u, w_j) \over
     \e} size(w_j).
\end{align*}
Here we use the fact that $\heta(u, w_j)$ is an unbiased estimator for
$\eta(w_j)\pi(u, w_j)$ and that $\eta(w_j) \le 1$ .
Taking average over all nodes $u\in V$, we have
\vspace{-1mm}\vspace{-1mm}\begin{align*}
  & \quad C_I = {1\over n} \sum_{u \in V}C_I(u) \le  {c_1\over n} \sum_{u \in
        V} \sum_{j=1}^{j_0} { \pi(u, w_j) \over
     \e} size(w_j)\\
 &=  c_1 \sum_{j=1}^{j_0} { {1\over n}\sum_{u \in
        V}\pi(u, w_j) \over
     \e} size(w_j) = c_1 \sum_{j=1}^{j_0} { \pi(w_j) \over
     \e} size(w_j).
\end{align*}
By $size(w_j)  = O\left({n \pi(w_j) \over \e}\right)$, we have $  C_I  =  O\left( {n\over \e^2}\sum_{j=1}^{j_0} \pi(w_j)^2 \right),$
and the lemma follows.
\end{proof}

\vspace{-2mm}\subsection{Proof of Lemma~\ref{lem:query_B}}
\vspace{-1mm}\begin{proof}
Next, we bound $C_{B} = {1\over n} \sum_{v\in V} C_B(u)$, the average
query cost for estimating the
  $\epi_\ell(v, w)$ for each node $w$ that is not in the Index. Given
  a source node $u$, for each
  node $w_j$ with $j > j_0$, recall that we perform $\pi_\ell(u, w_j)
  n_r$ backward walk on $w_j$ to estimate $\epi_\ell(v, w), v\in
  V$. By Lemma \ref{lem:vbbw_query}, the cost of a single backward walk on
  $w_j$, regardless of the level $\ell$, can be bounded by
  $O(n\pi(w_j))$. Ignoring the big-Oh,
\vspace{-1mm}\vspace{-1mm}\begin{align*}
  \E[(u)] = \sum_{\ell=0}^\infty \sum_{j=1}^{j_0}\pi_\ell(u, w) n_r
           \cdot n\pi(w_j)  =n_rn\sum_{j=1}^{j_0}\pi(u, w) \pi(w_j).
  \end{align*}
Taking average over all nodes $u\in V$, we have

\vspace{-1mm}\vspace{-1mm}\begin{align*}
 &\quad  \E[C_B] = {1\over n} \sum_{u \in V}\E[C_B(u)] \le  {1\over n} \sum_{u \in V} n_rn \sum_{j=1}^{j_0}\pi(u, w) \pi(w_j)\\
 &=  n_rn \sum_{j=1}^{j_0} \pi(w_j)\left( \sum_{u
   \in V}\pi(u, w)\right) =  O\left({ n \log n \over \e^2}  \sum_{j=1}^{j_0}\pi(w_j)^2\right).
\end{align*}
 The last equation is due to  $\sum_{u
   \in V}\pi(u, w) = n \pi(w)$.
\end{proof}

\vspace{-2mm}\subsection{Proof of Theorem~\ref{thm:query}}
\vspace{-1mm}\begin{proof} { We use $\beta = 1/\gamma$ to simplify the
  proof. }
  Ignoring the big-Oh notation in Lemma~\ref{lem:query_I}, we have $\E[C_I]  \le {n \over
    \e}\sum_{j=1}^{c_1\over \e} \pi(w_j)$ and $\E[C_I] \le {n  \over \e^2}
       \sum_{j=1}^{j_0} \pi(w_j)^2$.
   Plugging $\pi(w_j) ={ \kappa j^{-\beta} \over n^{1-\beta}}$ into
   ${n \over \e}\sum_{j=1}^{c_1\over \e} \pi(w_j)$, and we
   have
\vspace{-1mm}\vspace{-1mm}\begin{align}& \quad \E[C_I] \le {n \over \e}\sum_{j=1}^{c_1\over \e} \pi(w_j)= \sum_{j=1}^{c_1\over \e}
                        {n \cdot j^{-\beta} \over n^{1-\beta} \e}
                      =
  {n^\beta \cdot \sum_{j=1}^{c_1\over \e}  j^{-\beta} \over \e}
                                                              \nonumber
  \\ &=O\left({n^{\beta}\over
                                                            \left({\e
                                                            \over c_1}\right)^{1-\beta} \cdot \e}\right)
=O\left({n^{\beta}\over\e^{1-\beta} \cdot \e}\right) =
    O\left({n^{\beta} \over \e^{2-\beta}}\right). \label{eqn:C_I1}
  \end{align}
Plugging $\pi(w_j) ={ \kappa j^{-\beta} \over n^{1-\beta}}$ into ${n  \over \e^2}
\sum_{j=1}^{j_0} \pi(w_j)^2 $ follows that
\vspace{-1mm}$$\E[C_I] \le {n  \over \e^2}
\sum_{j=1}^{j_0} \pi(w_j)^2  ={n  \over \e^2}
\sum_{j=1}^{j_0} {\kappa j^{-2\beta} \over n^{2-2\beta}} = {\kappa n^{2\beta-1}  \over \e^2}
\sum_{j=1}^{j_0} j^{-2\beta} .\vspace{-1mm}$$
For $\beta < 1/2$, we have
$\sum_{j=1}^{j_0} j^{-2\beta}  = O(j_0^{1-2\beta}) = O(n^{1-2\beta}),$
and thus $\E[C_I] = O\left( {n^{2\beta-1}  \over \e^2} \cdot n^{1-2\beta}
\right) = O\left( {1 \over \e^2}
\right)$. For $\beta = 1/2$, we have $\sum_{j=1}^{j_0} j^{-2\beta}  = O(\log
j_0)$. Since $\log j_0 \le \log n$ and $n^{2\beta -1} = 1$, we have $\E[C_I] = O\left( {n^{2\beta-1}  \over \e^2} \cdot \log j_0
\right) = O\left( {\log n \over \e^2}
\right)$. For $\beta > 1/2$, we have $\sum_{j=1}^{j_0} j^{-2\beta}  = O(1)$ and
consequently $\E[C_I] = O\left( {n^{2\beta-1}  \over \e^2}
\right)$. Combining Equation~\eqref{eqn:C_I1} and above analysis, we have the following equation:



   \vspace{-1mm}\vspace{-1mm}\begin{equation}
\label{eqn:query_index_free}
\E[C_I] = \left\{
\vspace{-1mm}\begin{array}{ll}
 O({1 \over \e^2}), &\textrm{for } \beta < 1/2; \\
  O({\log n \over \e^2}) , & \textrm{for }\beta = 1/2; \\
O\left(\min \left\{{n^{2\beta-1} \over \e^2}, {n^{\beta} \over \e^{2-\beta}} \right\}\right) , & \textrm{for }\beta > 1/2.
\end{array}\right.
\end{equation}

By Lemma~\ref{lem:query_B} and the assumption $\pi(w_j) ={ \kappa j^{-\beta} \over n^{1-\beta}}$,  we have
$\E[C_B]= O\left( {c_1n^{2\beta-1} \log n  \over \e^2}
 \sum_{j=j_0+1}^{n} j^{-2\beta} \right).$
For $j < 1/2$, we have $ \sum_{j=j_0+1}^{n} j^{-2\beta}  =
O(n^{1-2\beta})$. Thus $$\E[C_B] = O\left({n^{2\beta-1} \log n  \over
    \e^2} \cdot n^{1-2\beta}\right) = O\left( {\log n \over
    \e^2}\right).$$
For $j = 1/2$, we have $ \sum_{j=j_0+1}^{n} j^{-2\beta}  =
O(\log n)$, and thus $\E[C_B]= O\left( {\log n \log {n\over \delta} \over
    \e^2}\right)$. For $j > 1/2$, we have $ \sum_{j=j_0+1}^{n} j^{-2\beta}  =
O(j_o^{1-2\beta})$. Plugging $j_0 \le n\left(\e\d \right)^{1\over
  1-\beta}$ follows that
\vspace{-1mm}\vspace{-1mm}\begin{align*}
\E[C_B] &= O\left({n^{2\beta-1} \log n  \over
    \e^2} \cdot \left(n(\e \d)^{1\over 1-\beta} \right)^{1-2\beta}
      \right) \\
  & = O\left({ \log n  \over
    \e^2} \cdot (\e \d)^{1-2\beta\over 1-\beta}
      \right) = O\left({ \log n  \over
    \e^{1\over 1-\beta} \d^{2\beta -1 \over 1-\beta} }
      \right).
\end{align*}
By $\e \ge {\log^{1-\beta \over 2\beta - 1} n / n^{1-\beta}
   \d^{2\beta -1 \over \beta}}$ and $\delta > {1\over n^{\Omega(1)}}$,
 it follows that ${ \log {n\over \delta}  /
    \e^{1\over 1-\beta} \d^{2\beta -1\over 1-\beta} } \le
 {n^{2\beta -1} \over \e^2}$ and ${ \log n /
    \e^{1\over 1-\beta} \d^{2\beta - 1\over 1-\beta} }\le
 {n^\beta \over \e^{2-\beta}}$, and thus $\E[C_B] $ is bounded by $O\left( \min \left\{{ n^{2\beta-1} \over
     \e^2} , {n^{\beta} \over \e^{2-\beta}} \right\} \right) $ for
$\beta > 1/2$.
 In summary, we have
     \vspace{-1mm}\vspace{-1mm}\begin{equation}
\label{eqn:query_index_free}
\E[C_B] = \left\{
\vspace{-1mm}\begin{array}{ll}
O({\log {n\over \delta} \over \e^2}), &\textrm{for } \beta < 1/2; \\
  O({\log n \log {n\over \delta} \over \e^2}) , & \textrm{for }\beta = 1/2; \\
               O\left( \min \left\{{ n^{2\beta-1} \over
     \e^2} , {n^{\beta} \over \e^{2-\beta}} \right\} \right) , & \textrm{for }\beta > 1/2.
\end{array}\right.
\end{equation}
Combing $C_F$, $C_I$, $C_B$ and { $\beta = 1/\gamma$}, 
 the theorem follows.
 \end{proof}
